\documentclass[a4paper]{article}%
\usepackage{amsmath}
\usepackage{amsfonts}
\usepackage{amssymb}
\usepackage{graphicx}%
\providecommand{\U}[1]{\protect\rule{.1in}{.1in}}
\newtheorem{theorem}{Theorem}

\newtheorem{condition}[theorem]{Condition}

\newtheorem{corollary}[theorem]{Corollary}

\newtheorem{lemma}[theorem]{Lemma}

\newtheorem{remark}[theorem]{Remark}

\newenvironment{proof}[1][Proof]{\noindent\textbf{#1.} }{\ \rule{0.5em}{0.5em}}
\begin{document}

\title{\textbf{Empirical phi-divergence test statistics for testing simple null
hypotheses based on exponentially tilted empirical likelihood
estimators\thanks{This paper was supported by the Spanish Grants
MTM-2012-33740 and ECO-2011-25706.}}}
\author{A. Felipe$^{1}$, N. Mart\'{\i}n$^{2}$\thanks{Corresponding author.} , P.
Miranda$^{1}$ and L. Pardo$^{1}$
\and $^{1}${\small Department of Statistics and O.R., Complutense University of
Madrid, Madrid, Spain}
\and $^{2}${\small Department of Statistics and Flores de Lemus Institute, Carlos
III University of Madrid, Getafe (Madrid), Spain}}
\date{October 28, 2015}
\maketitle

\begin{abstract}
In Econometrics, imposing restrictions without assuming underlying
distributions to modelize complex realities is a valuable methodological tool.
However, if a subset of restrictions were not correctly specified, the usual
test-statistics for correctly specified models tend to reject erronously a
simple null hypothesis. In this setting, we may say that the model suffers
from misspecification. We study the behavior of empirical phi-divergence
test-statistics, introduced in \emph{Balakrishnan et al. (2015)}, by using the
exponential tilted empirical likelihood estimators of \emph{Schennach (2007)},
as a good compromise between the efficiency of the significance level for
small sample sizes and the robustness under misspecification.

\end{abstract}

\bigskip\bigskip



\section{Introduction}

Let $\mathbf{X}_{1},...,\mathbf{X}_{n}$ be i.i.d. observations on a data
vector $\mathbf{X}$ with unknown distribution function $F$ having a finite
expectation, a non-singular variance-covariance matrix and a $p$-dimensional
parameter of interest, $\boldsymbol{\theta}\in\Theta\subset%
\mathbb{R}
^{p}$. All the information about $F$ and $\boldsymbol{\theta}$ is available in
the form of $r\geq p$ estimating functions of the data observation
$\mathbf{X}$ and the parameter $\boldsymbol{\theta}$
\begin{equation}
\boldsymbol{g}(\boldsymbol{X},\boldsymbol{\theta})=\left(  g_{1}%
(\boldsymbol{X},\boldsymbol{\theta}),...,g_{r}(\boldsymbol{X}%
,\boldsymbol{\theta}\right)  )^{T}. \label{T1}%
\end{equation}
The model has a true parameter $\boldsymbol{\theta}_{0}$ satisfying the moment
condition
\begin{equation}
\mathrm{E}_{F}\left[  \boldsymbol{g}(\boldsymbol{X},\boldsymbol{\theta}%
_{0})\right]  =\boldsymbol{0}_{r}, \label{T1b}%
\end{equation}
where $\mathrm{E}_{F}\left[  \cdot\right]  $ denotes expectation taken with
respect to the distribution of $F$ of $\mathbf{X}$\textbf{.} The parameter
$\boldsymbol{\theta}$ has been traditionally estimated using two-step
efficient generalized method of moments estimators (GMM). This method of
estimation was introduced by Hansen (1982). In Hayashi (2000), for instance,
all the estimation techniques are presented and discussed in the GMM
framework. A GMM estimator for $\boldsymbol{\theta}_{0}$ is
$\widehat{\boldsymbol{\theta}}_{GMM}$, defined by
\[
\widehat{\boldsymbol{\theta}}_{GMM}=\arg\min_{\boldsymbol{\theta}\in\Theta
}\overline{\boldsymbol{g}}_{n}^{T}(\boldsymbol{X},\boldsymbol{\theta
})\boldsymbol{W}_{n}^{-1}(\boldsymbol{\theta})\overline{\boldsymbol{g}}%
_{n}(\boldsymbol{X},\boldsymbol{\theta}),
\]
where%
\begin{equation}
\overline{\boldsymbol{g}}_{n}(\boldsymbol{X},\boldsymbol{\theta})=\frac{1}{n}%
{\displaystyle\sum\limits_{i=1}^{n}}
\boldsymbol{g}(\boldsymbol{X}_{i},\boldsymbol{\theta}) \label{gBar}%
\end{equation}
and $\boldsymbol{W}_{n}$ is a positive semidefinite matrix. Under some
regularity conditions $\widehat{\boldsymbol{\theta}}_{GMM}$ is consistent for
$\boldsymbol{\theta}_{0}$ but in general it is not efficient if $r>p.$ The
$\widehat{\boldsymbol{\theta}}_{GMM}$ will be asymptotically efficient if the
limit of the matrix $\boldsymbol{W}_{n}$ is the matrix
\begin{equation}
\mathbf{S}_{11}(\boldsymbol{\theta}_{0})=\mathrm{E}_{F}\left[  \boldsymbol{g}%
(\boldsymbol{X},\boldsymbol{\theta}_{0})\boldsymbol{g}^{T}(\boldsymbol{X}%
,\boldsymbol{\theta}_{0})\right]  . \label{S11}%
\end{equation}
A feasible version of this efficient procedure is based on obtaining an
initial consistent estimator $\widehat{\boldsymbol{\theta}}$ of
$\boldsymbol{\theta}_{0}$ by,
\[
\widehat{\boldsymbol{\theta}}=\arg\min_{\boldsymbol{\theta}\in\Theta}%
\overline{\boldsymbol{g}}_{n}^{T}(\boldsymbol{X},\boldsymbol{\theta}%
)\overline{\boldsymbol{g}}_{n}(\boldsymbol{X},\boldsymbol{\theta})
\]
and then to consider
\[
\widehat{\boldsymbol{\theta}}_{GMM}=\arg\min_{\boldsymbol{\theta}\in\Theta
}\overline{\boldsymbol{g}}_{n}^{T}(\boldsymbol{X},\boldsymbol{\theta
})\widehat{\mathbf{S}}_{11}^{-1}(\widehat{\boldsymbol{\theta}})\overline
{\boldsymbol{g}}_{n}(\boldsymbol{X},\boldsymbol{\theta}),
\]
with
\begin{equation}
\widehat{\mathbf{S}}_{11}(\boldsymbol{\theta})=\frac{1}{n}%
{\displaystyle\sum\limits_{i=1}^{n}}
\boldsymbol{g}(\boldsymbol{X}_{i},\boldsymbol{\theta})\boldsymbol{g}%
^{T}(\boldsymbol{X}_{i},\boldsymbol{\theta}). \label{S110}%
\end{equation}
An alternative to the GMM estimator is the (CU) continuous updating estimator
obtained by%
\[
\widehat{\boldsymbol{\theta}}_{CU}=\arg\min_{\boldsymbol{\theta}\in\Theta
}\overline{\boldsymbol{g}}_{n}^{T}(\boldsymbol{X},\boldsymbol{\theta
})\widehat{\mathbf{S}}_{11}^{-1}(\boldsymbol{\theta})\overline{\boldsymbol{g}%
}_{n}(\boldsymbol{X},\boldsymbol{\theta}).
\]

The GMM estimators have nice asymptotic properties (see Gallant and White
(1988), Newey and McFadden (1990)). They are consistent, asymptotically normal
and asymptotically efficient under some regularity assumptions. However,
several authors report that the two-step GMM estimator suffers from a
substantial amount of bias in finite samples (see Altonji and Segal (1996),
Andersen and S\o rensen (1996) and Hansen, Heaton and Yaron (1996)). This
encourages the increasing literature on alternatives to the GMM. Maybe the
most known alternative estimators to the GMM are: the continuously updated
(CU) estimator of Hansen, Heaton and Yaron (1996), the empirical likelihood
estimator (EL) of Owen (1988, 1990), Qin and Lawless (1994), and Imbens
(1997), the exponential tilting (ET) estimator of Kitamura and Stutzer (1997)
and Imbens, Spady and Johnson (1998), the minimum Hellinger distance estimator
of Kitamura, Otsa and Evdokimov (2013) and the generalized empirical
likelihood (GEL) estimators of Newey and Smith (2004). Although EL estimator
is preferable to the previous estimators in higher-order asymptotic
properties, these properties hold only under correct specification of the
moment condition, and the asymptotic behavior of EL estimator becomes
problematic under misspecification. The ET estimator is inferior to the EL
estimator in relation to higher-order asymptotic properties, but remain well
behaved in presence of misspecification under relative weak regularity
conditions. To overcome this problem, Schennach (2007) suggests the
exponentially tilted empirical likelihood (ETEL) that shares the same
higher-order property with EL under correct specification while maintaining
usual asymptotic properties such as $\sqrt{n}$-consistency and asymptotic
normality under misspecification.

Qin and Lawless (1994) studied the empirical likelihood ratio statistic for
testing simple null hypotheses based on the EL estimators. Later Balakrishnan
et al. (2015), using EL, considered some families of test statistics based on
$\phi$-divergence measures: empirical $\phi$-divergence test statistics, which
contain the empirical likelihood ratio test as a particular case. Some members
of this family have a better behavior for small sample sizes in the sense of
the size and power of the test. The contribution of the current paper is to
extend the empirical $\phi$-divergence test statistics replacing the EL
estimators by the ET and ETEL estimators to study their robustness, in
particular under misspecification, which is their major advantage with respect
to the previous ones.

In Section \ref{Sec2} we introduce the ETEL estimator given by Schennach
(2007) which is obtained as a combination of EL and ET procedures to deliver
an estimator and we present its asymptotic properties. Section \ref{sec2a} is
devoted to introduce the empirical $\phi$-divergence statistics for testing
simple null hypotheses on the basis of the ETEL estimator and we present their
asymptotic distribution. Based on it, power approximations of the empirical
$\phi$-divergence test statistics are derived. A rigorous study of the
robustness of the empirical $\phi$-divergence test statistics is derived in
Section \ref{robustness} and the asymptotic distribution of the empirical
$\phi$-divergence is developed under misspecified alternative hypotheses. In
Section \ref{Simulation} a simulation study is presented and finally, in
Section \ref{secComp} some conclusions are given.

\section{Exponentially tilted empirical likelihood\label{Sec2}}

Let $\boldsymbol{x}_{1},...,\boldsymbol{x}_{n}$\ be a realization of
$\boldsymbol{X}_{1},...,\boldsymbol{X}_{n}$. The empirical likelihood function
is given by
\[
\mathcal{L}_{F_{n}}(\boldsymbol{x}_{1},...,\boldsymbol{x}_{n})=\prod_{i=1}%
^{n}dF\left(  \boldsymbol{x}_{i}\right)  =\prod_{i=1}^{n}p_{i},
\]
where $p_{i}=dF\left(  \boldsymbol{x}_{i}\right)  =P(\boldsymbol{X}%
=\boldsymbol{x}_{i})$. Only distributions with an atom of probability at each
$\boldsymbol{x}_{i}$ have non-zero likelihood, and without consideration of
estimating functions,\ the empirical likelihood function $\mathcal{L}_{F_{n}}$
is seen to be maximized, at $\boldsymbol{X}_{1}=\boldsymbol{x}_{1}%
,...,\boldsymbol{X}_{n}=\boldsymbol{x}_{n}$, by the empirical distribution
function%
\[
F_{n}\left(  \boldsymbol{x}\right)  =\sum\limits_{i=1}^{n}u_{i}%
I(\boldsymbol{X_{i}}\leq\boldsymbol{x}),
\]
which is associated with the $n$-dimensional discrete uniform distribution%
\[
\boldsymbol{u}=(u_{1},...,u_{n})^{T}=(\tfrac{1}{n},\overset{\overset{n}{\smile
}}{\cdots},\tfrac{1}{n})^{T}.
\]
Let%
\[
F_{n,\boldsymbol{\theta}}\left(  \boldsymbol{x}\right)  =\sum\limits_{i=1}%
^{n}p_{i}\left(  \boldsymbol{\theta}\right)  I(\boldsymbol{X_{i}\leq
x})\text{,}%
\]
be an empirical distribution function associated with the probability vector%
\begin{equation}
\boldsymbol{p}(\boldsymbol{\theta})=(p_{1}(\boldsymbol{\theta}),...,p_{n}%
(\boldsymbol{\theta}))^{T},\quad p_{i}\left(  \boldsymbol{\theta}\right)
>0,\quad\sum\limits_{i=1}^{n}p_{i}\left(  \boldsymbol{\theta}\right)  =1,
\label{cond}%
\end{equation}
and%
\begin{equation}
\ell_{EL}(\boldsymbol{\theta})=\sum_{i=1}^{n}\log p_{i}\left(
\boldsymbol{\theta}\right)  \label{r1}%
\end{equation}
the kernel of the empirical log-likelihood function. The moment conditions
given in (\ref{T1b}) can be expressed from an empirical point of view as%
\begin{equation}
\mathrm{E}_{F_{n,\boldsymbol{\theta}}}\left[  \boldsymbol{g}(\boldsymbol{X}%
,\boldsymbol{\theta})\right]  =\sum\limits_{i=1}^{n}p_{i}\left(
\boldsymbol{\theta}\right)  \boldsymbol{g}(\boldsymbol{X}_{i}%
,\boldsymbol{\theta})=\boldsymbol{0}_{r}, \label{R1}%
\end{equation}
which are the so-called estimating equations. If we are interested in
maximizing (\ref{r1}) subject to (\ref{R1}), by applying the Lagrange
multipliers method it is possible to reduce the dimension of the probability
vector ($n$), to the number of estimating functions ($r$), since%
\begin{equation}
p_{EL,i}\left(  \boldsymbol{\theta}\right)  =\frac{1}{n}\frac{1}%
{1+\boldsymbol{t}_{EL}^{T}(\boldsymbol{\theta})\boldsymbol{g}(\boldsymbol{X}%
_{i},\boldsymbol{\theta})},\text{ }i=1,...,n, \label{empF2}%
\end{equation}
where $\boldsymbol{t}_{EL}(\boldsymbol{\theta})$ is an $r$-dimensional vector
to be determined by solving the non-linear system of $r$ equations,%
\begin{align}
\frac{1}{n}\sum\limits_{i=1}^{n}\frac{1}{1+\boldsymbol{t}_{EL}^{T}%
(\boldsymbol{\theta})\boldsymbol{g}(\boldsymbol{X}_{i},\boldsymbol{\theta}%
)}\boldsymbol{g}(\boldsymbol{X}_{i},\boldsymbol{\theta})  &  =\boldsymbol{0}%
_{r},\label{ec}\\
\text{s.t. }\boldsymbol{t}_{EL}^{T}(\boldsymbol{\theta})\boldsymbol{g}%
(\boldsymbol{X}_{i},\boldsymbol{\theta})  &  >\frac{1-n}{n}.\nonumber
\end{align}

Maximizing expression (\ref{r1}) is equivalent to minimize the expression
\[
-\frac{1}{n}%
{\displaystyle\sum\limits_{i=1}^{n}}
\log\left(  np_{i}\left(  \boldsymbol{\theta}\right)  \right)
\]
and this expression can be written as the Kullback--Leibler divergence measure
between the probability vectors $\boldsymbol{u}$ and $\boldsymbol{p}\left(
\boldsymbol{\theta}\right)  $, i.e.,%
\[
D_{\mathrm{Kull}}\left(  \boldsymbol{u},\boldsymbol{p}\left(
\boldsymbol{\theta}\right)  \right)  =%
{\displaystyle\sum\limits_{i=1}^{n}}
u_{i}\log\frac{u_{i}}{p_{i}\left(  \boldsymbol{\theta}\right)  }.
\]
Therefore,%
\[
\widehat{\boldsymbol{\theta}}_{EL}=\arg\min_{\boldsymbol{\theta}\in\Theta
}D_{\mathrm{Kull}}\left(  \boldsymbol{u},\boldsymbol{p}_{EL}\left(
\boldsymbol{\theta}\right)  \right)
\]
subject to the restrictions given in (\ref{R1}).

If we consider $D_{\mathrm{Kull}}\left(  \boldsymbol{p}\left(
\boldsymbol{\theta}\right)  ,\boldsymbol{u}\right)  $, rather than
$D_{\mathrm{Kull}}\left(  \boldsymbol{u},\boldsymbol{p}\left(
\boldsymbol{\theta}\right)  \right)  $, we get the empirical exponential
tilting (ET) estimator, considered for instance in Kitamura and Stutzer
(1997). In that case
\[
\widehat{\boldsymbol{\theta}}_{ET}=\arg\min_{\boldsymbol{\theta}\in\Theta
}D_{\mathrm{Kull}}\left(  \boldsymbol{p}\left(  \boldsymbol{\theta}\right)
,\boldsymbol{u}\right)  .
\]
where%
\begin{equation}
D_{\mathrm{Kull}}\left(  \boldsymbol{p}\left(  \boldsymbol{\theta}\right)
,\boldsymbol{u}\right)  =%
{\displaystyle\sum\limits_{i=1}^{n}}
p_{i}\left(  \boldsymbol{\theta}\right)  \log\left(  np_{i}\left(
\boldsymbol{\theta}\right)  \right)  \label{dET}%
\end{equation}
and
\begin{equation}
p_{ET,i}\left(  \boldsymbol{\theta}\right)  =\frac{\exp\{\boldsymbol{t}%
_{ET}^{T}(\boldsymbol{\theta})\boldsymbol{g}(\boldsymbol{X}_{i}%
,\boldsymbol{\theta})\}}{%
{\displaystyle\sum\limits_{j=1}^{n}}
\exp\{\boldsymbol{t}_{ET}^{T}(\boldsymbol{\theta})\boldsymbol{g}%
(\boldsymbol{X}_{j},\boldsymbol{\theta})\}},\text{ }i=1,...,n, \label{pET}%
\end{equation}
where $\boldsymbol{t}_{ET}(\boldsymbol{\theta})$ is an $r$-dimensional vector
to be determined by solving the non-linear system of $r$ equations
\begin{equation}
\frac{1}{n}%
{\displaystyle\sum\limits_{i=1}^{n}}
\exp\{\boldsymbol{t}_{ET}^{T}(\boldsymbol{\theta})\boldsymbol{g}%
(\boldsymbol{X}_{i},\boldsymbol{\theta})\}\boldsymbol{g}(\boldsymbol{X}%
_{i},\boldsymbol{\theta})=\boldsymbol{0}_{r}. \label{eET}%
\end{equation}
The exponentially tilted empirical likelihood (ETEL) introduced by Schennach
(2007) combines EL and ET procedures to deliver an estimator. The
ETEL\ estimator is defined as
\begin{equation}
\widehat{\boldsymbol{\theta}}_{ETEL}=\arg\min_{\boldsymbol{\theta}\in\Theta
}D_{\mathrm{Kull}}\left(  \boldsymbol{u},\boldsymbol{p}_{ET}\left(
\boldsymbol{\theta}\right)  \right)  , \label{ETEL}%
\end{equation}
where%
\begin{equation}
\boldsymbol{p}_{ET}\left(  \boldsymbol{\theta}\right)  =(p_{ET,1}\left(
\boldsymbol{\theta}\right)  ,...,p_{ET,n}\left(  \boldsymbol{\theta}\right)
)^{T}, \label{pppET}%
\end{equation}
and $p_{ET,i}\left(  \boldsymbol{\theta}\right)  $ is given by (\ref{pET}).
Theorem 1 in Schennach establishes that the ETEL estimator of
$\boldsymbol{\theta}$ maximizes the kernel of the empirical log-likelihood
function given by%
\begin{equation}
\ell_{ETEL}(\boldsymbol{\theta})=%
{\displaystyle\sum\limits_{i=1}^{n}}
\log p_{ET,i}\left(  \boldsymbol{\theta}\right)  =-\log\left(  \frac{1}{n}%
{\displaystyle\sum\limits_{i=1}^{n}}
\exp\left\{  \boldsymbol{t}_{ET}^{T}(\boldsymbol{\theta})\left[
\boldsymbol{g}(\boldsymbol{X}_{i},\boldsymbol{\theta})-\overline
{\boldsymbol{g}}_{n}(\boldsymbol{X},\boldsymbol{\theta})\right]  \right\}
\right)  , \label{T10}%
\end{equation}
where $\boldsymbol{t}_{ET}(\boldsymbol{\theta})$ is obtained by solving
(\ref{eET}) and $\overline{\boldsymbol{g}}_{n}(\boldsymbol{X}%
,\boldsymbol{\theta})$ was defined in (\ref{gBar}). In Schennach (2007, page
659) the following important relation for this paper is presented,%
\begin{equation}
\left(
\begin{array}
[c]{c}%
\overline{\boldsymbol{g}}_{n}(\boldsymbol{X},\boldsymbol{\theta}_{0})\\
\boldsymbol{0}_{p}%
\end{array}
\right)  +\left(
\begin{array}
[c]{cc}%
\boldsymbol{S}_{11}\left(  \boldsymbol{\theta}_{0}\right)  & \boldsymbol{S}%
_{12}\left(  \boldsymbol{\theta}_{0}\right) \\
\boldsymbol{S}_{12}^{T}\left(  \boldsymbol{\theta}_{0}\right)  &
\boldsymbol{0}_{p\times p}%
\end{array}
\right)  \left(
\begin{array}
[c]{c}%
\boldsymbol{t}_{ET}(\widehat{\boldsymbol{\theta}}_{ETEL})\\
\widehat{\boldsymbol{\theta}}_{ETEL}-\boldsymbol{\theta}_{0}%
\end{array}
\right)  =o_{p}(n^{-1/2}), \label{for}%
\end{equation}
with $\boldsymbol{S}_{11}\left(  \boldsymbol{\theta}_{0}\right)  $\ given in
(\ref{S11}), and%
\begin{align}
\boldsymbol{S}_{12}\left(  \boldsymbol{\theta}\right)   &  =\mathrm{E}%
_{F}\left[  \boldsymbol{G}_{\boldsymbol{X}}(\boldsymbol{\theta})\right]
,\label{S12}\\
\boldsymbol{G}_{\boldsymbol{X}}(\boldsymbol{\theta})  &  =\frac{\partial
}{\partial\boldsymbol{\theta}^{T}}\boldsymbol{g}(\boldsymbol{X}%
,\boldsymbol{\theta}). \label{G}%
\end{align}
Based on (\ref{for}), we have%
\[
\widehat{\boldsymbol{\theta}}_{ETEL}-\boldsymbol{\theta}_{0}=\boldsymbol{V}%
\left(  \boldsymbol{\theta}_{0}\right)  \boldsymbol{S}_{12}^{T}\left(
\boldsymbol{\theta}_{0}\right)  \boldsymbol{S}_{11}^{-1}\left(
\boldsymbol{\theta}_{0}\right)  \overline{\boldsymbol{g}}_{n}(\boldsymbol{X}%
,\boldsymbol{\theta}_{0})+o_{p}(n^{-1/2}),
\]
where%
\begin{equation}
\boldsymbol{V}\left(  \boldsymbol{\theta}_{0}\right)  =\left(  \boldsymbol{S}%
_{12}^{T}\left(  \boldsymbol{\theta}_{0}\right)  \boldsymbol{S}_{11}%
^{-1}\left(  \boldsymbol{\theta}_{0}\right)  \boldsymbol{S}_{12}\left(
\boldsymbol{\theta}_{0}\right)  \right)  ^{-1}, \label{V}%
\end{equation}
and
\begin{equation}
\boldsymbol{t}_{ET}(\widehat{\boldsymbol{\theta}}_{ETEL})=-\boldsymbol{R}%
\left(  \boldsymbol{\theta}_{0}\right)  \overline{\boldsymbol{g}}%
_{n}(\boldsymbol{X},\boldsymbol{\theta}_{0})+o_{p}(n^{-1/2}). \label{T12}%
\end{equation}
where
\[
\boldsymbol{R}\left(  \boldsymbol{\theta}_{0}\right)  =\boldsymbol{S}%
_{11}^{-1}\left(  \boldsymbol{\theta}_{0}\right)  -\boldsymbol{S}_{11}%
^{-1}\left(  \boldsymbol{\theta}_{0}\right)  \boldsymbol{S}_{12}\left(
\boldsymbol{\theta}_{0}\right)  \boldsymbol{V}\left(  \boldsymbol{\theta}%
_{0}\right)  \boldsymbol{S}_{12}^{T}\left(  \boldsymbol{\theta}_{0}\right)
\boldsymbol{S}_{11}^{-1}\left(  \boldsymbol{\theta}_{0}\right)  .
\]
Expression (\ref{T12}) is obtained from (\ref{for}). Hence,%
\[
\sqrt{n}(\widehat{\boldsymbol{\theta}}_{ETEL}-\boldsymbol{\theta}%
_{0})\overset{\mathcal{L}}{\underset{n\rightarrow\infty}{\longrightarrow}%
}\mathcal{N}\left(  \boldsymbol{0}_{p},\boldsymbol{V}\left(
\boldsymbol{\theta}_{0}\right)  \right)  ,
\]
and
\[
\sqrt{n}\boldsymbol{t}_{ET}(\widehat{\boldsymbol{\theta}}_{ETEL}%
)\overset{\mathcal{L}}{\underset{n\rightarrow\infty}{\longrightarrow}%
}\mathcal{N}\left(  \boldsymbol{0}_{r},\boldsymbol{R}\left(
\boldsymbol{\theta}_{0}\right)  \right)  .
\]

In the following section we propose a new family of empirical test statistics
for testing a simple null hypothesis, when the unknown parameters are
estimated using the ETEL\ estimator defined in (\ref{ETEL}) and then derive
their asymptotic distribution.

\section{New family of empirical phi-divergence test statistics\label{sec2a}}

The empirical likelihood ratio statistic for testing%
\begin{equation}
H_{0}\text{: }\boldsymbol{\theta}=\boldsymbol{\theta}_{0}\text{ vs. }%
H_{1}\text{: }\boldsymbol{\theta\neq\theta}_{0} \label{H}%
\end{equation}
based on the ETEL estimator has the expression
\begin{align}
G_{n}^{2}(\widehat{\boldsymbol{\theta}}_{ETEL},\boldsymbol{\theta}_{0})  &
=2\sum\limits_{i=1}^{n}\log p_{ET,i}(\widehat{\boldsymbol{\theta}}%
_{ETEL})-2\sum\limits_{i=1}^{n}\log p_{ET,i}(\boldsymbol{\theta}%
_{0})\label{E9}\\
&  =-2n\left(  \ell_{ETEL}(\boldsymbol{\theta}_{0})-\ell_{ETEL}%
(\widehat{\boldsymbol{\theta}}_{ETEL})\right)  ,\nonumber
\end{align}
where $\ell_{ETEL}(\boldsymbol{\bullet})$ was defined in (\ref{T10}).
Schennach (2007) established that under $H_{0}$
\[
G_{n}^{2}(\widehat{\boldsymbol{\theta}}_{ETEL},\boldsymbol{\theta}%
_{0})\underset{n\rightarrow\infty}{\overset{\mathcal{L}}{\rightarrow}}\chi
_{p}^{2}.
\]
It is clear that the empirical likelihood ratio test statistic given in
(\ref{E9}) can be expressed as
\[
G_{n}^{2}(\widehat{\boldsymbol{\theta}}_{ETEL},\boldsymbol{\theta}%
_{0})=2n\left(  D_{\mathrm{Kull}}\left(  \boldsymbol{u},\boldsymbol{p}%
_{ET}\left(  \boldsymbol{\theta}_{0}\right)  \right)  -D_{\mathrm{Kull}%
}\left(  \boldsymbol{u},\boldsymbol{p}_{ET}(\widehat{\boldsymbol{\theta}%
}_{ETEL})\right)  \right)  ,
\]
where $\boldsymbol{p}_{ET}\left(  \boldsymbol{\theta}\right)  $ is
(\ref{pppET}).

We shall denote by $\Phi^{\ast}$ the class of all convex functions
$\phi:\mathbb{R}^{+}\longrightarrow\mathbb{R}$ such that at $x=1$,
$\phi\left(  1\right)  =0$, $\phi^{\prime\prime}\left(  1\right)  >0$, and at
$x=0$, $0\phi\left(  0/0\right)  =0$ and $0\phi\left(  p/0\right)
=p\lim_{u\rightarrow\infty}\frac{\phi\left(  u\right)  }{u}$. If instead of
considering the Kullback--Leibler divergence measure,\ we consider a general
function $\phi\in\Phi^{\ast}$ to define the $\phi$-divergence measure between
the probability vectors $\boldsymbol{u}$ and $\boldsymbol{p}\left(
\boldsymbol{\theta}\right)  $ as
\begin{equation}
D_{\phi}\left(  \boldsymbol{u},\boldsymbol{p}\left(  \boldsymbol{\theta
}\right)  \right)  =%
{\displaystyle\sum\limits_{i=1}^{n}}
p_{i}\left(  \boldsymbol{\theta}\right)  \phi\left(  \frac{u_{i}}{p_{i}\left(
\boldsymbol{\theta}\right)  }\right)  ,\quad\phi\in\Phi^{\ast}, \label{div}%
\end{equation}
we obtain a new family of empirical test statistics for testing (\ref{H})
given by
\begin{equation}
T_{n}^{\phi}(\widehat{\boldsymbol{\theta}}_{ETEL},\boldsymbol{\theta}%
_{0})=\frac{2n}{\phi^{\prime\prime}(1)}\left(  D_{\phi}\left(  \boldsymbol{u}%
,\boldsymbol{p}_{ET}\left(  \boldsymbol{\theta}_{0}\right)  \right)  -D_{\phi
}\left(  \boldsymbol{u},\boldsymbol{p}_{ET}(\widehat{\boldsymbol{\theta}%
}_{ETEL})\right)  \right)  , \label{F1}%
\end{equation}
i.e.,%
\[
T_{n}^{\phi}(\widehat{\boldsymbol{\theta}}_{ETEL},\boldsymbol{\theta}%
_{0})=\frac{2}{\phi^{\prime\prime}(1)}\left(
{\displaystyle\sum\limits_{i=1}^{n}}
np_{ET,i}\left(  \boldsymbol{\theta}_{0}\right)  \phi\left(  \frac
{1}{np_{ET,i}\left(  \boldsymbol{\theta}_{0}\right)  }\right)  -%
{\displaystyle\sum\limits_{i=1}^{n}}
np_{ET,i}(\widehat{\boldsymbol{\theta}}_{ETEL})\phi\left(  \frac{1}%
{np_{ET,i}(\widehat{\boldsymbol{\theta}}_{ETEL})}\right)  \right)  .
\]
Moreover, the empirical likelihood ratio test statistic falls inside this new
family since $G_{n}^{2}(\widehat{\boldsymbol{\theta}}_{ETEL,n}%
,\boldsymbol{\theta}_{0})=T_{n}^{\phi}(\widehat{\boldsymbol{\theta}}%
_{ETEL},\boldsymbol{\theta}_{0})$, with $\phi\left(  x\right)  =x\log x-x+1$.

It is well-known that the family of test statistics based on $\phi$-divergence
measures has some nice and optimal properties for different inferential
problems in relation to efficiency, but especially in relation to robustness;
see Pardo (2006) and Basu et al. (2011).

For every $\phi\in\Phi^{\ast}$ differentiable at $x=1$, the function
$\varphi\left(  x\right)  \equiv\phi(x)-\left(  x-1\right)  \phi^{\prime
}\left(  1\right)  $ also belongs to $\Phi^{\ast}$. Then, we have
\[
T_{n}^{\phi}(\widehat{\boldsymbol{\theta}}_{ETEL},\boldsymbol{\theta}%
_{0})=T_{n}^{\varphi}(\widehat{\boldsymbol{\theta}}_{ETEL},\boldsymbol{\theta
}_{0})
\]
and $\varphi$ has the additional property that $\varphi^{\prime}\left(
1\right)  =0.$ Since the two divergence measures are equivalent, without any
loss of generality we can consider the set $\Phi=\Phi^{\ast}\cap\left\{
\phi:\phi^{\prime}\left(  1\right)  =0\right\}  $. In what follows, we shall
assume that $\phi\in\Phi$.

Another family of statistics for testing the hypotheses in (\ref{H}) based
only on the $\phi$-divergence measure between $\boldsymbol{p}_{ET}%
(\widehat{\boldsymbol{\theta}}_{ETEL})$ and $\boldsymbol{p}_{ET}\left(
\boldsymbol{\theta}_{0}\right)  $, namely, $D_{\phi}\left(  \boldsymbol{p}%
_{ET}(\widehat{\boldsymbol{\theta}}_{ETEL}),\boldsymbol{p}_{ET}\left(
\boldsymbol{\theta}_{0}\right)  \right)  $, is given by
\begin{align}
S_{n}^{\phi}(\widehat{\boldsymbol{\theta}}_{ETEL},\boldsymbol{\theta}_{0})  &
=\frac{2n}{\phi^{\prime\prime}(1)}D_{\phi}\left(  \boldsymbol{p}%
_{ET}(\widehat{\boldsymbol{\theta}}_{ETEL}),\boldsymbol{p}_{ET}\left(
\boldsymbol{\theta}_{0}\right)  \right) \label{F2}\\
&  =\frac{2n}{\phi^{\prime\prime}(1)}%
{\displaystyle\sum\limits_{i=1}^{n}}
p_{ET,i}\left(  \boldsymbol{\theta}_{0}\right)  \phi\left(  \frac
{p_{ET,i}(\widehat{\boldsymbol{\theta}}_{ETEL})}{p_{ET,i}\left(
\boldsymbol{\theta}_{0}\right)  }\right)  ,\nonumber
\end{align}
where $\phi$ is a function satisfying the same conditions as function $\phi$
used to construct $T_{n}^{\phi}(\widehat{\boldsymbol{\theta}}_{ETEL}%
,\boldsymbol{\theta}_{0})$.

We shall refer to both families of test statistics as empirical $\phi
$-divergence test statistics. The first family has been applied for the first
time in Broniatowski and Keziou (2012) but using the EL estimator rather than
the ETEL estimator and only in the case that the parameter dimension is equal
to the number of estimating equations ($p=r$). Both families were applied in
Balakrishnan et al. (2015) only with the EL estimator.

\begin{condition}
\label{RC}Let $\left\Vert \cdot\right\Vert $ denote any vector or matrix norm.
We shall assume the following regularity conditions (Theorem 1 in Qin and
Lawless, 1994):\newline i) $\boldsymbol{S}_{11}\left(  \boldsymbol{\theta}%
_{0}\right)  $ in (\ref{S11}) is positive definite, and for $\boldsymbol{S}%
_{12}\left(  \boldsymbol{\theta}_{0}\right)  $ in (\ref{S12}), $\mathrm{rank}%
(\boldsymbol{S}_{12}\left(  \boldsymbol{\theta}_{0}\right)  )=p$;\newline ii)
There exists a neighborhood of $\boldsymbol{\theta}_{0}$ in which $\left\Vert
\boldsymbol{g}\left(  \boldsymbol{X},\boldsymbol{\theta}\right)  \right\Vert
^{3}$ is bounded by some integrable function of $\boldsymbol{X}$;\newline iii)
There exists a neighborhood of $\boldsymbol{\theta}_{0}$ in which
$\boldsymbol{G}_{\boldsymbol{X}}(\boldsymbol{\theta})$, given in (\ref{G}), is
continuous and $\left\Vert \boldsymbol{G}_{\boldsymbol{X}}(\boldsymbol{\theta
})\right\Vert $ is bounded by some integrable function of $\boldsymbol{X}%
$;\newline iv) There exists a neighborhood of $\boldsymbol{\theta}_{0}$ in
which $\frac{\partial\boldsymbol{G}_{\boldsymbol{X}}(\boldsymbol{\theta}%
)}{\partial\boldsymbol{\theta}}$ is continuous and $\left\Vert \frac
{\partial\boldsymbol{G}_{\boldsymbol{X}}(\boldsymbol{\theta})}{\partial
\boldsymbol{\theta}}\right\Vert $ is bounded by some integrable function of
$\boldsymbol{X}$.
\end{condition}

The asymptotic distribution of the empirical $\phi$-divergence test
statistics, $T_{n}^{\phi}(\widehat{\boldsymbol{\theta}}_{ETEL}%
,\boldsymbol{\theta}_{0})$ and $S_{n}^{\phi}(\widehat{\boldsymbol{\theta}%
}_{ETEL},\boldsymbol{\theta}_{0})$, is given in the following theorem.

\begin{theorem}
\label{Th1}Under Condition \ref{RC} and under the null hypothesis given in
(\ref{H}),%
\[
T_{n}^{\phi}(\widehat{\boldsymbol{\theta}}_{ETEL},\boldsymbol{\theta}%
_{0}),\quad S_{n}^{\phi}(\widehat{\boldsymbol{\theta}}_{ETEL}%
,\boldsymbol{\theta}_{0})\overset{\mathcal{L}}{\underset{n\rightarrow
\infty}{\longrightarrow}}\chi_{p}^{2}.
\]

\begin{proof}
We shall prove the result for $S_{n}^{\phi}(\widehat{\boldsymbol{\theta}%
}_{ETEL},\boldsymbol{\theta}_{0})$. In a similar way can be established the
result for $T_{n}^{\phi}(\widehat{\boldsymbol{\theta}}_{ETEL}%
,\boldsymbol{\theta}_{0})$.\newline Let us consider
\[
\widehat{\boldsymbol{t}}_{ETEL}=\boldsymbol{t}_{ET}%
(\widehat{\boldsymbol{\theta}}_{ETEL})\quad\text{and}\quad\boldsymbol{t}%
_{0}=\boldsymbol{t}(\boldsymbol{\theta}_{0}).
\]
We rename $D_{\phi}\left(  \boldsymbol{p}_{ET}(\widehat{\boldsymbol{\theta}%
}_{ETEL}),\boldsymbol{p}_{ET}\left(  \boldsymbol{\theta}_{0}\right)  \right)
=d_{\phi}(\widehat{\boldsymbol{t}}_{ETEL},\boldsymbol{t}_{0})$ as a function
of $\widehat{\boldsymbol{t}}_{ETEL}$ and $\boldsymbol{t}_{0}$, i.e.%
\[
d_{\phi}(\widehat{\boldsymbol{t}}_{ETEL},\boldsymbol{t}_{0})=\sum
\limits_{i=1}^{n}\frac{\exp\{\boldsymbol{t}_{0}^{T}\boldsymbol{g}%
(\boldsymbol{x}_{i},\boldsymbol{\theta}_{0})\}}{%
{\textstyle\sum\nolimits_{j=1}^{n}}
\exp\{\boldsymbol{t}_{0}^{T}\boldsymbol{g}(\boldsymbol{x}_{j}%
,\boldsymbol{\theta}_{0})\}}\phi\left(  \frac{\exp\{\widehat{\boldsymbol{t}%
}_{ETEL}^{T}\boldsymbol{g}(\boldsymbol{x}_{i},\widehat{\boldsymbol{\theta}%
}_{ETEL})\}}{%
{\textstyle\sum\nolimits_{j=1}^{n}}
\exp\{\widehat{\boldsymbol{t}}_{ETEL}^{T}\boldsymbol{g}(\boldsymbol{x}%
_{j},\widehat{\boldsymbol{\theta}}_{ETEL})\}}\left/  \frac{\exp
\{\boldsymbol{t}_{0}^{T}\boldsymbol{g}(\boldsymbol{x}_{i},\boldsymbol{\theta
}_{0})\}}{%
{\textstyle\sum\nolimits_{j=1}^{n}}
\exp\{\boldsymbol{t}_{0}^{T}\boldsymbol{g}(\boldsymbol{x}_{j}%
,\boldsymbol{\theta}_{0})\}}\right.  \right)  .
\]
A second-order Taylor expansion of $d_{\phi}(\widehat{\boldsymbol{t}}%
_{ETEL},\boldsymbol{t}_{0})$ around $\left(  \boldsymbol{0}_{r},\boldsymbol{0}%
_{r}\right)  $ gives%
\begin{align*}
d_{\phi}(\widehat{\boldsymbol{t}}_{ETEL},\boldsymbol{t}_{0})  &  =d_{\phi
}\left(  \boldsymbol{0}_{r},\boldsymbol{0}_{r}\right)  +\left.  \frac{\partial
d_{\phi}\left(  \boldsymbol{t}_{1},\boldsymbol{t}_{2}\right)  }{\partial
\boldsymbol{t}_{1}^{T}}\right\vert _{\boldsymbol{t}_{1}=\boldsymbol{t}%
_{2}=\boldsymbol{0}_{r}}\widehat{\boldsymbol{t}}_{ETEL}+\left.  \frac{\partial
d_{\phi}\left(  \boldsymbol{t}_{1},\boldsymbol{t}_{2}\right)  }{\partial
\boldsymbol{t}_{2}^{T}}\right\vert _{\boldsymbol{t}_{1}=\boldsymbol{t}%
_{2}=\boldsymbol{0}_{r}}\boldsymbol{t}_{0}\\
&  +\frac{1}{2}\widehat{\boldsymbol{t}}_{ETEL}^{T}\left.  \frac{\partial
^{2}d_{\phi}\left(  \boldsymbol{t}_{1},\boldsymbol{t}_{2}\right)  }%
{\partial\boldsymbol{t}_{1}\partial\boldsymbol{t}_{1}^{T}}\right\vert
_{\boldsymbol{t}_{1}=\boldsymbol{t}_{2}=\boldsymbol{0}_{r}}%
\widehat{\boldsymbol{t}}_{ETEL}+\frac{1}{2}\boldsymbol{t}_{0}^{T}\left.
\frac{\partial^{2}d_{\phi}\left(  \boldsymbol{t}_{1},\boldsymbol{t}%
_{2}\right)  }{\partial\boldsymbol{t}_{2}\partial\boldsymbol{t}_{2}^{T}%
}\right\vert _{\boldsymbol{t}_{1}=\boldsymbol{t}_{2}=\boldsymbol{0}_{r}%
}\boldsymbol{t}_{0}\\
&  +\widehat{\boldsymbol{t}}_{ETEL}^{T}\left.  \frac{\partial^{2}d_{\phi
}\left(  \boldsymbol{t}_{1},\boldsymbol{t}_{2}\right)  }{\partial
\boldsymbol{t}_{2}\partial\boldsymbol{t}_{1}^{T}}\right\vert _{\boldsymbol{t}%
_{1}=\boldsymbol{t}_{2}=\boldsymbol{0}_{r}}\boldsymbol{t}_{0}%
+o(||\widehat{\boldsymbol{t}}_{ETEL}||^{2})+o\left(  ||\boldsymbol{t}%
_{0}||^{2}\right)  .
\end{align*}
It is easy to show that%
\begin{align*}
&  d_{\phi}\left(  \boldsymbol{0}_{r},\boldsymbol{0}_{r}\right)
=0,\quad\left.  \dfrac{\partial d_{\phi}\left(  \boldsymbol{t}_{1}%
,\boldsymbol{t}_{2}\right)  }{\partial\boldsymbol{t}_{1}^{T}}\right\vert
_{\boldsymbol{t}_{1}=\boldsymbol{t}_{2}=\boldsymbol{0}_{r}}=\left.
\dfrac{\partial d_{\phi}\left(  \boldsymbol{t}_{1},\boldsymbol{t}_{2}\right)
}{\partial\boldsymbol{t}_{2}^{T}}\right\vert _{\boldsymbol{t}_{1}%
=\boldsymbol{t}_{2}=\boldsymbol{0}_{r}}=\boldsymbol{0}_{r}^{T},\\
&  \left.  \dfrac{\partial^{2}d_{\phi}\left(  \boldsymbol{t}_{1}%
,\boldsymbol{t}_{2}\right)  }{\partial\boldsymbol{t}_{1}\partial
\boldsymbol{t}_{1}^{T}}\right\vert _{\boldsymbol{t}_{1}=\boldsymbol{t}%
_{2}=\boldsymbol{0}_{r}}=\left.  \dfrac{\partial^{2}d_{\phi}\left(
\boldsymbol{t}_{1},\boldsymbol{t}_{2}\right)  }{\partial\boldsymbol{t}%
_{2}\partial\boldsymbol{t}_{2}^{T}}\right\vert _{\boldsymbol{t}_{1}%
=\boldsymbol{t}_{2}=\boldsymbol{0}_{r}}=\phi^{\prime\prime}\left(  1\right)
\widehat{\boldsymbol{S}}_{11}\left(  \boldsymbol{\theta}_{0}\right)
=\phi^{\prime\prime}\left(  1\right)  \boldsymbol{S}_{11}\left(
\boldsymbol{\theta}_{0}\right)  +o_{p}(\boldsymbol{1}_{r\times r}),\\
&  \left.  \dfrac{\partial^{2}d_{\phi}\left(  \boldsymbol{t}_{1}%
,\boldsymbol{t}_{2}\right)  }{\partial\boldsymbol{t}_{2}\partial
\boldsymbol{t}_{1}^{T}}\right\vert _{\boldsymbol{t}_{1}=\boldsymbol{t}%
_{2}=\boldsymbol{0}_{r}}=-\phi^{\prime\prime}\left(  1\right)
\widehat{\boldsymbol{S}}_{11}\left(  \boldsymbol{\theta}_{0}\right)
=-\phi^{\prime\prime}\left(  1\right)  \boldsymbol{S}_{11}\left(
\boldsymbol{\theta}_{0}\right)  +o_{p}(\boldsymbol{1}_{r\times r}).
\end{align*}
\newline Then, we have%
\begin{align*}
&  S_{n}^{\phi}(\widehat{\boldsymbol{\theta}}_{ETEL},\boldsymbol{\theta}%
_{0})=\frac{2nd_{\phi}(\widehat{\boldsymbol{t}}_{ETEL},\boldsymbol{t}_{0}%
)}{\phi^{\prime\prime}\left(  1\right)  }\\
&  =n\widehat{\boldsymbol{t}}_{ETEL}^{T}\boldsymbol{S}_{11}\left(
\boldsymbol{\theta}_{0}\right)  \widehat{\boldsymbol{t}}_{ETEL}%
+n\boldsymbol{t}_{0}^{T}\boldsymbol{S}_{11}\left(  \boldsymbol{\theta}%
_{0}\right)  \boldsymbol{t}_{0}-2n\widehat{\boldsymbol{t}}_{ETEL}%
^{T}\boldsymbol{S}_{11}\left(  \boldsymbol{\theta}_{0}\right)  \boldsymbol{t}%
_{0}+o(n||\widehat{\boldsymbol{t}}_{ETEL}||^{2})+o\left(  n||\boldsymbol{t}%
_{0}||^{2}\right)  .
\end{align*}
Denoting%
\[
\boldsymbol{h}(\boldsymbol{t}(\boldsymbol{\theta}))=\frac{1}{n}%
{\displaystyle\sum\limits_{i=1}^{n}}
\exp\{\boldsymbol{t}^{T}(\boldsymbol{\theta})\boldsymbol{g}(\boldsymbol{X}%
_{i},\boldsymbol{\theta})\}\boldsymbol{g}(\boldsymbol{X}_{i}%
,\boldsymbol{\theta}),
\]
from (\ref{eET}) the Taylor expansion of $\boldsymbol{h}(\boldsymbol{t}_{0})$
around $\boldsymbol{t}_{0}=\boldsymbol{0}_{r}$ is equal to%
\[
\boldsymbol{0}_{r}=\boldsymbol{h}(\boldsymbol{0}_{r})+\left(  \frac{\partial
}{\partial\boldsymbol{t}_{0}^{T}}\left.  \boldsymbol{h}(\boldsymbol{t}%
_{0})\right\vert _{\boldsymbol{t}_{0}=\boldsymbol{0}_{r}}\right)
\boldsymbol{t}_{0}+o\left(  ||\boldsymbol{t}_{0}||\boldsymbol{1}_{r}\right)
,
\]
where $\boldsymbol{h}(\boldsymbol{0}_{r})=\overline{\boldsymbol{g}}%
_{n}(\boldsymbol{X},\boldsymbol{\theta}_{0})$, $\frac{\partial}{\partial
\boldsymbol{t}_{0}^{T}}\left.  \boldsymbol{h}(\boldsymbol{t}_{0})\right\vert
_{\boldsymbol{t}_{0}=\boldsymbol{0}_{r}}=\widehat{\boldsymbol{S}}_{11}\left(
\boldsymbol{\theta}_{0}\right)  =\boldsymbol{S}_{11}\left(  \boldsymbol{\theta
}_{0}\right)  +o_{p}(\boldsymbol{1}_{r\times r})$, and from it the following
relation is obtained%
\begin{equation}
n^{1/2}\boldsymbol{t}_{0}=-\boldsymbol{S}_{11}^{-1}\left(  \boldsymbol{\theta
}_{0}\right)  n^{1/2}\overline{\boldsymbol{g}}_{n}(\boldsymbol{X}%
,\boldsymbol{\theta}_{0})+o_{p}(\boldsymbol{1}_{r}), \label{ecT0}%
\end{equation}
Taking into account (\ref{V}), (\ref{T12}) and (\ref{ecT0}), it holds
\begin{align*}
n\widehat{\boldsymbol{t}}_{ETEL}^{T}\boldsymbol{S}_{11}\left(
\boldsymbol{\theta}_{0}\right)  \widehat{\boldsymbol{t}}_{ETEL}  &
=n\overline{\boldsymbol{g}}_{n}^{T}(\boldsymbol{X},\boldsymbol{\theta}%
_{0})\boldsymbol{R}\left(  \boldsymbol{\theta}_{0}\right)  \overline
{\boldsymbol{g}}_{n}(\boldsymbol{X},\boldsymbol{\theta}_{0})+o_{p}(1),\\
n\boldsymbol{t}_{0}^{T}\boldsymbol{S}_{11}\left(  \boldsymbol{\theta}%
_{0}\right)  \boldsymbol{t}_{0}  &  =n\overline{\boldsymbol{g}}_{n}%
^{T}(\boldsymbol{X},\boldsymbol{\theta}_{0})\boldsymbol{S}_{11}^{-1}\left(
\boldsymbol{\theta}_{0}\right)  \overline{\boldsymbol{g}}_{n}(\boldsymbol{X}%
,\boldsymbol{\theta}_{0})+o_{p}(1),\\
n\widehat{\boldsymbol{t}}_{ETEL}^{T}\boldsymbol{S}_{11}\left(
\boldsymbol{\theta}_{0}\right)  \boldsymbol{t}_{0}  &  =n\overline
{\boldsymbol{g}}_{n}^{T}(\boldsymbol{X},\boldsymbol{\theta}_{0})\boldsymbol{R}%
\left(  \boldsymbol{\theta}_{0}\right)  \overline{\boldsymbol{g}}%
_{n}(\boldsymbol{X},\boldsymbol{\theta}_{0})+o_{p}(1),
\end{align*}
and consequently
\begin{align*}
&  S_{n}^{\phi}(\widehat{\boldsymbol{\theta}}_{ETEL},\boldsymbol{\theta}%
_{0})=\frac{2nd_{\phi}(\widehat{\boldsymbol{t}}_{ETEL},\boldsymbol{t}_{0}%
)}{\phi^{\prime\prime}\left(  1\right)  }\\
&  =n\overline{\boldsymbol{g}}_{n}^{T}(\boldsymbol{X},\boldsymbol{\theta}%
_{0})\boldsymbol{S}_{11}^{-1}\left(  \boldsymbol{\theta}_{0}\right)
\boldsymbol{S}_{12}\left(  \boldsymbol{\theta}_{0}\right)  \boldsymbol{V}%
\left(  \boldsymbol{\theta}_{0}\right)  \boldsymbol{S}_{12}^{T}\left(
\boldsymbol{\theta}_{0}\right)  \boldsymbol{S}_{11}^{-1}\left(
\boldsymbol{\theta}_{0}\right)  \overline{\boldsymbol{g}}_{n}(\boldsymbol{X}%
,\boldsymbol{\theta}_{0})+o_{p}(1)\\
&  =n\overline{\boldsymbol{g}}_{n}^{T}(\boldsymbol{X},\boldsymbol{\theta}%
_{0})\boldsymbol{S}_{11}^{-1}\left(  \boldsymbol{\theta}_{0}\right)
\boldsymbol{S}_{12}\left(  \boldsymbol{\theta}_{0}\right)
\boldsymbol{V\left(  \boldsymbol{\theta}_{0}\right)  V}^{-1}\left(
\boldsymbol{\theta}_{0}\right)  \boldsymbol{V\left(  \boldsymbol{\theta}%
_{0}\right)  \boldsymbol{S}_{12}^{T}\left(  \boldsymbol{\theta}_{0}\right)
S}_{11}^{-1}\left(  \boldsymbol{\theta}_{0}\right)  \overline{\boldsymbol{g}%
}_{n}(\boldsymbol{X},\boldsymbol{\theta}_{0})+o_{p}(1)\\
&  =\sqrt{n}(\widehat{\boldsymbol{\theta}}_{ETEL}-\boldsymbol{\theta}_{0}%
)^{T}\boldsymbol{V}^{-1}\left(  \boldsymbol{\theta}_{0}\right)  \sqrt
{n}(\widehat{\boldsymbol{\theta}}_{ETEL}-\boldsymbol{\theta}_{0})+o_{p}(1)\\
&  =\left(  \sqrt{n}\boldsymbol{V}^{-1/2}\left(  \boldsymbol{\theta}%
_{0}\right)  (\widehat{\boldsymbol{\theta}}_{ETEL}-\boldsymbol{\theta}%
_{0})\right)  ^{T}\sqrt{n}\boldsymbol{V}^{-1/2}\left(  \boldsymbol{\theta}%
_{0}\right)  (\widehat{\boldsymbol{\theta}}_{ETEL}-\boldsymbol{\theta}%
_{0})+o_{p}(1).
\end{align*}
It is clear that
\[
\sqrt{n}\boldsymbol{V}^{-1/2}\left(  \boldsymbol{\theta}_{0}\right)
(\widehat{\boldsymbol{\theta}}_{ETEL}-\boldsymbol{\theta}_{0}%
)\overset{\mathcal{L}}{\underset{n\rightarrow\infty}{\longrightarrow}%
}\mathcal{N}(\boldsymbol{0},\boldsymbol{I}_{p}),
\]
where $\boldsymbol{I}_{p}$\ is the $p\times p$ identity matrix. Now, applying
Lemma 3 of Ferguson (1996), we readily obtain the desired asymptotic distribution.
\end{proof}
\end{theorem}

Based on the asymptotic null distribution presented in Theorem \ref{Th1}, we
reject the null hypothesis in (\ref{H}), with significance level $\alpha,$ in
favour of the alternative hypothesis, if $S_{n}^{\phi}%
(\widehat{\boldsymbol{\theta}}_{ETEL},\boldsymbol{\theta}_{0})>\chi_{p,\alpha
}^{2}$ (or\ if $T_{n}^{\phi}(\widehat{\boldsymbol{\theta}}_{ETEL}%
,\boldsymbol{\theta}_{0})>\chi_{p,\alpha}^{2})$), where $\chi_{p,\alpha}^{2}$
is the $\left(  1-\alpha\right)  $-th quantile of the chi-squared distribution
with $p$ degrees of freedom. In most cases, the power function of this test
procedure cannot be derived explicitly. In the following theorem, we present
an asymptotic result, which provides an approximation of the power of the
empirical $\phi$-divergence test statistics described previously.

\begin{theorem}
\label{th5}Under the assumption that $\boldsymbol{\theta}^{\ast}%
\neq\boldsymbol{\theta}_{0}$ is the true parameter value%
\[
\frac{n^{1/2}}{\sqrt{\boldsymbol{s}_{T_{n}^{\phi}}^{T}(\boldsymbol{\theta}%
_{0},\boldsymbol{\theta}^{\ast})\boldsymbol{M}_{T_{n}^{\phi}}%
(\boldsymbol{\theta}_{0},\boldsymbol{\theta}^{\ast})\boldsymbol{s}%
_{T_{n}^{\phi}}(\boldsymbol{\theta}_{0},\boldsymbol{\theta}^{\ast})}}\left(
\frac{\phi^{\prime\prime}(1)T_{n}^{\phi}(\widehat{\boldsymbol{\theta}}%
_{ETEL},\boldsymbol{\theta}_{0})}{2n}-\mu_{\phi}(\boldsymbol{\theta}%
_{0},\boldsymbol{\theta}^{\ast})\right)  \overset{\mathcal{L}%
}{\underset{n\rightarrow\infty}{\longrightarrow}}\mathcal{N}\left(
0,1\right)  ,
\]
where
\begin{equation}
\boldsymbol{s}_{T_{n}^{\phi}}(\boldsymbol{\theta}_{0},\boldsymbol{\theta
}^{\ast})=\mathrm{E}_{F_{\boldsymbol{\theta}^{\ast}}}^{-1}\left[
\exp\{\boldsymbol{\tau}^{T}\boldsymbol{g}(\boldsymbol{X},\boldsymbol{\theta
}_{0})\}\right]  \mathrm{E}_{F_{\boldsymbol{\theta}^{\ast}}}\left[
\exp\{\boldsymbol{\tau}^{T}\boldsymbol{g}(\boldsymbol{X},\boldsymbol{\theta
}_{0})\}\psi\left(  \frac{\mathrm{E}_{F_{\boldsymbol{\theta}^{\ast}}}\left[
\exp\{\boldsymbol{\tau}^{T}\boldsymbol{g}(\boldsymbol{X},\boldsymbol{\theta
}_{0})\}\right]  }{\exp\{\boldsymbol{\tau}^{T}\boldsymbol{g}(\boldsymbol{X}%
,\boldsymbol{\theta}_{0})\}}\right)  \boldsymbol{g}(\boldsymbol{X}%
,\boldsymbol{\theta}_{0})\right]  , \label{sT}%
\end{equation}
$\boldsymbol{\tau}$ is the solution of
\[
\mathrm{E}_{F_{\boldsymbol{\theta}^{\ast}}}[\exp\{\boldsymbol{\tau}%
^{T}\boldsymbol{g}(\boldsymbol{X},\boldsymbol{\theta}_{0})\}\boldsymbol{g}%
(\boldsymbol{X},\boldsymbol{\theta}_{0})]=\boldsymbol{0}_{r},
\]%
\begin{equation}
\psi(x)=\phi(x)-x\phi^{\prime}(x), \label{psi}%
\end{equation}%
\begin{align}
\boldsymbol{M}_{T_{n}^{\phi}}(\boldsymbol{\theta}_{0},\boldsymbol{\theta
}^{\ast})  &  =\mathrm{E}_{F_{\boldsymbol{\theta}^{\ast}}}^{-1}\left[
\exp\{\boldsymbol{\tau}^{T}\boldsymbol{g}(\boldsymbol{X},\boldsymbol{\theta
}_{0})\}\boldsymbol{g}(\boldsymbol{X},\boldsymbol{\theta}_{0})\boldsymbol{g}%
^{T}(\boldsymbol{X},\boldsymbol{\theta}_{0})\right]  \mathrm{E}%
_{F_{\boldsymbol{\theta}^{\ast}}}[\exp\{2\boldsymbol{\tau}^{T}\boldsymbol{g}%
(\boldsymbol{X},\boldsymbol{\theta}_{0})\}\boldsymbol{g}(\boldsymbol{X}%
,\boldsymbol{\theta}_{0})\boldsymbol{g}^{T}(\boldsymbol{X},\boldsymbol{\theta
}_{0})])\nonumber\\
&  \times\mathrm{E}_{F_{\boldsymbol{\theta}^{\ast}}}^{-1}\left[
\exp\{\boldsymbol{\tau}^{T}\boldsymbol{g}(\boldsymbol{X},\boldsymbol{\theta
}_{0})\}\boldsymbol{g}(\boldsymbol{X},\boldsymbol{\theta}_{0})\boldsymbol{g}%
^{T}(\boldsymbol{X},\boldsymbol{\theta}_{0})\right]  , \label{MT}%
\end{align}
and
\begin{equation}
\mu_{\phi}(\boldsymbol{\theta}_{0},\boldsymbol{\theta}^{\ast})=\mathrm{E}%
_{F_{\boldsymbol{\theta}^{\ast}}}^{-1}\left[  \exp\{\boldsymbol{\tau}%
^{T}\boldsymbol{g}(\boldsymbol{X},\boldsymbol{\theta}_{0})\}\right]
\mathrm{E}_{F_{\boldsymbol{\theta}^{\ast}}}\left[  \exp\{\boldsymbol{\tau}%
^{T}\boldsymbol{g}(\boldsymbol{X},\boldsymbol{\theta}_{0})\}\phi\left(
\frac{\mathrm{E}\left[  \exp\{\boldsymbol{\tau}^{T}\boldsymbol{g}%
(\boldsymbol{X},\boldsymbol{\theta}_{0})\}\right]  }{\exp\{\boldsymbol{\tau
}^{T}\boldsymbol{g}(\boldsymbol{X},\boldsymbol{\theta}_{0})\}}\right)
\right]  . \label{muT}%
\end{equation}

\end{theorem}

\begin{proof}
We rename $D_{\phi}\left(  \boldsymbol{u},\boldsymbol{p}_{ET}\left(
\boldsymbol{\theta}\right)  \right)  =d_{\phi}(\boldsymbol{u},\boldsymbol{t}%
(\boldsymbol{\theta}))$ as a function of $\boldsymbol{u}$ and $\boldsymbol{t}%
(\boldsymbol{\theta})$, i.e.%
\[
d_{\phi}(\boldsymbol{u},\boldsymbol{t}(\boldsymbol{\theta}))=\left(
{\displaystyle\sum\limits_{j=1}^{n}}
\exp\{\boldsymbol{t}^{T}(\boldsymbol{\theta})\boldsymbol{g}(\boldsymbol{X}%
_{j},\boldsymbol{\theta})\}\right)  ^{-1}%
{\displaystyle\sum\limits_{i=1}^{n}}
\exp\{\boldsymbol{t}^{T}(\boldsymbol{\theta})\boldsymbol{g}(\boldsymbol{X}%
_{i},\boldsymbol{\theta})\}\phi\left(  \frac{%
{\displaystyle\sum\limits_{j=1}^{n}}
\exp\{\boldsymbol{t}^{T}(\boldsymbol{\theta})\boldsymbol{g}(\boldsymbol{X}%
_{j},\boldsymbol{\theta})\}}{n\exp\{\boldsymbol{t}^{T}(\boldsymbol{\theta
})\boldsymbol{g}(\boldsymbol{X}_{i},\boldsymbol{\theta})\}}\right)  ,
\]
and in particular for $\boldsymbol{\theta=\theta}_{0}$ and $\boldsymbol{\theta
=}\widehat{\boldsymbol{\theta}}_{ETEL}$, $D_{\phi}\left(  \boldsymbol{u}%
,\boldsymbol{p}_{ET}\left(  \boldsymbol{\theta}_{0}\right)  \right)  =d_{\phi
}(\boldsymbol{u},\boldsymbol{t}_{0})$\ and $D_{\phi}(\boldsymbol{u}%
,\boldsymbol{p}_{ET}(\widehat{\boldsymbol{\theta}}_{ETEL}))=d_{\phi
}(\boldsymbol{u},\widehat{\boldsymbol{t}}_{ETEL})$. Since $\boldsymbol{t}%
_{0}\overset{P}{\underset{n\rightarrow\infty}{\longrightarrow}}%
\boldsymbol{\tau}$, we shall consider, on one hand, the first order Taylor
expansion of $d_{\phi}(\boldsymbol{u},\boldsymbol{t}_{0})$ around
$\boldsymbol{t}_{0}=\boldsymbol{\tau}$
\[
d_{\phi}(\boldsymbol{u},\boldsymbol{t}_{0})=d_{\phi}(\boldsymbol{u}%
,\boldsymbol{\tau})+\left.  \frac{\partial d_{\phi}\left(  \boldsymbol{u}%
,\boldsymbol{t}_{0}\right)  }{\partial\boldsymbol{t}_{0}^{T}}\right\vert
_{\boldsymbol{t}_{0}=\boldsymbol{\tau}}(\boldsymbol{t}_{0}-\boldsymbol{\tau
})+o(||\boldsymbol{t}_{0}-\boldsymbol{\tau}||),
\]
where%
\[
\frac{\partial d_{\phi}\left(  \boldsymbol{u},\boldsymbol{t}%
(\boldsymbol{\theta})\right)  }{\partial\boldsymbol{t}(\boldsymbol{\theta}%
)}=\left(
{\displaystyle\sum\limits_{j=1}^{n}}
\exp\{\boldsymbol{t}^{T}(\boldsymbol{\theta})\boldsymbol{g}(\boldsymbol{X}%
_{j},\boldsymbol{\theta})\}\right)  ^{-1}%
{\displaystyle\sum\limits_{i=1}^{n}}
\exp\{\boldsymbol{t}^{T}(\boldsymbol{\theta})\boldsymbol{g}(\boldsymbol{X}%
_{i},\boldsymbol{\theta})\}\psi\left(  \frac{%
{\displaystyle\sum\limits_{j=1}^{n}}
\exp\{\boldsymbol{t}^{T}(\boldsymbol{\theta})\boldsymbol{g}(\boldsymbol{X}%
_{j},\boldsymbol{\theta})\}}{n\exp\{\boldsymbol{t}^{T}(\boldsymbol{\theta
})\boldsymbol{g}(\boldsymbol{X}_{i},\boldsymbol{\theta})\}}\right)
\boldsymbol{g}(\boldsymbol{X}_{i},\boldsymbol{\theta}),
\]
and since $\widehat{\boldsymbol{t}}_{ETEL}\overset{P}{\underset{n\rightarrow
\infty}{\longrightarrow}}\boldsymbol{0}_{r}$, we shall consider, on the other
hand, the first order Taylor expansion of $d_{\phi}(\boldsymbol{u}%
,\widehat{\boldsymbol{t}}_{ETEL})$ around $\widehat{\boldsymbol{t}}%
_{ETEL}=\boldsymbol{0}_{r}$%
\[
d_{\phi}(\boldsymbol{u},\widehat{\boldsymbol{t}}_{ETEL}%
)=o(||\widehat{\boldsymbol{t}}_{ETEL}||).
\]
Then,%
\begin{equation}
d_{\phi}(\boldsymbol{u},\boldsymbol{t}_{0})-d_{\phi}(\boldsymbol{u}%
,\widehat{\boldsymbol{t}}_{ETEL})=d_{\phi}(\boldsymbol{u},\boldsymbol{\tau
})+\boldsymbol{s}_{T_{n}^{\phi}}^{T}(\boldsymbol{\theta}_{0}%
,\boldsymbol{\theta}^{\ast})(\boldsymbol{t}_{0}-\boldsymbol{\tau
})+o(||\boldsymbol{t}_{0}-\boldsymbol{\tau}||)+o(||\widehat{\boldsymbol{t}%
}_{ETEL}||), \label{taylorT}%
\end{equation}
where $\boldsymbol{s}_{T_{n}^{\phi}}$, given by (\ref{sT}), is such that%
\[
\left.  \frac{\partial d_{\phi}\left(  \boldsymbol{u},\boldsymbol{t}%
_{0}\right)  }{\partial\boldsymbol{t}_{0}}\right\vert _{\boldsymbol{t}%
_{0}=\boldsymbol{\tau}}\overset{P}{\underset{n\rightarrow\infty
}{\longrightarrow}}\boldsymbol{s}_{T_{n}^{\phi}}(\boldsymbol{\theta}%
_{0},\boldsymbol{\theta}^{\ast}).
\]
Denoting%
\[
\boldsymbol{h}(\boldsymbol{t}(\boldsymbol{\theta}))=\frac{1}{n}%
{\displaystyle\sum\limits_{i=1}^{n}}
\exp\{\boldsymbol{t}^{T}(\boldsymbol{\theta})\boldsymbol{g}(\boldsymbol{X}%
_{i},\boldsymbol{\theta})\}\boldsymbol{g}(\boldsymbol{X}_{i}%
,\boldsymbol{\theta}),
\]
the Taylor expansion of $\boldsymbol{h}(\boldsymbol{t}_{0})$ around
$\boldsymbol{t}_{0}=\boldsymbol{\tau}$ is equal to%
\[
\boldsymbol{0}_{r}=\boldsymbol{h}(\boldsymbol{\tau})+\left(  \frac{\partial
}{\partial\boldsymbol{t}_{0}^{T}}\left.  \boldsymbol{h}(\boldsymbol{t}%
_{0})\right\vert _{\boldsymbol{t}_{0}=\boldsymbol{\tau}}\right)
(\boldsymbol{t}_{0}-\boldsymbol{\tau})+o\left(  ||\boldsymbol{t}%
_{0}-\boldsymbol{\tau}||\boldsymbol{1}_{r}\right)  ,
\]
where
\begin{align*}
\boldsymbol{h}(\boldsymbol{\tau})  &  =\frac{1}{n}%
{\displaystyle\sum\limits_{i=1}^{n}}
\exp\{\boldsymbol{\tau}^{T}\boldsymbol{g}(\boldsymbol{X}_{i}%
,\boldsymbol{\theta}_{0})\}\boldsymbol{g}(\boldsymbol{X}_{i}%
,\boldsymbol{\theta}_{0}),\\
\frac{\partial}{\partial\boldsymbol{t}_{0}^{T}}\left.  \boldsymbol{h}%
(\boldsymbol{t}_{0})\right\vert _{\boldsymbol{t}_{0}=\boldsymbol{\tau}}  &
=\frac{1}{n}%
{\displaystyle\sum\limits_{i=1}^{n}}
\exp\{\boldsymbol{\tau}^{T}\boldsymbol{g}(\boldsymbol{X}_{i}%
,\boldsymbol{\theta}_{0})\}\boldsymbol{g}(\boldsymbol{X}_{i}%
,\boldsymbol{\theta}_{0})\boldsymbol{g}^{T}(\boldsymbol{X}_{i}%
,\boldsymbol{\theta}_{0})\\
&  =\mathrm{E}_{F_{\boldsymbol{\theta}^{\ast}}}\left[  \exp\{\boldsymbol{\tau
}^{T}\boldsymbol{g}(\boldsymbol{X},\boldsymbol{\theta}_{0})\}\boldsymbol{g}%
(\boldsymbol{X},\boldsymbol{\theta}_{0})\boldsymbol{g}^{T}(\boldsymbol{X}%
,\boldsymbol{\theta}_{0})\right]  +o_{p}(\boldsymbol{1}_{r\times r}),
\end{align*}
and from it the following relation is obtained%
\[
\boldsymbol{t}_{0}-\boldsymbol{\tau}=-\mathrm{E}_{F_{\boldsymbol{\theta}%
^{\ast}}}^{-1}\left[  \exp\{\boldsymbol{\tau}^{T}\boldsymbol{g}(\boldsymbol{X}%
,\boldsymbol{\theta}_{0})\}\boldsymbol{g}(\boldsymbol{X},\boldsymbol{\theta
}_{0})\boldsymbol{g}^{T}(\boldsymbol{X},\boldsymbol{\theta}_{0})\right]
\left(  \frac{1}{n}%
{\displaystyle\sum\limits_{i=1}^{n}}
\exp\{\boldsymbol{\tau}^{T}\boldsymbol{g}(\boldsymbol{X}_{i}%
,\boldsymbol{\theta}_{0})\}\boldsymbol{g}(\boldsymbol{X}_{i}%
,\boldsymbol{\theta}_{0})\right)  +o_{p}(\boldsymbol{1}_{r}).
\]
We obtain in virtue of the Central Limit Theorem%
\[
\sqrt{n}\left(  \boldsymbol{t}_{0}-\boldsymbol{\tau}\right)
\overset{\mathcal{L}}{\underset{n\rightarrow\infty}{\longrightarrow}%
}\mathcal{N}\left(  \boldsymbol{0}_{r},\boldsymbol{M}_{T_{n}^{\phi}%
}(\boldsymbol{\theta}_{0},\boldsymbol{\theta}^{\ast})\right)  ,
\]
where $\boldsymbol{M}_{T_{n}^{\phi}}(\boldsymbol{\theta}_{0}%
,\boldsymbol{\theta}^{\ast})$ is (\ref{MT}), since%
\[
\mathrm{E}_{F_{\boldsymbol{\theta}^{\ast}}}[\exp\{\boldsymbol{\tau}%
^{T}\boldsymbol{g}(\boldsymbol{X},\boldsymbol{\theta}_{0})\}\boldsymbol{g}%
(\boldsymbol{X},\boldsymbol{\theta}_{0})]=\boldsymbol{0}_{r},
\]
and%
\[
\sqrt{n}\left(  \frac{1}{n}%
{\displaystyle\sum\limits_{i=1}^{n}}
\exp\{\boldsymbol{\tau}^{T}\boldsymbol{g}(\boldsymbol{X}_{i}%
,\boldsymbol{\theta}_{0})\}\boldsymbol{g}(\boldsymbol{X}_{i}%
,\boldsymbol{\theta}_{0})\right)  \overset{\mathcal{L}}{\underset{n\rightarrow
\infty}{\longrightarrow}}\mathcal{N}(\boldsymbol{0}_{r},\mathrm{E}%
_{F_{\boldsymbol{\theta}^{\ast}}}[\exp\{2\boldsymbol{\tau}^{T}\boldsymbol{g}%
(\boldsymbol{X},\boldsymbol{\theta}_{0})\}\boldsymbol{g}(\boldsymbol{X}%
,\boldsymbol{\theta}_{0})\boldsymbol{g}^{T}(\boldsymbol{X},\boldsymbol{\theta
}_{0})]),
\]
On the other hand, since%
\[
d_{\phi}(\boldsymbol{u},\boldsymbol{\tau})=\left(
{\displaystyle\sum\limits_{j=1}^{n}}
\exp\{\boldsymbol{\tau}^{T}\boldsymbol{g}(\boldsymbol{X}_{j}%
,\boldsymbol{\theta}_{0})\}\right)  ^{-1}%
{\displaystyle\sum\limits_{i=1}^{n}}
\exp\{\boldsymbol{\tau}^{T}\boldsymbol{g}(\boldsymbol{X}_{i}%
,\boldsymbol{\theta}_{0})\}\phi\left(  \frac{%
{\displaystyle\sum\limits_{j=1}^{n}}
\exp\{\boldsymbol{\tau}^{T}\boldsymbol{g}(\boldsymbol{X}_{j}%
,\boldsymbol{\theta}_{0})\}}{n\exp\{\boldsymbol{\tau}^{T}\boldsymbol{g}%
(\boldsymbol{X}_{i},\boldsymbol{\theta}_{0})\}}\right)  ,
\]
it holds
\[
d_{\phi}(\boldsymbol{u},\boldsymbol{t}_{0})\overset{P}{\underset{n\rightarrow
\infty}{\longrightarrow}}\mu_{\phi}(\boldsymbol{\theta}_{0},\boldsymbol{\theta
}^{\ast}),
\]
where $\mu_{\phi}(\boldsymbol{\theta}_{0},\boldsymbol{\theta}^{\ast})$ is
(\ref{muT}). Hence, from (\ref{taylorT}) it follows%
\[
\sqrt{n}\left(  \frac{d_{\phi}(\boldsymbol{u},\boldsymbol{t}_{0})-d_{\phi
}(\boldsymbol{u},\widehat{\boldsymbol{t}}_{ETEL})-\mu_{\phi}%
(\boldsymbol{\theta}_{0},\boldsymbol{\theta}^{\ast})}{\sqrt{\boldsymbol{s}%
_{T_{n}^{\phi}}^{T}(\boldsymbol{\theta}_{0},\boldsymbol{\theta}^{\ast
})\boldsymbol{M}_{T_{n}^{\phi}}(\boldsymbol{\theta}_{0},\boldsymbol{\theta
}^{\ast})\boldsymbol{s}_{T_{n}^{\phi}}(\boldsymbol{\theta}_{0}%
,\boldsymbol{\theta}^{\ast})}}\right)  \overset{\mathcal{L}%
}{\underset{n\rightarrow\infty}{\longrightarrow}}\mathcal{N}\left(
0,1\right)  ,
\]
which is equivalent to the enunciated result.
\end{proof}

\begin{theorem}
\label{th5b}Under the assumption that $\boldsymbol{\theta}^{\ast}%
\neq\boldsymbol{\theta}_{0}$ is the true parameter value%
\[
\frac{n^{1/2}}{\sqrt{\boldsymbol{s}_{S_{n}^{\phi}}^{T}(\boldsymbol{\theta}%
_{0},\boldsymbol{\theta}^{\ast})\boldsymbol{M}_{S_{n}^{\phi}}%
(\boldsymbol{\theta}_{0},\boldsymbol{\theta}^{\ast})\boldsymbol{s}%
_{S_{n}^{\phi}}(\boldsymbol{\theta}_{0},\boldsymbol{\theta}^{\ast})}}\left(
\frac{\phi^{\prime\prime}(1)S_{n}^{\phi}(\widehat{\boldsymbol{\theta}}%
_{ETEL},\boldsymbol{\theta}_{0})}{2n}-\mu_{\phi}(\boldsymbol{\theta}%
_{0},\boldsymbol{\theta}^{\ast})\right)  \overset{\mathcal{L}%
}{\underset{n\rightarrow\infty}{\longrightarrow}}\mathcal{N}\left(
0,1\right)  ,
\]
where
\begin{align}
\boldsymbol{s}_{S_{n}^{\phi}}(\boldsymbol{\theta}_{0},\boldsymbol{\theta
}^{\ast})  &  =%
\begin{pmatrix}
\boldsymbol{s}_{1,S_{n}^{\phi}}(\boldsymbol{\theta}_{0},\boldsymbol{\theta
}^{\ast})\\
\boldsymbol{s}_{2,S_{n}^{\phi}}(\boldsymbol{\theta}_{0},\boldsymbol{\theta
}^{\ast})
\end{pmatrix}
,\label{sS}\\
\boldsymbol{s}_{1,S_{n}^{\phi}}(\boldsymbol{\theta}_{0},\boldsymbol{\theta
}^{\ast})  &  =-\mathbf{R}(\boldsymbol{\theta}^{\ast})\mathrm{E}%
_{F_{\boldsymbol{\theta}^{\ast}}}\left[  \phi^{\prime}\left(  \frac
{\mathrm{E}_{F_{\boldsymbol{\theta}^{\ast}}}\left[  \exp\{\boldsymbol{\tau
}^{T}\boldsymbol{g}(\boldsymbol{X},\boldsymbol{\theta}_{0})\}\right]  }%
{\exp\{\boldsymbol{\tau}^{T}\boldsymbol{g}(\boldsymbol{X},\boldsymbol{\theta
}_{0})\}}\right)  \boldsymbol{g}(\boldsymbol{X},\boldsymbol{\theta}^{\ast
})\right]  ,\nonumber\\
\boldsymbol{s}_{2,S_{n}^{\phi}}(\boldsymbol{\theta}_{0},\boldsymbol{\theta
}^{\ast})  &  =-\mathrm{E}_{F_{\boldsymbol{\theta}^{\ast}}}^{-1}\left[
\exp\{\boldsymbol{\tau}^{T}\boldsymbol{g}(\boldsymbol{X},\boldsymbol{\theta
}_{0})\}\right]  \mathrm{E}_{F_{\boldsymbol{\theta}^{\ast}}}^{-1}\left[
\exp\{\boldsymbol{\tau}^{T}\boldsymbol{g}(\boldsymbol{X},\boldsymbol{\theta
}_{0})\}\boldsymbol{g}(\boldsymbol{X},\boldsymbol{\theta}_{0})\boldsymbol{g}%
^{T}(\boldsymbol{X},\boldsymbol{\theta}_{0})\right] \nonumber\\
&  \times\mathrm{E}_{F_{\boldsymbol{\theta}^{\ast}}}\left[  \exp
\{\boldsymbol{\tau}^{T}\boldsymbol{g}(\boldsymbol{X},\boldsymbol{\theta}%
_{0})\}\psi\left(  \frac{\mathrm{E}_{F_{\boldsymbol{\theta}^{\ast}}}\left[
\exp\{\boldsymbol{\tau}^{T}\boldsymbol{g}(\boldsymbol{X},\boldsymbol{\theta
}_{0})\}\right]  }{\exp\{\boldsymbol{\tau}^{T}\boldsymbol{g}(\boldsymbol{X}%
,\boldsymbol{\theta}_{0})\}}\right)  \boldsymbol{g}(\boldsymbol{X}%
,\boldsymbol{\theta}_{0})\right]  ,\nonumber
\end{align}%
\begin{align}
\boldsymbol{M}_{S_{n}^{\phi}}(\boldsymbol{\theta}_{0},\boldsymbol{\theta
}^{\ast})  &  =%
\begin{pmatrix}
\mathbf{S}_{11}(\boldsymbol{\theta}^{\ast}) & \boldsymbol{\Sigma}%
_{12}(\boldsymbol{\theta}^{\ast},\boldsymbol{\theta}_{0})\\
\boldsymbol{\Sigma}_{12}^{T}(\boldsymbol{\theta}^{\ast},\boldsymbol{\theta
}_{0}) & \boldsymbol{\Sigma}_{22}(\boldsymbol{\theta}^{\ast}%
,\boldsymbol{\theta}_{0})
\end{pmatrix}
,\label{MS}\\
\boldsymbol{\Sigma}_{12}(\boldsymbol{\theta}_{0},\boldsymbol{\theta}^{\ast})
&  =\mathrm{E}_{F_{\boldsymbol{\theta}^{\ast}}}\left[  \exp\{\boldsymbol{\tau
}^{T}\boldsymbol{g}(\boldsymbol{X},\boldsymbol{\theta}_{0})\}\boldsymbol{g}%
(\boldsymbol{X},\boldsymbol{\theta}^{\ast})\boldsymbol{g}^{T}(\boldsymbol{X}%
,\boldsymbol{\theta}_{0})\right]  ,\nonumber\\
\boldsymbol{\Sigma}_{22}(\boldsymbol{\theta}_{0},\boldsymbol{\theta}^{\ast})
&  =\mathrm{E}_{F_{\boldsymbol{\theta}^{\ast}}}\left[  \exp\{2\boldsymbol{\tau
}^{T}\boldsymbol{g}(\boldsymbol{X},\boldsymbol{\theta}_{0})\}\boldsymbol{g}%
(\boldsymbol{X},\boldsymbol{\theta}_{0})\boldsymbol{g}^{T}(\boldsymbol{X}%
,\boldsymbol{\theta}_{0})\right]  .\nonumber
\end{align}
$\boldsymbol{\tau}$, $\psi$\ and $\mu_{\phi}(\boldsymbol{\theta}%
_{0},\boldsymbol{\theta}^{\ast})$ as in Theorem \ref{th5}.
\end{theorem}

\begin{proof}
Since $\widehat{\boldsymbol{t}}_{ETEL}\overset{P}{\underset{n\rightarrow
\infty}{\longrightarrow}}\boldsymbol{0}_{r}$ and $\boldsymbol{t}%
_{0}\overset{P}{\underset{n\rightarrow\infty}{\longrightarrow}}%
\boldsymbol{\tau}$, we shall consider the first order Taylor expansion of
$d_{\phi}(\widehat{\boldsymbol{t}}_{ETEL},\boldsymbol{t}_{0})$ around
$(\widehat{\boldsymbol{t}}_{ETEL},\boldsymbol{t}_{0})=(\boldsymbol{0}%
_{r},\boldsymbol{\tau})$,%
\begin{align*}
d_{\phi}(\widehat{\boldsymbol{t}}_{ETEL},\boldsymbol{t}_{0})  &  =d_{\phi
}(\boldsymbol{0}_{r},\boldsymbol{\tau})+\left.  \frac{\partial d_{\phi
}(\widehat{\boldsymbol{t}}_{ETEL},\boldsymbol{\tau})}{\partial
\widehat{\boldsymbol{t}}_{ETEL}^{T}}\right\vert _{\widehat{\boldsymbol{t}%
}_{ETEL}=\boldsymbol{0}_{r}}\widehat{\boldsymbol{t}}_{ETEL}+\left.
\frac{\partial d_{\phi}(\boldsymbol{0}_{r},\boldsymbol{t}_{0})}{\partial
\boldsymbol{t}_{0}^{T}}\right\vert _{\boldsymbol{t}_{0}=\boldsymbol{\tau}%
}(\boldsymbol{t}_{0}-\boldsymbol{\tau})\\
&  +o(||\widehat{\boldsymbol{t}}_{ETEL}||)+o(||\boldsymbol{t}_{0}%
-\boldsymbol{\tau}||),
\end{align*}
where%
\[
d_{\phi}(\boldsymbol{0}_{r},\boldsymbol{\tau})=\left(
{\displaystyle\sum\limits_{j=1}^{n}}
\exp\{\boldsymbol{\tau}^{T}\boldsymbol{g}(\boldsymbol{X}_{j}%
,\boldsymbol{\theta}_{0})\}\right)  ^{-1}%
{\displaystyle\sum\limits_{i=1}^{n}}
\exp\{\boldsymbol{\tau}^{T}\boldsymbol{g}(\boldsymbol{X}_{i}%
,\boldsymbol{\theta}_{0})\}\phi\left(  \frac{%
{\displaystyle\sum\limits_{j=1}^{n}}
\exp\{\boldsymbol{\tau}^{T})\boldsymbol{g}(\boldsymbol{X}_{j}%
,\boldsymbol{\theta}_{0})\}}{n\exp\{\boldsymbol{\tau}^{T}\boldsymbol{g}%
(\boldsymbol{X}_{i},\boldsymbol{\theta}_{0})\}}\right)  ,
\]%
\begin{align*}
&  \frac{\partial d_{\phi}\left(  \widehat{\boldsymbol{t}}_{ETEL}%
,\boldsymbol{t}_{0}\right)  }{\partial\widehat{\boldsymbol{t}}_{ETEL}}=\left(
%
{\displaystyle\sum\limits_{j=1}^{n}}
\exp\{\widehat{\boldsymbol{t}}_{ETEL}^{T}\boldsymbol{g}(\boldsymbol{X}%
_{j},\widehat{\boldsymbol{\theta}}_{ETEL})\}\right)  ^{-1}%
{\displaystyle\sum\limits_{i=1}^{n}}
\exp\{\widehat{\boldsymbol{t}}_{ETEL}^{T}\boldsymbol{g}(\boldsymbol{X}%
_{i},\widehat{\boldsymbol{\theta}}_{ETEL})\}\\
&  \times\phi^{\prime}\left(  \frac{%
{\displaystyle\sum\limits_{j=1}^{n}}
\exp\{\boldsymbol{t}_{0}^{T}\boldsymbol{g}(\boldsymbol{X}_{j}%
,\boldsymbol{\theta}_{0})\}}{\exp\{\boldsymbol{t}_{0}^{T}\boldsymbol{g}%
(\boldsymbol{X}_{i},\boldsymbol{\theta}_{0})\}}\frac{\exp
\{\widehat{\boldsymbol{t}}_{ETEL}^{T}\boldsymbol{g}(\boldsymbol{X}%
_{i},\widehat{\boldsymbol{\theta}}_{ETEL})\}}{%
{\displaystyle\sum\limits_{j=1}^{n}}
\exp\{\widehat{\boldsymbol{t}}_{ETEL}^{T}\boldsymbol{g}(\boldsymbol{X}%
_{j},\widehat{\boldsymbol{\theta}}_{ETEL})\}}\right)  \boldsymbol{g}%
(\boldsymbol{X}_{i},\widehat{\boldsymbol{\theta}}_{ETEL})
\end{align*}
and
\begin{align*}
\frac{\partial d_{\phi}\left(  \widehat{\boldsymbol{t}}_{ETEL},\boldsymbol{t}%
_{0}\right)  }{\partial\boldsymbol{t}_{0}}  &  =\left(
{\displaystyle\sum\limits_{j=1}^{n}}
\exp\{\boldsymbol{t}_{0}^{T}\boldsymbol{g}(\boldsymbol{X}_{j}%
,\boldsymbol{\theta}_{0})\}\right)  ^{-1}%
{\displaystyle\sum\limits_{i=1}^{n}}
\exp\{\boldsymbol{t}_{0}^{T}\boldsymbol{g}(\boldsymbol{X}_{i}%
,\boldsymbol{\theta}_{0})\}\\
&  \times\psi\left(  \frac{%
{\displaystyle\sum\limits_{j=1}^{n}}
\exp\{\boldsymbol{t}_{0}^{T}\boldsymbol{g}(\boldsymbol{X}_{j}%
,\boldsymbol{\theta}_{0})\}}{\exp\{\boldsymbol{t}_{0}^{T}\boldsymbol{g}%
(\boldsymbol{X}_{i},\boldsymbol{\theta}_{0})\}}\frac{\exp
\{\widehat{\boldsymbol{t}}_{ETEL}^{T}\boldsymbol{g}(\boldsymbol{X}%
_{i},\widehat{\boldsymbol{\theta}}_{ETEL})\}}{%
{\displaystyle\sum\limits_{j=1}^{n}}
\exp\{\widehat{\boldsymbol{t}}_{ETEL}^{T}\boldsymbol{g}(\boldsymbol{X}%
_{j},\widehat{\boldsymbol{\theta}}_{ETEL})\}}\right)  \boldsymbol{g}%
(\boldsymbol{X}_{i},\boldsymbol{\theta}_{0}),
\end{align*}
with $\psi(x)$ given by (\ref{psi}). Then,%
\begin{equation}
d_{\phi}(\widehat{\boldsymbol{t}}_{ETEL},\boldsymbol{t}_{0})=\mu_{\phi
}(\boldsymbol{\theta}_{0},\boldsymbol{\theta}^{\ast})+\boldsymbol{\bar{s}%
}_{1,S_{n}^{\phi}}^{T}(\boldsymbol{\theta}_{0},\boldsymbol{\theta}^{\ast
})\widehat{\boldsymbol{t}}_{ETEL}+\boldsymbol{\bar{s}}_{2,S_{n}^{\phi}}%
^{T}(\boldsymbol{\theta}_{0},\boldsymbol{\theta}^{\ast})(\boldsymbol{t}%
_{0}-\boldsymbol{\tau})+o(||\widehat{\boldsymbol{t}}_{ETEL}%
||)+o(||\boldsymbol{t}_{0}-\boldsymbol{\tau}||), \label{taylorS}%
\end{equation}
where%
\begin{align}
\boldsymbol{\bar{s}}_{1,S_{n}^{\phi}}(\boldsymbol{\theta}^{\ast}%
,\boldsymbol{\theta}_{0})  &  =\mathrm{E}_{F_{\boldsymbol{\theta}^{\ast}}%
}\left[  \phi^{\prime}\left(  \frac{\mathrm{E}_{F_{\boldsymbol{\theta}^{\ast}%
}}\left[  \exp\{\boldsymbol{\tau}^{T}\boldsymbol{g}(\boldsymbol{X}%
,\boldsymbol{\theta}_{0})\}\right]  }{\exp\{\boldsymbol{\tau}^{T}%
\boldsymbol{g}(\boldsymbol{X},\boldsymbol{\theta}_{0})\}}\right)
\boldsymbol{g}(\boldsymbol{X},\boldsymbol{\theta}^{\ast})\right]
,\label{q0}\\
\boldsymbol{\bar{s}}_{2,S_{n}^{\phi}}(\boldsymbol{\theta}^{\ast}%
,\boldsymbol{\theta}_{0})  &  =\mathrm{E}_{F_{\boldsymbol{\theta}^{\ast}}%
}^{-1}\left[  \exp\{\boldsymbol{\tau}^{T}\boldsymbol{g}(\boldsymbol{X}%
,\boldsymbol{\theta}_{0})\}\right]  \mathrm{E}_{F_{\boldsymbol{\theta}^{\ast}%
}}\left[  \exp\{\boldsymbol{\tau}^{T}\boldsymbol{g}(\boldsymbol{X}%
,\boldsymbol{\theta}_{0})\}\psi\left(  \frac{\mathrm{E}_{F_{\boldsymbol{\theta
}^{\ast}}}\left[  \exp\{\boldsymbol{\tau}^{T}\boldsymbol{g}(\boldsymbol{X}%
,\boldsymbol{\theta}_{0})\}\right]  }{\exp\{\boldsymbol{\tau}^{T}%
\boldsymbol{g}(\boldsymbol{X},\boldsymbol{\theta}_{0})\}}\right)
\boldsymbol{g}(\boldsymbol{X},\boldsymbol{\theta}_{0})\right]  ,\nonumber
\end{align}
are such that%
\begin{align*}
&  d_{\phi}(\boldsymbol{0}_{r},\boldsymbol{\tau}%
)\overset{P}{\underset{n\rightarrow\infty}{\longrightarrow}}\mu_{\phi
}(\boldsymbol{\theta}_{0},\boldsymbol{\theta}^{\ast}),\\
&  \left.  \frac{\partial d_{\phi}\left(  \widehat{\boldsymbol{t}}%
_{ETEL},\boldsymbol{t}_{0}\right)  }{\partial\widehat{\boldsymbol{t}}_{ETEL}%
}\right\vert _{\widehat{\boldsymbol{t}}_{ETEL}=\boldsymbol{0}_{r}%
}\overset{P}{\underset{n\rightarrow\infty}{\longrightarrow}}\boldsymbol{\bar
{s}}_{1,S_{n}^{\phi}}(\boldsymbol{\theta}^{\ast},\boldsymbol{\theta}_{0}),\\
&  \left.  \frac{\partial d_{\phi}\left(  \widehat{\boldsymbol{t}}%
_{ETEL},\boldsymbol{t}_{0}\right)  }{\partial\boldsymbol{t}_{0}}\right\vert
_{\boldsymbol{t}_{0}=\boldsymbol{\tau}}\overset{P}{\underset{n\rightarrow
\infty}{\longrightarrow}}\boldsymbol{\bar{s}}_{2,S_{n}^{\phi}}%
(\boldsymbol{\theta}^{\ast},\boldsymbol{\theta}_{0}).
\end{align*}
Denoting%
\[
\boldsymbol{h}(\boldsymbol{t}(\boldsymbol{\theta}))=\frac{1}{n}%
{\displaystyle\sum\limits_{i=1}^{n}}
\exp\{\boldsymbol{t}^{T}(\boldsymbol{\theta})\boldsymbol{g}(\boldsymbol{X}%
_{i},\boldsymbol{\theta})\}\boldsymbol{g}(\boldsymbol{X}_{i}%
,\boldsymbol{\theta}),
\]
the Taylor expansion of $\boldsymbol{h}(\boldsymbol{t}_{0})$ around
$\boldsymbol{t}_{0}=\boldsymbol{\tau}$ is equal to%
\[
\boldsymbol{0}_{r}=\boldsymbol{h}(\boldsymbol{\tau})+\left(  \frac{\partial
}{\partial\boldsymbol{t}_{0}^{T}}\left.  \boldsymbol{h}(\boldsymbol{t}%
_{0})\right\vert _{\boldsymbol{t}_{0}=\boldsymbol{\tau}}\right)
(\boldsymbol{t}_{0}-\boldsymbol{\tau})+o\left(  ||\boldsymbol{t}%
_{0}-\boldsymbol{\tau}||\boldsymbol{1}_{r}\right)  ,
\]
where
\begin{align*}
\boldsymbol{h}(\boldsymbol{\tau})  &  =\frac{1}{n}%
{\displaystyle\sum\limits_{i=1}^{n}}
\exp\{\boldsymbol{\tau}^{T}\boldsymbol{g}(\boldsymbol{X}_{i}%
,\boldsymbol{\theta}_{0})\}\boldsymbol{g}(\boldsymbol{X}_{i}%
,\boldsymbol{\theta}_{0}),\\
\frac{\partial}{\partial\boldsymbol{t}_{0}^{T}}\left.  \boldsymbol{h}%
(\boldsymbol{t}_{0})\right\vert _{\boldsymbol{t}_{0}=\boldsymbol{\tau}}  &
=\frac{1}{n}%
{\displaystyle\sum\limits_{i=1}^{n}}
\exp\{\boldsymbol{\tau}^{T}\boldsymbol{g}(\boldsymbol{X}_{i}%
,\boldsymbol{\theta}_{0})\}\boldsymbol{g}(\boldsymbol{X}_{i}%
,\boldsymbol{\theta}_{0})\boldsymbol{g}^{T}(\boldsymbol{X}_{i}%
,\boldsymbol{\theta}_{0})\\
&  =\mathrm{E}_{F_{\boldsymbol{\theta}^{\ast}}}\left[  \exp\{\boldsymbol{\tau
}^{T}\boldsymbol{g}(\boldsymbol{X},\boldsymbol{\theta}_{0})\}\boldsymbol{g}%
(\boldsymbol{X},\boldsymbol{\theta}_{0})\boldsymbol{g}^{T}(\boldsymbol{X}%
,\boldsymbol{\theta}_{0})\right]  +o_{p}(\boldsymbol{1}_{r\times r}),
\end{align*}
and from it the following relation is obtained%
\[
\boldsymbol{t}_{0}-\boldsymbol{\tau}=-\mathrm{E}_{F_{\boldsymbol{\theta}%
^{\ast}}}^{-1}\left[  \exp\{\boldsymbol{\tau}^{T}\boldsymbol{g}(\boldsymbol{X}%
,\boldsymbol{\theta}_{0})\}\boldsymbol{g}(\boldsymbol{X},\boldsymbol{\theta
}_{0})\boldsymbol{g}^{T}(\boldsymbol{X},\boldsymbol{\theta}_{0})\right]
\left(  \frac{1}{n}%
{\displaystyle\sum\limits_{i=1}^{n}}
\exp\{\boldsymbol{\tau}^{T}\boldsymbol{g}(\boldsymbol{X}_{i}%
,\boldsymbol{\theta}_{0})\}\boldsymbol{g}(\boldsymbol{X}_{i}%
,\boldsymbol{\theta}_{0})\right)  +o_{p}(\boldsymbol{1}_{r}).
\]
From (\ref{T12}) its follows%
\[
\widehat{\boldsymbol{t}}_{ETEL}=-\mathbf{R}\overline{\boldsymbol{g}}%
_{n}(\boldsymbol{X},\boldsymbol{\theta}^{\ast})+o_{p}(n^{-1/2}),
\]
and then%
\begin{align*}
&  \boldsymbol{\bar{s}}_{1,S_{n}^{\phi}}^{T}(\boldsymbol{\theta}^{\ast
},\boldsymbol{\theta}_{0})\widehat{\boldsymbol{t}}_{ETEL}+\boldsymbol{\bar{s}%
}_{2,S_{n}^{\phi}}^{T}(\boldsymbol{\theta}^{\ast},\boldsymbol{\theta}%
_{0})(\boldsymbol{t}_{0}-\boldsymbol{\tau})\\
&  =\boldsymbol{s}_{1,S_{n}^{\phi}}^{T}(\boldsymbol{\theta}^{\ast
},\boldsymbol{\theta}_{0})\overline{\boldsymbol{g}}_{n}(\boldsymbol{X}%
,\widehat{\boldsymbol{\theta}}_{ETEL})+\boldsymbol{s}_{2,S_{n}^{\phi}}%
^{T}(\boldsymbol{\theta}^{\ast},\boldsymbol{\theta}_{0})\left(  \frac{1}{n}%
{\displaystyle\sum\limits_{i=1}^{n}}
\exp\{\boldsymbol{\tau}^{T}\boldsymbol{g}(\boldsymbol{X}_{i}%
,\boldsymbol{\theta}_{0})\}\boldsymbol{g}(\boldsymbol{X}_{i}%
,\boldsymbol{\theta}_{0})\right) \\
&  =\frac{1}{n}%
{\displaystyle\sum\limits_{i=1}^{n}}
\boldsymbol{s}_{1,S_{n}^{\phi}}^{T}(\boldsymbol{\theta}^{\ast}%
,\boldsymbol{\theta}_{0})\boldsymbol{g}(\boldsymbol{X}_{i}%
,\widehat{\boldsymbol{\theta}}_{ETEL})+\boldsymbol{s}_{2,S_{n}^{\phi}}%
^{T}(\boldsymbol{\theta}^{\ast},\boldsymbol{\theta}_{0})\exp\{\boldsymbol{\tau
}^{T}\boldsymbol{g}(\boldsymbol{X}_{i},\boldsymbol{\theta}_{0}%
)\}\boldsymbol{g}(\boldsymbol{X}_{i},\boldsymbol{\theta}_{0})\\
&  =\frac{1}{n}%
{\displaystyle\sum\limits_{i=1}^{n}}
\boldsymbol{s}_{S_{n}^{\phi}}^{T}(\boldsymbol{\theta}^{\ast}%
,\boldsymbol{\theta}_{0})\widetilde{\boldsymbol{g}}(\boldsymbol{X}%
_{i},\widehat{\boldsymbol{\theta}}_{ETEL},\boldsymbol{\theta}_{0}),
\end{align*}
where $\boldsymbol{s}_{S_{n}^{\phi}}^{T}(\boldsymbol{\theta}^{\ast
},\boldsymbol{\theta}_{0})$ is (\ref{sS}),%
\[
\widetilde{\boldsymbol{g}}(\boldsymbol{X}_{i},\boldsymbol{\theta}^{\ast
},\boldsymbol{\theta}_{0})=%
\begin{pmatrix}
\boldsymbol{g}(\boldsymbol{X}_{i},\boldsymbol{\theta}^{\ast})\\
\exp\{\boldsymbol{\tau}^{T}\boldsymbol{g}(\boldsymbol{X}_{i}%
,\boldsymbol{\theta}_{0})\}\boldsymbol{g}(\boldsymbol{X}_{i}%
,\boldsymbol{\theta}_{0})
\end{pmatrix}
,
\]
and taking into account that%
\[
\mathrm{E}_{F_{\boldsymbol{\theta}^{\ast}}}\left[  \boldsymbol{s}_{S_{n}%
^{\phi}}^{T}(\boldsymbol{\theta}^{\ast},\boldsymbol{\theta}_{0}%
)\widetilde{\boldsymbol{g}}(\boldsymbol{X},\widehat{\boldsymbol{\theta}%
}_{ETEL},\boldsymbol{\theta}_{0})\right]  =\boldsymbol{s}_{S_{n}^{\phi}}%
^{T}(\boldsymbol{\theta}^{\ast},\boldsymbol{\theta}_{0})\mathrm{E}%
_{F_{\boldsymbol{\theta}^{\ast}}}\left[  \widetilde{\boldsymbol{g}%
}(\boldsymbol{X},\widehat{\boldsymbol{\theta}}_{ETEL},\boldsymbol{\theta}%
_{0})\right]  =0,
\]
we obtain in virtue of the Central Limit Theorem%
\[
\frac{\sqrt{n}}{\sqrt{\boldsymbol{s}_{S_{n}^{\phi}}^{T}(\boldsymbol{\theta
}^{\ast},\boldsymbol{\theta}_{0})Var_{F_{\boldsymbol{\theta}^{\ast}}}\left[
\widetilde{\boldsymbol{g}}(\boldsymbol{X}_{i},\boldsymbol{\theta}^{\ast
},\boldsymbol{\theta}_{0})\right]  \boldsymbol{s}_{S_{n}^{\phi}}%
(\boldsymbol{\theta}^{\ast},\boldsymbol{\theta}_{0})}}\left(  d_{\phi
}(\widehat{\boldsymbol{t}}_{ETEL},\boldsymbol{t}_{0})-\mu_{\phi}%
(\boldsymbol{\theta}_{0},\boldsymbol{\theta}^{\ast})\right)
\overset{\mathcal{L}}{\underset{n\rightarrow\infty}{\longrightarrow}%
}\mathcal{N}(0,1),
\]
which is equivalent to the theorems result.
\end{proof}

Let%
\begin{align*}
\beta_{T_{n}^{\phi}}(\boldsymbol{\theta}^{\ast})  &  =P(T_{n}^{\phi
}(\widehat{\boldsymbol{\theta}}_{ETEL},\boldsymbol{\theta}_{0})>\chi
_{p,\alpha}^{2}|\boldsymbol{\theta}^{\ast}),\\
\beta_{S_{n}^{\phi}}(\boldsymbol{\theta}^{\ast})  &  =P(S_{n}^{\phi
}(\widehat{\boldsymbol{\theta}}_{ETEL},\boldsymbol{\theta}_{0})>\chi
_{p,\alpha}^{2}|\boldsymbol{\theta}^{\ast}),
\end{align*}
be the exact power functions of $T_{n}^{\phi}(\widehat{\boldsymbol{\theta}%
}_{ETEL},\boldsymbol{\theta}_{0})$ and $S_{n}^{\phi}%
(\widehat{\boldsymbol{\theta}}_{ETEL},\boldsymbol{\theta}_{0})$\ respectively,
with respect to the asymptotic critical value of the test, at
$\boldsymbol{\theta}^{\ast}\neq\boldsymbol{\theta}_{0}$, for a significance
level $\alpha$. Notice that in practice, since the exact distributions of
$T_{n}^{\phi}(\widehat{\boldsymbol{\theta}}_{ETEL},\boldsymbol{\theta}_{0})$
and $S_{n}^{\phi}(\widehat{\boldsymbol{\theta}}_{ETEL},\boldsymbol{\theta}%
_{0})$\ are unknown, $\beta_{T_{n}^{\phi}}^{\ast}(\boldsymbol{\theta}^{\ast})$
and $\beta_{S_{n}^{\phi}}^{\ast}(\boldsymbol{\theta}^{\ast})$ are also
unknown. The following result provides an approximation for $\beta
_{T_{n}^{\phi}}(\boldsymbol{\theta}^{\ast})$ and $\beta_{S_{n}^{\phi}%
}(\boldsymbol{\theta}^{\ast})$.

\begin{remark}
\label{beta1}From Theorem \ref{th5}, we can present the approximation to the
asymptotic power $\beta_{T_{n}^{\phi}}(\boldsymbol{\theta}^{\ast})$, at
$\boldsymbol{\theta}^{\ast}\neq\boldsymbol{\theta}_{0}$, of the empirical
$\phi$-divergence test $T_{n}^{\phi}(\widehat{\boldsymbol{\theta}}%
_{ETEL},\boldsymbol{\theta}_{0})$ for a significance level $\alpha$, as%
\begin{equation}
\beta_{T_{n}^{\phi}}^{\ast}(\boldsymbol{\theta}^{\ast})=1-\Phi\left(
\nu_{T_{n}^{\phi}}(\boldsymbol{\theta}^{\ast},\boldsymbol{\boldsymbol{\theta
}_{0}})\right)  \simeq\beta_{T_{n}^{\phi}}(\boldsymbol{\theta}),
\label{beta1A}%
\end{equation}
where $\Phi(\cdot)$ is the standard normal distribution function and%
\[
\nu_{T_{n}^{\phi}}(\boldsymbol{\theta}^{\ast},\boldsymbol{\boldsymbol{\theta
}_{0}})=\frac{n^{1/2}}{\sqrt{\boldsymbol{s}_{T_{n}^{\phi}}^{T}%
(\boldsymbol{\theta}_{0},\boldsymbol{\theta}^{\ast})\boldsymbol{M}%
_{T_{n}^{\phi}}(\boldsymbol{\theta}_{0},\boldsymbol{\theta}^{\ast
})\boldsymbol{s}_{T_{n}^{\phi}}(\boldsymbol{\theta}_{0},\boldsymbol{\theta
}^{\ast})}}\left(  \frac{\phi^{\prime\prime}(1)\chi_{p,\alpha}^{2}}{2n}%
-\mu_{\phi}(\boldsymbol{\theta}_{0},\boldsymbol{\theta}^{\ast})\right)  .
\]
If some alternative $\boldsymbol{\theta}^{\ast}\neq\boldsymbol{\theta}_{0}$ is
the true parameter, then the probability of rejecting (\ref{H}) with the
rejection rule $T_{n}^{\phi}(\widehat{\boldsymbol{\theta}}_{ETEL}%
,\boldsymbol{\theta}_{0})>\chi_{p,\alpha}^{2}$ , for fixed significance level
$\alpha$, tends to one as $n\rightarrow\infty$. Thus, the test is consistent
in the sense of Fraser (1957). In a similar way, an approximation to the
asymptotic power function $\beta_{S_{n}^{\phi}}(\boldsymbol{\theta}^{\ast})$,
at $\boldsymbol{\theta}^{\ast}\neq\boldsymbol{\theta}_{0}$, for the empirical
$\phi$-divergence test $S_{n}^{\phi}(\widehat{\boldsymbol{\theta}}%
_{ETEL},\boldsymbol{\theta}_{0})$ can be obtained\ as%
\begin{equation}
\beta_{S_{n}^{\phi}}^{\ast}(\boldsymbol{\theta}^{\ast})=1-\Phi\left(
\nu_{S_{n}^{\phi}}(\boldsymbol{\theta}^{\ast},\boldsymbol{\boldsymbol{\theta
}_{0}})\right)  , \label{beta1B}%
\end{equation}
where%
\[
\nu_{S_{n}^{\phi}}(\boldsymbol{\theta}^{\ast},\boldsymbol{\boldsymbol{\theta
}_{0}})=\frac{n^{1/2}}{\sqrt{\boldsymbol{s}_{S_{n}^{\phi}}^{T}%
(\boldsymbol{\theta}_{0},\boldsymbol{\theta}^{\ast})\boldsymbol{M}%
_{S_{n}^{\phi}}(\boldsymbol{\theta}_{0},\boldsymbol{\theta}^{\ast
})\boldsymbol{s}_{S_{n}^{\phi}}(\boldsymbol{\theta}_{0},\boldsymbol{\theta
}^{\ast})}}\left(  \frac{\phi^{\prime\prime}(1)\chi_{p,\alpha}^{2}}{2n}%
-\mu_{\phi}(\boldsymbol{\theta}_{0},\boldsymbol{\theta}^{\ast})\right)  .
\]
From the parametric statistical inference, $\beta_{T_{n}^{\phi}}^{\ast
}(\boldsymbol{\theta}^{\ast})$ and $\beta_{S_{n}^{\phi}}^{\ast}%
(\boldsymbol{\theta}^{\ast})$ are known to be good approximations of
$\beta_{T_{n}^{\phi}}(\boldsymbol{\theta}^{\ast})$ and $\beta_{S_{n}^{\phi}%
}(\boldsymbol{\theta}^{\ast})$\ respectively (see for instance, Men\'{e}ndez
et al. (1998)). Notice that in practice, since $F$ is unknown, $\beta
_{T_{n}^{\phi}}^{\ast}(\boldsymbol{\theta}^{\ast})$ and $\beta_{S_{n}^{\phi}%
}^{\ast}(\boldsymbol{\theta}^{\ast})$ are also unknown. However, in practice
$\beta_{T_{n}^{\phi}}^{\ast}(\boldsymbol{\theta}^{\ast})$ and $\beta
_{S_{n}^{\phi}}^{\ast}(\boldsymbol{\theta}^{\ast})$ are consistently
estimated, by replacing expectations by sample means.
\end{remark}

To produce some less trivial asymptotic powers that are not all equal to $1$,
we can use a Pitman-type local analysis, as developed by Le Cam (1960), by
confining attention to $n^{1/2}$-neighborhoods of the true parameter values. A
key tool to get the asymptotic distribution of the statistic $T_{n}^{\phi
}(\widehat{\boldsymbol{\theta}}_{ETEL},\boldsymbol{\theta}_{0})$ (or
$S_{n}^{\phi}(\widehat{\boldsymbol{\theta}}_{ETEL},\boldsymbol{\theta}_{0})$)
under such a contiguous hypothesis is Le Calm's third lemma, as presented in
H\'{a}jek and Sid\'{a}k (1967). Instead of relying on these results, we
present in the following theorem a proof which is easy and direct to follow.
This proof is based on the results of Morales and Pardo (2001). Specifically,
we consider the power at contiguous alternative hypotheses of the form
\begin{equation}
H_{1,n}:\boldsymbol{\theta}_{n}=\boldsymbol{\theta}_{0}+n^{-1/2}%
\boldsymbol{\Delta}, \label{cont}%
\end{equation}
where $\boldsymbol{\Delta}$ is a fixed vector in $\mathbb{R}^{p}$ such that
$\boldsymbol{\theta}_{n}\in\Theta\subset\mathbb{R}^{p}$.

\begin{theorem}
\label{Th1B}Under Condition \ref{RC} and $H_{1,n}$ in (\ref{cont}), the
asymptotic distribution of the empirical $\phi$-divergence test statistics
$S_{n}^{\phi}(\widehat{\boldsymbol{\theta}}_{ETEL},\boldsymbol{\theta}_{0})$
and $T_{n}^{\phi}(\widehat{\boldsymbol{\theta}}_{ETEL},\boldsymbol{\theta}%
_{0})$ is a non-central chi-squared with $p$ degrees of freedom and
non-centrality parameter
\begin{equation}
\delta\boldsymbol{\left(  \boldsymbol{\theta}_{0}\right)  }=\boldsymbol{\Delta
}^{T}\boldsymbol{V}^{-1}\boldsymbol{\left(  \boldsymbol{\theta}_{0}\right)
\Delta}. \label{ncp}%
\end{equation}
i.e.%
\[
S_{n}^{\phi}(\widehat{\boldsymbol{\theta}}_{ETEL},\boldsymbol{\theta}%
_{0})\text{ (or }T_{n}^{\phi}(\widehat{\boldsymbol{\theta}}_{ETEL}%
,\boldsymbol{\theta}_{0})\text{)}\overset{\mathcal{L}}{\underset{n\rightarrow
\infty}{\longrightarrow}}\chi_{p}^{2}(\delta\boldsymbol{\left(
\boldsymbol{\theta}_{0}\right)  }),
\]
where $\boldsymbol{V\left(  \boldsymbol{\theta}_{0}\right)  }$\ was defined in
(\ref{V}).

\begin{proof}
We can write%
\[
\sqrt{n}(\widehat{\boldsymbol{\theta}}_{ETEL}-\boldsymbol{\theta}_{0}%
)=\sqrt{n}(\widehat{\boldsymbol{\theta}}_{ETEL}-\boldsymbol{\theta}_{n}%
)+\sqrt{n}\left(  \boldsymbol{\theta}_{n}-\boldsymbol{\theta}_{0}\right)
=\sqrt{n}(\widehat{\boldsymbol{\theta}}_{ETEL}-\boldsymbol{\theta}%
_{n})+\boldsymbol{\Delta.}%
\]
Under $H_{1,n}$, we have
\[
\sqrt{n}(\widehat{\boldsymbol{\theta}}_{ETEL}-\boldsymbol{\theta}%
_{n})\overset{\mathcal{L}}{\underset{n\rightarrow\infty}{\longrightarrow}%
}\mathcal{N}(\boldsymbol{0},\boldsymbol{V}\left(  \boldsymbol{\theta}%
_{0}\right)  )
\]
and
\[
\sqrt{n}(\widehat{\boldsymbol{\theta}}_{ETEL}-\boldsymbol{\theta}%
_{0})\overset{\mathcal{L}}{\underset{n\rightarrow\infty}{\longrightarrow}%
}\mathcal{N}(\boldsymbol{\Delta},\boldsymbol{V}(\boldsymbol{\theta}%
_{0}))\boldsymbol{.}%
\]
In Theorem \ref{Th1}, it has been shown that
\[
S_{n}^{\phi}(\widehat{\boldsymbol{\theta}}_{ETEL},\boldsymbol{\theta}%
_{0})=\left(  \boldsymbol{V}(\boldsymbol{\theta}_{0})^{-1/2}\sqrt
{n}(\widehat{\boldsymbol{\theta}}_{ETEL}-\boldsymbol{\theta}_{0})\right)
^{T}\boldsymbol{V}(\boldsymbol{\theta}_{0})^{-1/2}\sqrt{n}%
(\widehat{\boldsymbol{\theta}}_{ETEL}-\boldsymbol{\theta}_{0})+o_{p}(1).
\]
On the other hand, we have%
\[
\boldsymbol{V}(\boldsymbol{\theta}_{0})^{-1/2}\sqrt{n}%
(\widehat{\boldsymbol{\theta}}_{ETEL}-\boldsymbol{\theta}_{0}%
)\overset{\mathcal{L}}{\underset{n\rightarrow\infty}{\longrightarrow}%
}\mathcal{N}(\boldsymbol{V}(\boldsymbol{\theta}_{0})^{-1/2}\boldsymbol{\Delta
},\boldsymbol{I}_{p}).
\]
We thus obtain%
\[
S_{n}^{\phi}(\widehat{\boldsymbol{\theta}}_{ETEL},\boldsymbol{\theta}%
_{0})\overset{\mathcal{L}}{\underset{n\rightarrow\infty}{\longrightarrow}}%
\chi_{p}^{2}(\delta(\boldsymbol{\boldsymbol{\theta}_{0}})),
\]
with $\delta(\boldsymbol{\boldsymbol{\theta}_{0}})$ as in (\ref{ncp}). A
similar procedure can be followed for the proof of $T_{n}^{\phi}%
(\widehat{\boldsymbol{\theta}}_{ETEL},\boldsymbol{\theta}_{0})$.
\end{proof}
\end{theorem}

\section{Robustness of empirical $\phi$-divergence test
statistics\label{robustness}}

In Robust Statistics, two concepts of robustness can be distinguished,
robustness with respect to contamination and robustness with respect to model
misspecification. We shall understand misspecification in the sense that
(\ref{T1b}) is not verified for any $\boldsymbol{\theta\in}\Theta$, in
particular there is misspecification for the null hypothesis in (\ref{H}) if%
\[
\left\Vert \mathrm{E}_{F}\left[  \boldsymbol{g}(\boldsymbol{X}%
,\boldsymbol{\theta}_{0})\right]  \right\Vert >0\text{.}%
\]
For brevity, in the sequel $\mathrm{E}_{F}[\cdot]$ is denoted by
$\mathrm{E}[\cdot]$.

It is well-known (see Imbens et al. (1998)) that the estimating equation with
respect to $\boldsymbol{\theta}$\ for the EL and ET estimators are given by%
\[
\sum_{i=1}^{n}\rho_{\ell}\left(  \boldsymbol{x}_{i}%
,\widehat{\boldsymbol{\theta}}_{\ell},\boldsymbol{t}_{\ell}%
(\widehat{\boldsymbol{\theta}}_{\ell})\right)  =0,\qquad\ell\in\{EL,ET\},
\]
with%
\begin{align}
\rho_{EL}\left(  \boldsymbol{x},\boldsymbol{\theta},\boldsymbol{t}%
_{EL}(\boldsymbol{\theta})\right)   &  =\frac{\boldsymbol{t}_{EL}%
^{T}(\boldsymbol{\theta})\boldsymbol{G}_{\boldsymbol{x}}(\boldsymbol{\theta}%
)}{1+\boldsymbol{t}_{EL}^{T}(\boldsymbol{\theta})\boldsymbol{g}(\boldsymbol{x}%
,\boldsymbol{\theta})},\label{roEL}\\
\rho_{ET}\left(  \boldsymbol{x},\boldsymbol{\theta},\boldsymbol{t}%
_{EL}(\boldsymbol{\theta})\right)   &  =\boldsymbol{t}_{ET}^{T}%
(\boldsymbol{\theta})\boldsymbol{G}_{\boldsymbol{x}}(\boldsymbol{\theta}%
)\exp\left\{  \boldsymbol{t}_{ET}^{T}(\boldsymbol{\theta})\boldsymbol{g}%
(\boldsymbol{x},\boldsymbol{\theta})\right\}  . \label{roET}%
\end{align}
In relation to the ETEL estimators, from Theorem 2 of Schennach (2007) the
following estimating equation with respect to $\boldsymbol{\theta}$\ is
obtained%
\[
\sum_{i=1}^{n}\rho_{ETEL}\left(  \boldsymbol{x}_{i}%
,\widehat{\boldsymbol{\theta}}_{ETEL},\boldsymbol{t}_{ET}%
(\widehat{\boldsymbol{\theta}}_{ETEL})\right)  =0,
\]
with%
\begin{align}
\rho_{ETEL}\left(  \boldsymbol{x},\boldsymbol{\theta},\boldsymbol{t}%
_{ET}(\boldsymbol{\theta})\right)   &  =\boldsymbol{t}_{ET}^{T}%
(\boldsymbol{\theta})\boldsymbol{G}_{\boldsymbol{x}}(\boldsymbol{\theta
})\left(  \exp\left\{  \boldsymbol{t}_{ET}^{T}(\boldsymbol{\theta
})\boldsymbol{g}(\boldsymbol{x},\boldsymbol{\theta})\right\}  -\overline{\exp
}_{ET}(\boldsymbol{\theta})\right)  ,\label{roETEL}\\
\overline{\exp}_{ET}(\boldsymbol{\theta})  &  =\frac{1}{n}%
{\textstyle\sum\nolimits_{j=1}^{n}}
\exp\left\{  \boldsymbol{t}_{ET}^{T}(\boldsymbol{\theta})\boldsymbol{g}%
(\boldsymbol{x}_{j},\boldsymbol{\theta})\right\}  ,\qquad\boldsymbol{x}%
\in\{\boldsymbol{x}_{j}\}_{j=1}^{n}.\nonumber
\end{align}
The influence functions for the three types of estimators, EL, ET, ETEL, are
proportional to the $\rho_{\ell}\left(  \boldsymbol{x},\boldsymbol{\theta
},\boldsymbol{t}_{\ell}(\boldsymbol{\theta})\right)  $ function, for $\ell
\in\{EL,ET,ETEL\}$, respectively, given in (\ref{roEL})-(\ref{roETEL}),%
\[
\mathcal{IF(}\boldsymbol{x},\widehat{\boldsymbol{\theta}}_{\ell}%
,F_{n,\boldsymbol{\theta}})\propto\rho_{\ell}\left(  \boldsymbol{x}%
,\widehat{\boldsymbol{\theta}}_{\ell},\boldsymbol{t}_{\ell}%
(\widehat{\boldsymbol{\theta}}_{\ell})\right)  ,
\]
where $\boldsymbol{t}_{ETEL}(\boldsymbol{\theta})=\boldsymbol{t}%
_{ET}(\boldsymbol{\theta})$. Evaluating $\rho_{EL}\left(  \boldsymbol{x}%
,\widehat{\boldsymbol{\theta}}_{EL},\boldsymbol{t}_{EL}%
(\widehat{\boldsymbol{\theta}}_{EL})\right)  $ at perturbations of
$\boldsymbol{t}_{EL}(\widehat{\boldsymbol{\theta}}_{EL})\neq\boldsymbol{0}%
_{r}$, it can become unbounded even if $\boldsymbol{g}(\boldsymbol{x}%
,\boldsymbol{\theta})$\ is bounded, i.e. the influence function of
$\widehat{\boldsymbol{\theta}}_{EL}$\ can be unbounded. This is in contrast
with the influence function of $\widehat{\boldsymbol{\theta}}_{ET}$ and
$\widehat{\boldsymbol{\theta}}_{ETEL}$, since $\rho_{ET}\left(  \boldsymbol{x}%
,\widehat{\boldsymbol{\theta}}_{EL},\boldsymbol{t}_{EL}%
(\widehat{\boldsymbol{\theta}}_{ET})\right)  $ and $\rho_{ETEL}\left(
\boldsymbol{x},\widehat{\boldsymbol{\theta}}_{EL},\boldsymbol{t}%
_{EL}(\widehat{\boldsymbol{\theta}}_{ETEL})\right)  $ are affected to a much
less extent by perturbations of $\boldsymbol{t}_{ET}%
(\widehat{\boldsymbol{\theta}}_{\ell})$, $\ell\in\{ET,ETEL\}$, respectively.
At the limiting values of the estimators, $\widehat{\boldsymbol{\theta}}%
_{\ell}\underset{n\rightarrow\infty}{\overset{P}{\longrightarrow}%
}\boldsymbol{\theta}_{0}$, $\boldsymbol{t}_{\ell}(\widehat{\boldsymbol{\theta
}}_{\ell})\underset{n\rightarrow\infty}{\overset{P}{\longrightarrow}%
}\boldsymbol{0}_{r}$, for $\ell\in\{EL,ET,ETEL\}$, respectively, the influence
functions for the three types of estimators, are identical,%
\[
\mathcal{IF(}\boldsymbol{x},\widehat{\boldsymbol{\theta}}_{\ell}%
,F_{n,\boldsymbol{\theta}_{0}})=\boldsymbol{V}\left(  \boldsymbol{\theta}%
_{0}\right)  \boldsymbol{S}_{12}^{T}\left(  \boldsymbol{\theta}_{0}\right)
\boldsymbol{S}_{11}^{-1}\left(  \boldsymbol{\theta}_{0}\right)  \boldsymbol{g}%
(\boldsymbol{x},\boldsymbol{\theta}_{0}),
\]
reflecting the first order equivalence of the estimators (for a detailed proof
see Lemma 1 in Balakrishnan et al. (2015)).

Let $\boldsymbol{T}(\bullet)$ be the functional associated the ETEL estimator
of $\boldsymbol{\theta}$, i.e.%
\[
\boldsymbol{T}(F_{n,\boldsymbol{\theta}})=\widehat{\boldsymbol{\theta}}%
_{ETEL},\qquad\boldsymbol{T}(F_{n,\boldsymbol{\theta}_{0}})=\boldsymbol{\theta
}_{0},
\]
and the test-statistic $S_{n}^{\phi}(\widehat{\boldsymbol{\theta}}%
_{ETEL},\boldsymbol{\theta}_{0})$,\ given in (\ref{F2}), defined now through
its functional%
\[
S_{n}^{\phi}(F_{n,\boldsymbol{\theta}})=\frac{2n}{\phi^{\prime\prime}%
(1)}D_{\phi}\left(  \boldsymbol{p}_{ET}\left(  \boldsymbol{T}%
(F_{n,\boldsymbol{\theta}})\right)  ,\boldsymbol{p}_{ET}\left(
\boldsymbol{\theta}_{0}\right)  \right)  =\frac{2n}{\phi^{\prime\prime}(1)}%
{\displaystyle\sum\limits_{i=1}^{n}}
p_{ET,i}\left(  \boldsymbol{\theta}_{0}\right)  \phi\left(  \frac
{p_{ET,i}\left(  \boldsymbol{T}(F_{n,\boldsymbol{\theta}})\right)  }%
{p_{ET,i}\left(  \boldsymbol{\theta}_{0}\right)  }\right)  .
\]

\begin{theorem}
\label{Thh}The first and second order influence functions of $S_{n}^{\phi
}(F_{n,\boldsymbol{\theta}})$ are%
\[
\mathcal{IF(}\boldsymbol{x},S_{n}^{\phi},F_{n,\boldsymbol{\theta}}%
)=\frac{\partial}{\partial\boldsymbol{\theta}^{T}}\left.  S_{n}^{\phi
}(F_{n,\boldsymbol{\theta}})\right\vert _{\boldsymbol{\theta}=\boldsymbol{T}%
(F_{n,\boldsymbol{\theta}})}\mathcal{IF(}\boldsymbol{x}%
,\widehat{\boldsymbol{\theta}}_{ETEL},F_{n,\boldsymbol{\theta}}),
\]
and%
\begin{align*}
\mathcal{IF}_{2}\mathcal{(}\boldsymbol{x},S_{n}^{\phi},F_{n,\boldsymbol{\theta
}})  &  =\frac{2n}{\phi^{\prime\prime}(1)}\mathcal{IF}^{T}\mathcal{(}%
\boldsymbol{x},\widehat{\boldsymbol{\theta}}_{ETEL},F_{n,\boldsymbol{\theta}%
})\sum\limits_{i=1}^{n}\phi^{\prime\prime}\left(  \frac{p_{ET,i}\left(
\boldsymbol{T}(F_{n,\boldsymbol{\theta}})\right)  }{p_{ET,i}\left(
\boldsymbol{\theta}_{0}\right)  }\right)  \frac{1}{p_{ET,i}\left(
\boldsymbol{\theta}_{0}\right)  }\\
&  \times\left(  \frac{\partial}{\partial\boldsymbol{T}%
(F_{n,\boldsymbol{\theta}})}p_{ET,i}\left(  \boldsymbol{T}%
(F_{n,\boldsymbol{\theta}})\right)  \right)  \left(  \frac{\partial}%
{\partial\boldsymbol{T}^{T}(F_{n,\boldsymbol{\theta}})}p_{ET,i}\left(
\boldsymbol{T}(F_{n,\boldsymbol{\theta}})\right)  \right)  \mathcal{IF(}%
\boldsymbol{x},\widehat{\boldsymbol{\theta}}_{ETEL},F_{n,\boldsymbol{\theta}%
})\\
&  +\frac{2n}{\phi^{\prime\prime}(1)}\sum\limits_{i=1}^{n}\phi^{\prime}\left(
\frac{p_{ET,i}\left(  \boldsymbol{T}(F_{n,\boldsymbol{\theta}})\right)
}{p_{ET,i}\left(  \boldsymbol{\theta}_{0}\right)  }\right)  \left(
\frac{\partial}{\partial\boldsymbol{T}^{T}(F_{n,\boldsymbol{\theta}})}%
p_{ET,i}\left(  \boldsymbol{T}(F_{n,\boldsymbol{\theta}})\right)  \right)
\mathcal{IF}_{2}\mathcal{(}\boldsymbol{x},\widehat{\boldsymbol{\theta}}%
_{ETEL},F_{n,\boldsymbol{\theta}}).
\end{align*}

\end{theorem}

\begin{proof}
Let%
\[
F_{n,\varepsilon,\boldsymbol{\theta}}=(1-\varepsilon)F_{n,\boldsymbol{\theta}%
}+\varepsilon\delta_{\boldsymbol{x}},\qquad\delta_{\boldsymbol{x}%
}(\boldsymbol{s})=\left\{
\begin{array}
[c]{cc}%
0, & \boldsymbol{s}<\boldsymbol{x},\\
1, & \boldsymbol{s}\geq\boldsymbol{x}.
\end{array}
\right.  ,
\]
the $\varepsilon$-perturbation of $F_{n,\boldsymbol{\theta}}$\ at
$\boldsymbol{x}$. The first and second order influence functions of
$S_{n}^{\phi}(F_{n,\boldsymbol{\theta}})$ are defined as%
\begin{align*}
\mathcal{IF(}\boldsymbol{x},S_{n}^{\phi},F_{n,\boldsymbol{\theta}})  &
=\frac{\partial}{\partial\varepsilon}\left.  S_{n}^{\phi}(F_{n,\varepsilon
,\boldsymbol{\theta}})\right\vert _{\varepsilon=0}\\
&  =\frac{2n}{\phi^{\prime\prime}(1)}\sum\limits_{i=1}^{n}\phi^{\prime}\left(
\frac{p_{ET,i}\left(  \boldsymbol{T}(F_{n,\boldsymbol{\theta}})\right)
}{p_{ET,i}\left(  \boldsymbol{\theta}_{0}\right)  }\right)  \left.
\frac{\partial}{\partial\varepsilon}p_{ET,i}\left(  \boldsymbol{T}%
(F_{n,\varepsilon,\boldsymbol{\theta}})\right)  \right\vert _{\varepsilon=0}\\
&  =\frac{2n}{\phi^{\prime\prime}(1)}\sum\limits_{i=1}^{n}\phi^{\prime}\left(
\frac{p_{ET,i}\left(  \boldsymbol{T}(F_{n,\boldsymbol{\theta}})\right)
}{p_{ET,i}\left(  \boldsymbol{\theta}_{0}\right)  }\right)  \left(
\frac{\partial}{\partial\boldsymbol{T}^{T}(F_{n,\boldsymbol{\theta}})}%
p_{ET,i}\left(  \boldsymbol{T}(F_{n,\boldsymbol{\theta}})\right)  \right)
\left(  \left.  \frac{\partial}{\partial\varepsilon}\boldsymbol{T}%
(F_{n,\varepsilon,\boldsymbol{\theta}})\right\vert _{\varepsilon=0}\right) \\
&  =\frac{\partial}{\partial\boldsymbol{\theta}^{T}}\left.  S_{n}^{\phi
}(F_{n,\boldsymbol{\theta}})\right\vert _{\boldsymbol{\theta}=\boldsymbol{T}%
(F_{n,\boldsymbol{\theta}})}\mathcal{IF(}\boldsymbol{x}%
,\widehat{\boldsymbol{\theta}}_{ETEL},F_{n,\boldsymbol{\theta}}),
\end{align*}
and%
\begin{align*}
\mathcal{IF}_{2}\mathcal{(}\boldsymbol{x},S_{n}^{\phi},F_{n,\boldsymbol{\theta
}})  &  =\frac{\partial^{2}}{\partial\varepsilon^{2}}\left.  S_{n}^{\phi
}(F_{n,\varepsilon,\boldsymbol{\theta}})\right\vert _{\varepsilon=0}\\
&  =\frac{2n}{\phi^{\prime\prime}(1)}\sum\limits_{i=1}^{n}\phi^{\prime\prime
}\left(  \frac{p_{ET,i}\left(  \boldsymbol{T}(F_{n,\boldsymbol{\theta}%
})\right)  }{p_{ET,i}\left(  \boldsymbol{\theta}_{0}\right)  }\right)
\frac{\left(  \left.  \frac{\partial}{\partial\varepsilon}p_{ET,i}\left(
\boldsymbol{T}(F_{n,\varepsilon,\boldsymbol{\theta}})\right)  \right\vert
_{\varepsilon=0}\right)  ^{2}}{p_{ET,i}\left(  \boldsymbol{\theta}_{0}\right)
}\\
&  +\frac{2n}{\phi^{\prime\prime}(1)}\sum\limits_{i=1}^{n}\phi^{\prime}\left(
\frac{p_{ET,i}\left(  \boldsymbol{T}(F_{n,\boldsymbol{\theta}})\right)
}{p_{ET,i}\left(  \boldsymbol{\theta}_{0}\right)  }\right)  \left.
\frac{\partial^{2}}{\partial\varepsilon^{2}}p_{ET,i}\left(  \boldsymbol{T}%
(F_{n,\boldsymbol{\theta}})\right)  \right\vert _{\varepsilon=0}\\
&  =\frac{2n}{\phi^{\prime\prime}(1)}\sum\limits_{i=1}^{n}\frac{1}%
{p_{ET,i}\left(  \boldsymbol{\theta}_{0}\right)  }\phi^{\prime\prime}\left(
\frac{p_{ET,i}\left(  \boldsymbol{T}(F_{n,\boldsymbol{\theta}})\right)
}{p_{ET,i}\left(  \boldsymbol{\theta}_{0}\right)  }\right)  \left(  \left.
\frac{\partial}{\partial\varepsilon}\boldsymbol{T}^{T}(F_{n,\varepsilon
,\boldsymbol{\theta}})\right\vert _{\varepsilon=0}\right) \\
&  \times\left(  \frac{\partial}{\partial\boldsymbol{T}%
(F_{n,\boldsymbol{\theta}})}p_{ET,i}\left(  \boldsymbol{T}%
(F_{n,\boldsymbol{\theta}})\right)  \right)  \left(  \frac{\partial}%
{\partial\boldsymbol{T}^{T}(F_{n,\boldsymbol{\theta}})}p_{ET,i}\left(
\boldsymbol{T}(F_{n,\boldsymbol{\theta}})\right)  \right)  \left(  \left.
\frac{\partial}{\partial\varepsilon}\boldsymbol{T}(F_{n,\varepsilon
,\boldsymbol{\theta}})\right\vert _{\varepsilon=0}\right) \\
&  +\frac{2n}{\phi^{\prime\prime}(1)}\sum\limits_{i=1}^{n}\phi^{\prime}\left(
\frac{p_{ET,i}\left(  \boldsymbol{T}(F_{n,\boldsymbol{\theta}})\right)
}{p_{ET,i}\left(  \boldsymbol{\theta}_{0}\right)  }\right)  \left(
\frac{\partial}{\partial\boldsymbol{T}^{T}(F_{n,\boldsymbol{\theta}})}%
p_{ET,i}\left(  \boldsymbol{T}(F_{n,\boldsymbol{\theta}})\right)  \right)
\left(  \left.  \frac{\partial^{2}}{\partial\varepsilon^{2}}\boldsymbol{T}%
(F_{n,\varepsilon,\boldsymbol{\theta}})\right\vert _{\varepsilon=0}\right) \\
&  =\frac{2n}{\phi^{\prime\prime}(1)}\mathcal{IF}^{T}\mathcal{(}%
\boldsymbol{x},\widehat{\boldsymbol{\theta}}_{ETEL},F_{n,\boldsymbol{\theta}%
})\sum\limits_{i=1}^{n}\phi^{\prime\prime}\left(  \frac{p_{ET,i}\left(
\boldsymbol{T}(F_{n,\boldsymbol{\theta}})\right)  }{p_{ET,i}\left(
\boldsymbol{\theta}_{0}\right)  }\right)  \frac{1}{p_{ET,i}\left(
\boldsymbol{\theta}_{0}\right)  }\\
&  \times\left(  \frac{\partial}{\partial\boldsymbol{T}%
(F_{n,\boldsymbol{\theta}})}p_{ET,i}\left(  \boldsymbol{T}%
(F_{n,\boldsymbol{\theta}})\right)  \right)  \left(  \frac{\partial}%
{\partial\boldsymbol{T}^{T}(F_{n,\boldsymbol{\theta}})}p_{ET,i}\left(
\boldsymbol{T}(F_{n,\boldsymbol{\theta}})\right)  \right)  \mathcal{IF(}%
\boldsymbol{x},\widehat{\boldsymbol{\theta}}_{ETEL},F_{n,\boldsymbol{\theta}%
})\\
&  +\frac{2n}{\phi^{\prime\prime}(1)}\sum\limits_{i=1}^{n}\phi^{\prime}\left(
\frac{p_{ET,i}\left(  \boldsymbol{T}(F_{n,\boldsymbol{\theta}})\right)
}{p_{ET,i}\left(  \boldsymbol{\theta}_{0}\right)  }\right)  \left(
\frac{\partial}{\partial\boldsymbol{T}^{T}(F_{n,\boldsymbol{\theta}})}%
p_{ET,i}\left(  \boldsymbol{T}(F_{n,\boldsymbol{\theta}})\right)  \right)
\mathcal{IF}_{2}\mathcal{(}\boldsymbol{x},\widehat{\boldsymbol{\theta}}%
_{ETEL},F_{n,\boldsymbol{\theta}}),
\end{align*}

\end{proof}

\begin{corollary}
\label{Corr}Under the null hypothesis of the test (\ref{H}), the first and
second order influence functions of the test-statistic $S_{n}^{\phi
}(\widehat{\boldsymbol{\theta}}_{ETEL},\boldsymbol{\theta}_{0})$ are given by%
\begin{align*}
\mathcal{IF(}\boldsymbol{x},S_{n}^{\phi},F_{n,\boldsymbol{\theta}_{0}})  &
=\frac{\partial}{\partial\boldsymbol{\theta}^{T}}\left.  S_{n}^{\phi
}(F_{n,\boldsymbol{\theta}})\right\vert _{\boldsymbol{\theta}%
=\boldsymbol{\theta}_{0}}\mathcal{IF(}\boldsymbol{x}%
,\widehat{\boldsymbol{\theta}}_{ETEL},F_{n,\boldsymbol{\theta}_{0}})=0,\\
\mathcal{IF}_{2}\mathcal{(}\boldsymbol{x},S_{n}^{\phi},F_{n,\boldsymbol{\theta
}_{0}})  &  =\mathcal{IF}^{T}\mathcal{(}\boldsymbol{x}%
,\widehat{\boldsymbol{\theta}}_{ETEL},F_{n,\boldsymbol{\theta}_{0}}%
)\frac{\partial^{2}}{\partial\boldsymbol{\theta}\partial\boldsymbol{\theta
}^{T}}\left.  S_{n}^{\phi}(F_{n,\boldsymbol{\theta}})\right\vert
_{\boldsymbol{\theta}=\boldsymbol{\theta}_{0}}\mathcal{IF(}\boldsymbol{x}%
,\widehat{\boldsymbol{\theta}}_{ETEL},F_{n,\boldsymbol{\theta}_{0}}).
\end{align*}
In particular, for large samples%
\begin{align}
\mathcal{IF}_{2}\mathcal{(}\boldsymbol{x},S_{n}^{\phi},F_{n,\boldsymbol{\theta
}_{0}})  &  =\mathcal{IF}^{T}\mathcal{(}\boldsymbol{x}%
,\widehat{\boldsymbol{\theta}}_{ETEL},F_{n,\boldsymbol{\theta}_{0}%
})\boldsymbol{V}^{-1}\left(  \boldsymbol{\theta}_{0}\right)  \mathcal{IF(}%
\boldsymbol{x},\widehat{\boldsymbol{\theta}}_{ETEL},F_{n,\boldsymbol{\theta
}_{0}})\nonumber\\
&  =\boldsymbol{g}^{T}(\boldsymbol{x},\boldsymbol{\theta}_{0})\boldsymbol{S}%
_{11}^{-1}\left(  \boldsymbol{\theta}_{0}\right)  \boldsymbol{S}_{12}\left(
\boldsymbol{\theta}_{0}\right)  \boldsymbol{V}\left(  \boldsymbol{\theta}%
_{0}\right)  \boldsymbol{S}_{12}^{T}\left(  \boldsymbol{\theta}_{0}\right)
\boldsymbol{S}_{11}^{-1}\left(  \boldsymbol{\theta}_{0}\right)  \boldsymbol{g}%
(\boldsymbol{x},\boldsymbol{\theta}_{0}). \label{IF2}%
\end{align}

\end{corollary}

\begin{proof}
Both equalities are obtained taking into account%
\[
\frac{\partial}{\partial\boldsymbol{\theta}^{T}}\left.  S_{n}^{\phi
}(F_{n,\boldsymbol{\theta}})\right\vert _{\boldsymbol{\theta}%
=\boldsymbol{\theta}_{0}}=\frac{2n\phi^{\prime}\left(  1\right)  }%
{\phi^{\prime\prime}(1)}\sum\limits_{i=1}^{n}\frac{\partial}{\partial
\boldsymbol{\theta}^{T}}\left.  p_{ET,i}\left(  \boldsymbol{\theta}\right)
\right\vert _{\boldsymbol{\theta}=\boldsymbol{T}(F_{n,\boldsymbol{\theta}_{0}%
})=\boldsymbol{\theta}_{0}}=\boldsymbol{0}_{p}^{T},
\]
since $\phi^{\prime}\left(  1\right)  =0$, and%
\begin{align*}
&  \mathcal{IF}^{T}\mathcal{(}\boldsymbol{x},\widehat{\boldsymbol{\theta}%
}_{ETEL},F_{n,\boldsymbol{\theta}_{0}})\frac{\partial^{2}}{\partial
\boldsymbol{\theta}\partial\boldsymbol{\theta}^{T}}\left.  S_{n}^{\phi
}(F_{n,\boldsymbol{\theta}})\right\vert _{\boldsymbol{\theta}=\boldsymbol{T}%
(F_{n,\boldsymbol{\theta}_{0}})=\boldsymbol{\theta}_{0}}\mathcal{IF(}%
\boldsymbol{x},\widehat{\boldsymbol{\theta}}_{ETEL},F_{n,\boldsymbol{\theta
}_{0}})\\
&  =\mathcal{IF}^{T}\mathcal{(}\boldsymbol{x},\widehat{\boldsymbol{\theta}%
}_{ETEL},F_{n,\boldsymbol{\theta}_{0}})\frac{2n}{\phi^{\prime\prime}(1)}%
\sum\limits_{i=1}^{n}\phi^{\prime\prime}\left(  1\right)  \frac{1}%
{p_{ET,i}\left(  \boldsymbol{\theta}_{0}\right)  }\left.  \frac{\partial
}{\partial\boldsymbol{\theta}}p_{ET,i}\left(  \boldsymbol{\theta}\right)
\right\vert _{\boldsymbol{\theta}=\boldsymbol{\theta}_{0}}\left.
\frac{\partial}{\partial\boldsymbol{\theta}^{T}}p_{ET,i}\left(
\boldsymbol{\theta}\right)  \right\vert _{\boldsymbol{\theta}%
=\boldsymbol{\theta}_{0}}\mathcal{IF(}\boldsymbol{x}%
,\widehat{\boldsymbol{\theta}}_{ETEL},F_{n,\boldsymbol{\theta}_{0}})\\
&  +\frac{2n}{\phi^{\prime\prime}(1)}\sum\limits_{i=1}^{n}\phi^{\prime}\left(
1\right)  \left.  \frac{\partial}{\partial\boldsymbol{\theta}^{T}}%
p_{ET,i}\left(  \boldsymbol{\theta}\right)  \right\vert _{\boldsymbol{\theta
}=\boldsymbol{\theta}_{0}}\mathcal{IF}_{2}\mathcal{(}\boldsymbol{x}%
,\widehat{\boldsymbol{\theta}}_{ETEL},F_{n,\boldsymbol{\theta}}).
\end{align*}
Since%
\begin{align*}
\frac{\partial^{2}}{\partial\boldsymbol{\theta}\partial\boldsymbol{\theta}%
^{T}}\left.  S_{n}^{\phi}(F_{n,\boldsymbol{\theta}})\right\vert
_{\boldsymbol{\theta}=\boldsymbol{\theta}_{0}}  &  =2n\sum\limits_{i=1}%
^{n}\left.  \frac{\partial}{\partial\boldsymbol{\theta}}\log p_{ET,i}\left(
\boldsymbol{\theta}\right)  \right\vert _{\boldsymbol{\theta}%
=\boldsymbol{\theta}_{0}}p_{ET,i}\left(  \boldsymbol{\theta}_{0}\right)
\left.  \frac{\partial}{\partial\boldsymbol{\theta}^{T}}\log p_{ET,i}\left(
\boldsymbol{\theta}\right)  \right\vert _{\boldsymbol{\theta}%
=\boldsymbol{\theta}_{0}}\\
&  \mathcal{=}2\boldsymbol{V}^{-1}\left(  \boldsymbol{\theta}_{0}\right)
+o_{p}(1),
\end{align*}
an alternative expression for the second order influence function, for large
sample sizes, is (\ref{IF2}).
\end{proof}

Notice that $\frac{\partial^{2}}{\partial\boldsymbol{\theta}\partial
\boldsymbol{\theta}^{T}}\left.  S_{n}^{\phi}(F_{n,\boldsymbol{\theta}%
})\right\vert _{\boldsymbol{\theta}=\boldsymbol{\theta}_{0}}$ is the same for
any $\phi$\ function and plugging any estimator into $S_{n}^{\phi}$, either
EL, ET or ETEL, $\mathcal{IF}_{2}\mathcal{(}\boldsymbol{x},S_{n}^{\phi
},F_{n,\boldsymbol{\theta}_{0}})$ remains unchanged.

A similar results of Theorem \ref{Thh} and Corollary \ref{Corr} can be
enuntiated for the other family of test-statistics, $T_{n}^{\phi
}(F_{n,\boldsymbol{\theta}})$.

Let $\boldsymbol{\theta}_{\ast,ETEL}$ denote the ETEL's pseudo-true value
associated with the misspecified model, i.e.%
\begin{align*}
&  \boldsymbol{\theta}_{\ast,ETEL}=\arg\min\log\mathrm{E}\left[  \exp\left\{
\boldsymbol{t}^{T}(\boldsymbol{\theta})\left(  \boldsymbol{g}(\boldsymbol{X}%
,\boldsymbol{\theta})-\mathrm{E}\left[  \boldsymbol{g}(\boldsymbol{X}%
,\boldsymbol{\theta})\right]  \right)  \right\}  \right]  ,\\
&  \text{s.t. }\mathrm{E}\left[  \exp\left\{  \boldsymbol{t}^{T}%
(\boldsymbol{\theta})\boldsymbol{g}(\boldsymbol{X},\boldsymbol{\theta
})\right\}  \boldsymbol{g}(\boldsymbol{X},\boldsymbol{\theta})\right]
=\boldsymbol{0}_{r}.
\end{align*}
The ETEL's pseudo-true value can be interpreted as the best approximation to
the true value, according to the ETEL's estimation method.

\begin{condition}
\label{RC2}We shall assume the following regularity conditions (Schennach,
2007):\newline i) There exists a neighborhood of $\boldsymbol{\theta}%
_{\ast,ETEL}$ in which $\frac{\partial\boldsymbol{G}_{\boldsymbol{X}%
}(\boldsymbol{\theta})}{\partial\boldsymbol{\theta}}$ is continuous and
$\left\Vert \frac{\partial\boldsymbol{G}_{\boldsymbol{X}}(\boldsymbol{\theta
})}{\partial\boldsymbol{\theta}}\right\Vert $ is bounded by some integrable
function of $\boldsymbol{X}$;\newline ii) $\mathrm{E}\left[  \sup
_{\boldsymbol{\theta\in\Theta}}\exp\left\{  \boldsymbol{t}^{T}%
(\boldsymbol{\theta})\boldsymbol{g}(\boldsymbol{X},\boldsymbol{\theta
})\right\}  \right]  <\infty$ s.t. $\mathrm{E}\left[  \exp\left\{
\boldsymbol{t}^{T}(\boldsymbol{\theta})\boldsymbol{g}(\boldsymbol{X}%
,\boldsymbol{\theta})\right\}  \boldsymbol{g}(\boldsymbol{X}%
,\boldsymbol{\theta})\right]  =\boldsymbol{0}_{r}$;\newline iii) There exists
a function of $\boldsymbol{X}$, $f(\boldsymbol{X})$, such that $\left\Vert
\boldsymbol{G}_{\boldsymbol{X}}(\boldsymbol{\theta})\right\Vert $, $\left\Vert
\frac{\partial\boldsymbol{G}_{\boldsymbol{X}}(\boldsymbol{\theta})}%
{\partial\boldsymbol{\theta}}\right\Vert $ are bounded by $f(\boldsymbol{X}%
)$\ and \newline$\mathrm{E}\left[  \sup_{\boldsymbol{\theta\in\Theta}}%
\exp\left\{  k_{1}\boldsymbol{t}^{T}(\boldsymbol{\theta})\boldsymbol{g}%
(\boldsymbol{X},\boldsymbol{\theta})\right\}  f^{k_{2}}(\boldsymbol{X}%
)\right]  <\infty$, $k_{2}=1,2$, $k_{2}=0,1,2,3,4$, s.t. $\mathrm{E}\left[
\exp\left\{  \boldsymbol{t}^{T}(\boldsymbol{\theta})\boldsymbol{g}%
(\boldsymbol{X},\boldsymbol{\theta})\right\}  \boldsymbol{g}(\boldsymbol{X}%
,\boldsymbol{\theta})\right]  =\boldsymbol{0}_{r}$.
\end{condition}

The ETEL estimator of $\boldsymbol{\theta}_{\ast,ETEL}$,
$\widehat{\boldsymbol{\theta}}_{ETEL}$, associated with the misspecified
model, is obtained in the same manner done for the true model, in fact in
practice it is not possible to know when the model is misspecified. By
following Lemma 9 of Schennach (2007), it is convenient to study, apart from
the vector of parameters of interest $\boldsymbol{\theta}$ and the Lagrange
multipliers vector $\boldsymbol{t}$, two additional auxiliary variables
$\boldsymbol{\kappa}\in%
\mathbb{R}
^{r}$ and $\tau\in%
\mathbb{R}
$ in a joint vector%
\[
\boldsymbol{\beta}=(\boldsymbol{\theta}^{T},\boldsymbol{t}^{T}%
,\boldsymbol{\kappa}^{T},\tau)^{T}.
\]
According to Theorem 10 of Schennach (2007), by calculating first the
asymptotic distribution of $\widehat{\boldsymbol{\beta}}_{ETEL}%
=(\widehat{\boldsymbol{\theta}}_{ETEL}^{T},\widehat{\boldsymbol{t}}_{ETEL}%
^{T},\widehat{\boldsymbol{\kappa}}_{ETEL}^{T},\widehat{\tau}_{ETEL})^{T}$, and
subtracting thereafter the marginal distribution of
$\widehat{\boldsymbol{\theta}}_{ETEL}$, the procedure to calculate the
asymptotic distribution of $\widehat{\boldsymbol{\theta}}_{ETEL}^{T}$ is
simplified, under misspecification. The following auxiliary function%
\[
\boldsymbol{\varphi}(\boldsymbol{X},\boldsymbol{\beta})=(\boldsymbol{\varphi
}_{1}^{T}(\boldsymbol{X},\boldsymbol{\beta}),\boldsymbol{\varphi}_{2}%
^{T}(\boldsymbol{X},\boldsymbol{\beta}),\boldsymbol{\varphi}_{3}%
^{T}(\boldsymbol{X},\boldsymbol{\beta}),\varphi_{4}(\boldsymbol{X}%
,\boldsymbol{\beta}))^{T},
\]
with%
\begin{align*}
\boldsymbol{\varphi}_{1}(\boldsymbol{X},\boldsymbol{\beta})  &  =\exp
\{\boldsymbol{t}^{T}\boldsymbol{g}(\boldsymbol{X},\boldsymbol{\theta
})\}\boldsymbol{G}_{\boldsymbol{X}}^{T}(\boldsymbol{\theta})\left(
\boldsymbol{\kappa+tg}^{T}(\boldsymbol{X},\boldsymbol{\theta}%
)\boldsymbol{\kappa-t}\right)  +\tau\boldsymbol{G}_{\boldsymbol{X}}%
^{T}(\boldsymbol{\theta})\boldsymbol{t},\\
\boldsymbol{\varphi}_{2}(\boldsymbol{X},\boldsymbol{\beta})  &  =\left(
\tau-\exp\{\boldsymbol{t}^{T}\boldsymbol{g}(\boldsymbol{X},\boldsymbol{\theta
})\}\right)  \boldsymbol{g}(\boldsymbol{X},\boldsymbol{\theta})+\exp
\{\boldsymbol{t}^{T}\boldsymbol{g}(\boldsymbol{X},\boldsymbol{\theta
})\}\boldsymbol{g}(\boldsymbol{X},\boldsymbol{\theta})\boldsymbol{g}%
^{T}(\boldsymbol{X},\boldsymbol{\theta})\boldsymbol{\kappa},\\
\boldsymbol{\varphi}_{3}(\boldsymbol{X},\boldsymbol{\beta})  &  =\exp
\{\boldsymbol{t}^{T}\boldsymbol{g}(\boldsymbol{X},\boldsymbol{\theta
})\}\boldsymbol{g}(\boldsymbol{X},\boldsymbol{\theta}),\\
\varphi_{4}(\boldsymbol{X},\boldsymbol{\beta})  &  =\exp\{\boldsymbol{t}%
^{T}\boldsymbol{g}(\boldsymbol{X},\boldsymbol{\theta})\}-\tau,
\end{align*}
defines $\widehat{\boldsymbol{\beta}}_{ETEL}$, as the solution of $\frac{1}%
{n}\sum\nolimits_{i=1}^{n}\boldsymbol{\varphi}(\boldsymbol{X}_{i}%
,\boldsymbol{\beta})=\boldsymbol{0}_{p+2r+1}$, and the pseudo-true value
\[
\boldsymbol{\beta}_{\ast,ETEL}=(\boldsymbol{\theta}_{\ast,ETEL}^{T}%
,\boldsymbol{t}_{\ast,ETEL}^{T},\boldsymbol{\kappa}_{\ast,ETEL}^{T},\tau
_{\ast,ETEL})^{T},
\]
as the solution of $\mathrm{E}\left[  \boldsymbol{\varphi}(\boldsymbol{X}%
,\boldsymbol{\beta})\right]  =\boldsymbol{0}_{p+2r+1}$. Under Condition
\ref{RC2}, the asymptotic distribution of $\widehat{\boldsymbol{\beta}}%
_{ETEL}$ is given by%
\[
\sqrt{n}(\widehat{\boldsymbol{\beta}}_{ETEL}-\boldsymbol{\beta}_{\ast
,ETEL})\underset{n\rightarrow\infty}{\overset{\mathcal{L}}{\longrightarrow}%
}\mathcal{N}\left(  \boldsymbol{0}_{p+2r+1},\Gamma^{-1}(\boldsymbol{\beta
}_{\ast,ETEL})\Phi(\boldsymbol{\beta}_{\ast,ETEL})\left(  \Gamma
^{-1}(\boldsymbol{\beta}_{\ast,ETEL})\right)  ^{T}\right)  ,
\]
with%
\begin{align*}
\Gamma(\boldsymbol{\beta}_{\ast,ETEL})  &  =\mathrm{E}\left[  \frac{\partial
}{\partial\boldsymbol{\beta}}\left.  \boldsymbol{\varphi}(\boldsymbol{X}%
,\boldsymbol{\beta})\right\vert _{\boldsymbol{\beta}=\boldsymbol{\beta}%
_{\ast,ETEL}}\right]  ,\\
\Phi(\boldsymbol{\beta}_{\ast,ETEL})  &  =\mathrm{E}\left[
\boldsymbol{\varphi}(\boldsymbol{X},\boldsymbol{\beta}_{\ast,ETEL}%
)\boldsymbol{\varphi}^{T}(\boldsymbol{X},\boldsymbol{\beta}_{\ast
,ETEL})\right]  ,
\end{align*}
assuming that $\Gamma(\boldsymbol{\beta}_{\ast,ETEL})$ is nonsingular. Based
on this result,%
\begin{equation}
\sqrt{n}(\widehat{\boldsymbol{\theta}}_{ETEL}-\boldsymbol{\theta}_{\ast
,ETEL})\underset{n\rightarrow\infty}{\longrightarrow}\mathcal{N}\left(
\boldsymbol{0}_{p},\boldsymbol{\Sigma}_{\sqrt{n}\widehat{\boldsymbol{\theta}%
}_{ETEL}}\right)  , \label{thetaETEL}%
\end{equation}
with%
\begin{equation}
\boldsymbol{\Sigma}_{\sqrt{n}\widehat{\boldsymbol{\theta}}_{ETEL}}=%
\begin{pmatrix}
\boldsymbol{I}_{p} & \boldsymbol{0}_{(2r+1)\times(2r+1)}%
\end{pmatrix}
\Gamma^{-1}(\boldsymbol{\beta}_{\ast,ETEL})\Phi(\boldsymbol{\beta}_{\ast
,ETEL})\left(  \Gamma^{-1}(\boldsymbol{\beta}_{\ast,ETEL})\right)  ^{T}%
\begin{pmatrix}
\boldsymbol{I}_{p}\\
\boldsymbol{0}_{(2r+1)\times(2r+1)}%
\end{pmatrix}
. \label{sigmaTh}%
\end{equation}

\begin{lemma}
\label{ThDp}The first derivative of (\ref{pET}) is given by%
\[
\frac{\partial}{\partial\boldsymbol{\theta}}p_{ET,i}\left(  \boldsymbol{\theta
}\right)  =p_{ET,i}\left(  \boldsymbol{\theta}\right)  \left[  \boldsymbol{G}%
_{\boldsymbol{X}_{i}}^{T}(\boldsymbol{\theta})\boldsymbol{t}_{ET}%
(\boldsymbol{\theta})-\overline{\exp}_{ET}^{-1}(\boldsymbol{\theta}%
)\overline{\exp_{ET}\boldsymbol{G}^{T}}(\boldsymbol{\theta})\boldsymbol{t}%
_{ET}(\boldsymbol{\theta})-\widehat{\boldsymbol{K}}(\boldsymbol{\theta
})\boldsymbol{g}(\boldsymbol{X}_{i},\boldsymbol{\theta})\right]  ,
\]
where $\overline{\exp}_{ET}(\boldsymbol{\theta})$ was defined in
(\ref{roETEL}),%
\begin{align*}
\overline{\exp_{ET}\boldsymbol{G}^{T}}(\boldsymbol{\theta})  &  =\frac{1}{n}%
{\displaystyle\sum\limits_{i=1}^{n}}
\exp\{\boldsymbol{t}_{ET}^{T}(\boldsymbol{\theta})\boldsymbol{g}%
(\boldsymbol{X}_{i},\boldsymbol{\theta})\}\boldsymbol{G}_{\boldsymbol{X}_{i}%
}^{T}(\boldsymbol{\theta}),\\
\widehat{\boldsymbol{K}}(\boldsymbol{\theta})  &  =\left(  \overline{\exp
_{ET}\boldsymbol{G}^{T}\boldsymbol{t}_{ET}\boldsymbol{g}^{T}}%
(\boldsymbol{\theta})+\overline{\exp_{ET}\boldsymbol{G}^{T}}%
(\boldsymbol{\theta})\right)  \overline{\exp_{ET}\boldsymbol{gg}^{T}}%
^{-1}(\boldsymbol{\theta}),\\
\overline{\exp_{ET}\boldsymbol{G}^{T}\boldsymbol{t}_{ET}\boldsymbol{g}^{T}%
}(\boldsymbol{\theta})  &  =\frac{1}{n}%
{\displaystyle\sum\limits_{i=1}^{n}}
\exp\{\boldsymbol{t}_{ET}^{T}(\boldsymbol{\theta})\boldsymbol{g}%
(\boldsymbol{X}_{i},\boldsymbol{\theta})\}\boldsymbol{G}_{\boldsymbol{X}_{i}%
}^{T}(\boldsymbol{\theta})\boldsymbol{t}_{ET}(\boldsymbol{\theta
})\boldsymbol{g}^{T}(\boldsymbol{X}_{i},\boldsymbol{\theta}),\\
\overline{\exp_{ET}\boldsymbol{gg}^{T}}(\boldsymbol{\theta})  &  =\frac{1}{n}%
{\displaystyle\sum\limits_{i=1}^{n}}
\exp\{\boldsymbol{t}^{T}\boldsymbol{g}(\boldsymbol{X}_{i},\boldsymbol{\theta
})\}\boldsymbol{g}(\boldsymbol{X}_{i},\boldsymbol{\theta})\boldsymbol{g}%
^{T}(\boldsymbol{X}_{i},\boldsymbol{\theta}).
\end{align*}

\end{lemma}

(For the proof see Appendix)

\begin{lemma}
\label{ThDD1}The first derivative of $D_{\phi}\left(  \boldsymbol{u}%
,\boldsymbol{p}_{ET}\left(  \boldsymbol{\theta}\right)  \right)  $ is given by%
\begin{align}
&  \frac{\partial}{\partial\boldsymbol{\theta}}D_{\phi}\left(  \boldsymbol{u}%
,\boldsymbol{p}_{ET}\left(  \boldsymbol{\theta}\right)  \right)
=\overline{\exp}_{ET}^{-1}(\boldsymbol{\theta})\left[  \frac{1}{n}\sum
_{i=1}^{n}\exp\{\boldsymbol{t}_{ET}^{T}(\boldsymbol{\theta})\boldsymbol{g}%
(\boldsymbol{X}_{i},\boldsymbol{\theta})\}\psi\left(  \frac{\overline{\exp
}_{ET}(\boldsymbol{\theta})}{\exp\{\boldsymbol{t}_{ET}^{T}(\boldsymbol{\theta
})\boldsymbol{g}(\boldsymbol{X}_{i},\boldsymbol{\theta})\}}\right)
\boldsymbol{G}_{\boldsymbol{X}_{i}}^{T}(\boldsymbol{\theta})\boldsymbol{t}%
_{ET}(\boldsymbol{\theta})\right. \nonumber\\
&  -\overline{\exp}_{ET}^{-1}(\boldsymbol{\theta})\frac{1}{n}\sum_{i=1}%
^{n}\exp\{\boldsymbol{t}_{ET}^{T}(\boldsymbol{\theta})\boldsymbol{g}%
(\boldsymbol{X}_{i},\boldsymbol{\theta})\}\psi\left(  \frac{\overline{\exp
}_{ET}(\boldsymbol{\theta})}{\exp\{\boldsymbol{t}_{ET}^{T}(\boldsymbol{\theta
})\boldsymbol{g}(\boldsymbol{X}_{i},\boldsymbol{\theta})\}}\right)
\overline{\exp_{ET}\boldsymbol{G}^{T}}(\boldsymbol{\theta})\boldsymbol{t}%
_{ET}(\boldsymbol{\theta})\nonumber\\
&  \left.  -\widehat{\boldsymbol{K}}(\boldsymbol{\theta})\frac{1}{n}\sum
_{i=1}^{n}\exp\{\boldsymbol{t}_{ET}^{T}(\boldsymbol{\theta})\boldsymbol{g}%
(\boldsymbol{X}_{i},\boldsymbol{\theta})\}\psi\left(  \frac{\overline{\exp
}_{ET}(\boldsymbol{\theta})}{\exp\{\boldsymbol{t}_{ET}^{T}(\boldsymbol{\theta
})\boldsymbol{g}(\boldsymbol{X}_{i},\boldsymbol{\theta})\}}\right)
\boldsymbol{g}(\boldsymbol{X}_{i},\boldsymbol{\theta})\right]  , \label{DD1}%
\end{align}
and
\[
\frac{\partial}{\partial\boldsymbol{\theta}}D_{\phi}\left(  \boldsymbol{u}%
,\boldsymbol{p}_{ET}\left(  \boldsymbol{\theta}\right)  \right)
\underset{n\rightarrow\infty}{\overset{P}{\longrightarrow}}\boldsymbol{r}%
_{T_{n}^{\phi}}(\boldsymbol{\theta}),
\]
with $\psi(x)$ given by (\ref{psi}),
\begin{equation}
\boldsymbol{r}_{T_{n}^{\phi}}(\boldsymbol{\theta})=\mathrm{E}^{-1}\left[
\exp\{\overline{\boldsymbol{t}}_{ET}^{T}(\boldsymbol{\theta})\boldsymbol{g}%
(\boldsymbol{X},\boldsymbol{\theta})\}\right]  \left\{  \boldsymbol{r}%
_{1}(\boldsymbol{\theta})-\boldsymbol{r}_{2}(\boldsymbol{\theta}%
)-\boldsymbol{r}_{3}(\boldsymbol{\theta})\right\}  , \label{s}%
\end{equation}%
\begin{align*}
\boldsymbol{r}_{1}(\boldsymbol{\theta})  &  =\mathrm{E}\left[  \exp
\{\overline{\boldsymbol{t}}_{ET}^{T}(\boldsymbol{\theta})\boldsymbol{g}%
(\boldsymbol{X},\boldsymbol{\theta})\}\psi\left(  \frac{\mathrm{E}\left[
\exp\{\overline{\boldsymbol{t}}_{ET}^{T}(\boldsymbol{\theta})\boldsymbol{g}%
(\boldsymbol{X},\boldsymbol{\theta})\right]  \}}{\exp\{\overline
{\boldsymbol{t}}_{ET}^{T}(\boldsymbol{\theta})\boldsymbol{g}(\boldsymbol{X}%
,\boldsymbol{\theta})\}}\right)  \boldsymbol{G}_{\boldsymbol{X}}%
^{T}(\boldsymbol{\theta})\right]  \overline{\boldsymbol{t}}_{ET}%
(\boldsymbol{\theta}),\\
\boldsymbol{r}_{2}(\boldsymbol{\theta})  &  =\mathrm{E}^{-1}\left[
\exp\{\overline{\boldsymbol{t}}_{ET}^{T}(\boldsymbol{\theta})\boldsymbol{g}%
(\boldsymbol{X},\boldsymbol{\theta})\}\right]  \mathrm{E}\left[
\exp\{\overline{\boldsymbol{t}}_{ET}^{T}(\boldsymbol{\theta})\boldsymbol{g}%
(\boldsymbol{X},\boldsymbol{\theta})\}\psi\left(  \frac{\mathrm{E}\left[
\exp\{\overline{\boldsymbol{t}}_{ET}^{T}(\boldsymbol{\theta})\boldsymbol{g}%
(\boldsymbol{X},\boldsymbol{\theta})\right]  }{\exp\{\overline{\boldsymbol{t}%
}_{ET}^{T}(\boldsymbol{\theta})\boldsymbol{g}(\boldsymbol{X}%
,\boldsymbol{\theta})\}}\right)  \right] \\
&  \times\mathrm{E}\left[  \exp\{\overline{\boldsymbol{t}}_{ET}^{T}%
(\boldsymbol{\theta})\boldsymbol{g}(\boldsymbol{X},\boldsymbol{\theta
})\}\boldsymbol{G}_{\boldsymbol{X}}^{T}(\boldsymbol{\theta})\right]
\overline{\boldsymbol{t}}_{ET}(\boldsymbol{\theta}),\\
\boldsymbol{r}_{3}(\boldsymbol{\theta})  &  =\boldsymbol{K}(\boldsymbol{\theta
})\mathrm{E}\left[  \exp\{\overline{\boldsymbol{t}}_{ET}^{T}%
(\boldsymbol{\theta})\boldsymbol{g}(\boldsymbol{X},\boldsymbol{\theta}%
)\}\psi\left(  \frac{\mathrm{E}\left[  \exp\{\overline{\boldsymbol{t}}%
_{ET}^{T}(\boldsymbol{\theta})\boldsymbol{g}(\boldsymbol{X},\boldsymbol{\theta
})\right]  }{\exp\{\overline{\boldsymbol{t}}_{ET}^{T}(\boldsymbol{\theta
})\boldsymbol{g}(\boldsymbol{X},\boldsymbol{\theta})\}}\right)  \boldsymbol{g}%
(\boldsymbol{X},\boldsymbol{\theta})\right]  ,
\end{align*}%
\begin{align}
\boldsymbol{K}(\boldsymbol{\theta})  &  =\left\{  \mathrm{E}\left[
\exp\{\overline{\boldsymbol{t}}_{ET}^{T}(\boldsymbol{\theta})\boldsymbol{g}%
(\boldsymbol{X},\boldsymbol{\theta})\}\boldsymbol{G}_{\boldsymbol{X}}%
^{T}(\boldsymbol{\theta})\overline{\boldsymbol{t}}_{ET}^{T}(\boldsymbol{\theta
})\boldsymbol{g}^{T}(\boldsymbol{X},\boldsymbol{\theta})\right]
+\mathrm{E}\left[  \exp\{\overline{\boldsymbol{t}}_{ET}^{T}(\boldsymbol{\theta
})\boldsymbol{g}(\boldsymbol{X},\boldsymbol{\theta})\}\boldsymbol{G}%
_{\boldsymbol{X}}^{T}(\boldsymbol{\theta})\right]  \right\} \nonumber\\
&  \times\mathrm{E}^{-1}\left[  \exp\{\overline{\boldsymbol{t}}_{ET}%
^{T}(\boldsymbol{\theta})\boldsymbol{g}(\boldsymbol{X},\boldsymbol{\theta
})\}\boldsymbol{g}(\boldsymbol{X},\boldsymbol{\theta})\boldsymbol{g}%
^{T}(\boldsymbol{X},\boldsymbol{\theta})\right]  , \label{K}%
\end{align}
$\overline{\boldsymbol{t}}_{ET}(\boldsymbol{\theta})$ is the solution in
$\boldsymbol{t}$ of $\mathrm{E}\left[  \exp\left\{  \boldsymbol{t}%
^{T}\boldsymbol{g}(\boldsymbol{X},\boldsymbol{\theta})\right\}  \boldsymbol{g}%
(\boldsymbol{X},\boldsymbol{\theta})\right]  =\boldsymbol{0}_{r}$.
\end{lemma}

Let $\widehat{\mathbf{S}}_{12}(\boldsymbol{\theta})=\frac{1}{n}%
{\textstyle\sum\nolimits_{i=1}^{n}}
\boldsymbol{G}_{\boldsymbol{X}_{i}}(\boldsymbol{\theta})$ be a consistent
estimator of $\mathbf{S}_{12}(\boldsymbol{\theta})$\ given in (\ref{S12}). It
is interesting that according to formula (42) of Schennach (2007),
\begin{align*}
&  \frac{\partial}{\partial\boldsymbol{\theta}}D_{\mathrm{Kull}}\left(
\boldsymbol{u},\boldsymbol{p}_{ET}\left(  \boldsymbol{\theta}\right)  \right)
=-\frac{\partial}{\partial\boldsymbol{\theta}}\ell_{ETEL}(\boldsymbol{\theta
})\\
&  =\widehat{\boldsymbol{K}}(\boldsymbol{\theta})\left(  \frac{1}{n}%
{\displaystyle\sum\limits_{i=1}^{n}}
\boldsymbol{g}(\boldsymbol{X}_{i},\boldsymbol{\theta})\right)  +\left(
{\displaystyle\sum\limits_{i=1}^{n}}
p_{ET,i}\left(  \boldsymbol{\theta}\right)  \boldsymbol{G}_{\boldsymbol{X}%
_{i}}^{T}(\boldsymbol{\theta})-\widehat{\mathbf{S}}_{12}^{T}%
(\boldsymbol{\theta})\right)  \boldsymbol{t}_{ET}(\boldsymbol{\theta})\\
&  =\widehat{\boldsymbol{K}}(\boldsymbol{\theta})\left(  \frac{1}{n}%
{\displaystyle\sum\limits_{i=1}^{n}}
\boldsymbol{g}(\boldsymbol{X}_{i},\boldsymbol{\theta})\right)  +\left[
\overline{\exp}_{ET}^{-1}(\boldsymbol{\theta})\left(  \frac{1}{n}%
{\displaystyle\sum\limits_{i=1}^{n}}
\exp\{\boldsymbol{t}^{T}\boldsymbol{g}(\boldsymbol{X}_{i},\boldsymbol{\theta
})\}\boldsymbol{G}_{\boldsymbol{X}_{i}}^{T}(\boldsymbol{\theta})\right)
-\widehat{\mathbf{S}}_{12}^{T}(\boldsymbol{\theta})\right]  \boldsymbol{t}%
_{ET}(\boldsymbol{\theta}),
\end{align*}
which matches (\ref{DD1}) with $\phi(x)=x\log x-x+1$.

(For the proof see Appendix)

\begin{lemma}
\label{ThDD2}The first derivative of $D_{\phi}\left(  \boldsymbol{p}%
_{ET}\left(  \boldsymbol{\theta}\right)  ,\boldsymbol{p}_{ET}\left(
\boldsymbol{\theta}_{0}\right)  \right)  $ is given by%
\begin{align}
&  \frac{\partial}{\partial\boldsymbol{\theta}}D_{\phi}\left(  \boldsymbol{p}%
_{ET}\left(  \boldsymbol{\theta}\right)  ,\boldsymbol{p}_{ET}\left(
\boldsymbol{\theta}_{0}\right)  \right)  =\overline{\exp}_{ET}^{-1}%
(\boldsymbol{\theta})\frac{1}{n}\sum_{i=1}^{n}\phi^{\prime}\left(
\frac{p_{ET,i}\left(  \boldsymbol{\theta}\right)  }{p_{ET,i}\left(
\boldsymbol{\theta}_{0}\right)  }\right)  \overline{\exp_{ET}\boldsymbol{G}%
^{T}}(\boldsymbol{\theta})\boldsymbol{t}_{ET}(\boldsymbol{\theta})\nonumber\\
&  +\widehat{\boldsymbol{K}}(\boldsymbol{\theta})\frac{1}{n}\sum_{i=1}^{n}%
\phi^{\prime}\left(  \frac{p_{ET,i}\left(  \boldsymbol{\theta}\right)
}{p_{ET,i}\left(  \boldsymbol{\theta}_{0}\right)  }\right)  \boldsymbol{g}%
(\boldsymbol{X}_{i},\boldsymbol{\theta})-\frac{1}{n}\sum_{i=1}^{n}\phi
^{\prime}\left(  \frac{p_{ET,i}\left(  \boldsymbol{\theta}\right)  }%
{p_{ET,i}\left(  \boldsymbol{\theta}_{0}\right)  }\right)  \boldsymbol{G}%
_{\boldsymbol{X}_{i}}^{T}(\boldsymbol{\theta})\boldsymbol{t}_{ET}%
(\boldsymbol{\theta}), \label{DD2b}%
\end{align}
where%
\[
\frac{p_{ET,i}\left(  \boldsymbol{\theta}\right)  }{p_{ET,i}\left(
\boldsymbol{\theta}_{0}\right)  }=\frac{\exp\{\boldsymbol{t}_{ET}%
^{T}(\boldsymbol{\theta})\boldsymbol{g}(\boldsymbol{X}_{i},\boldsymbol{\theta
})\}}{\exp\{\boldsymbol{t}_{ET}^{T}(\boldsymbol{\theta})\boldsymbol{g}%
(\boldsymbol{X}_{i},\boldsymbol{\theta}_{0})\}}\frac{\overline{\exp}%
_{ET}(\boldsymbol{\theta}_{0})}{\overline{\exp}_{ET}(\boldsymbol{\theta})},
\]
and
\[
\frac{\partial}{\partial\boldsymbol{\theta}}D_{\phi}\left(  \boldsymbol{p}%
_{ET}\left(  \boldsymbol{\theta}\right)  ,\boldsymbol{p}_{ET}\left(
\boldsymbol{\theta}_{0}\right)  \right)  \underset{n\rightarrow\infty
}{\overset{P}{\longrightarrow}}\boldsymbol{q}_{S_{n}^{\phi}}%
(\boldsymbol{\theta},\boldsymbol{\theta}_{0}),
\]
with%
\begin{equation}
\boldsymbol{q}_{S_{n}^{\phi}}(\boldsymbol{\theta},\boldsymbol{\theta}%
_{0})=\boldsymbol{q}_{1}(\boldsymbol{\theta},\boldsymbol{\theta}%
_{0})+\boldsymbol{q}_{2}(\boldsymbol{\theta},\boldsymbol{\theta}%
_{0})-\boldsymbol{q}_{3}(\boldsymbol{\theta},\boldsymbol{\theta}_{0}),
\label{q}%
\end{equation}%
\begin{align*}
\boldsymbol{q}_{1}(\boldsymbol{\theta},\boldsymbol{\theta}_{0})  &
=\mathrm{E}^{-1}\left[  \exp\{\overline{\boldsymbol{t}}_{ET}^{T}%
(\boldsymbol{\theta})\boldsymbol{g}(\boldsymbol{X},\boldsymbol{\theta
})\}\right]  \mathrm{E}\left[  \phi^{\prime}\left(  \frac{\exp\{\overline
{\boldsymbol{t}}_{ET}^{T}(\boldsymbol{\theta})\boldsymbol{g}(\boldsymbol{X}%
,\boldsymbol{\theta})\}}{\exp\{\overline{\boldsymbol{t}}_{ET}^{T}%
(\boldsymbol{\theta}_{0})\boldsymbol{g}(\boldsymbol{X},\boldsymbol{\theta}%
_{0})\}}\frac{\mathrm{E}\left[  \exp\{\overline{\boldsymbol{t}}_{ET}%
^{T}(\boldsymbol{\theta}_{0})\boldsymbol{g}(\boldsymbol{X},\boldsymbol{\theta
}_{0})\right]  }{\mathrm{E}\left[  \exp\{\overline{\boldsymbol{t}}_{ET}%
^{T}(\boldsymbol{\theta})\boldsymbol{g}(\boldsymbol{X},\boldsymbol{\theta
})\right]  }\right)  \right] \\
&  \times\mathrm{E}\left[  \exp\{\overline{\boldsymbol{t}}_{ET}^{T}%
(\boldsymbol{\theta})\boldsymbol{g}(\boldsymbol{X},\boldsymbol{\theta
})\}\boldsymbol{G}_{\boldsymbol{X}}^{T}(\boldsymbol{\theta})\right]
\overline{\boldsymbol{t}}_{ET}(\boldsymbol{\theta}),\\
\boldsymbol{q}_{2}(\boldsymbol{\theta},\boldsymbol{\theta}_{0})  &
=\boldsymbol{K}(\boldsymbol{\theta})\mathrm{E}\left[  \phi^{\prime}\left(
\frac{\exp\{\overline{\boldsymbol{t}}_{ET}^{T}(\boldsymbol{\theta
})\boldsymbol{g}(\boldsymbol{X},\boldsymbol{\theta})\}}{\exp\{\overline
{\boldsymbol{t}}_{ET}^{T}(\boldsymbol{\theta}_{0})\boldsymbol{g}%
(\boldsymbol{X},\boldsymbol{\theta}_{0})\}}\frac{\mathrm{E}\left[
\exp\{\overline{\boldsymbol{t}}_{ET}^{T}(\boldsymbol{\theta}_{0}%
)\boldsymbol{g}(\boldsymbol{X},\boldsymbol{\theta}_{0})\right]  }%
{\mathrm{E}\left[  \exp\{\overline{\boldsymbol{t}}_{ET}^{T}(\boldsymbol{\theta
})\boldsymbol{g}(\boldsymbol{X},\boldsymbol{\theta})\right]  }\right)
\boldsymbol{g}(\boldsymbol{X},\boldsymbol{\theta})\right]  ,\\
\boldsymbol{q}_{3}(\boldsymbol{\theta},\boldsymbol{\theta}_{0})  &
=\mathrm{E}\left[  \phi^{\prime}\left(  \frac{\exp\{\overline{\boldsymbol{t}%
}_{ET}^{T}(\boldsymbol{\theta})\boldsymbol{g}(\boldsymbol{X}%
,\boldsymbol{\theta})\}}{\exp\{\overline{\boldsymbol{t}}_{ET}^{T}%
(\boldsymbol{\theta}_{0})\boldsymbol{g}(\boldsymbol{X},\boldsymbol{\theta}%
_{0})\}}\frac{\mathrm{E}\left[  \exp\{\overline{\boldsymbol{t}}_{ET}%
^{T}(\boldsymbol{\theta}_{0})\boldsymbol{g}(\boldsymbol{X},\boldsymbol{\theta
}_{0})\right]  }{\mathrm{E}\left[  \exp\{\overline{\boldsymbol{t}}_{ET}%
^{T}(\boldsymbol{\theta})\boldsymbol{g}(\boldsymbol{X},\boldsymbol{\theta
})\right]  }\right)  \boldsymbol{G}_{\boldsymbol{X}}^{T}(\boldsymbol{\theta
})\right]  \overline{\boldsymbol{t}}_{ET}(\boldsymbol{\theta}),
\end{align*}
$\overline{\boldsymbol{t}}_{ET}(\boldsymbol{\theta})$ is the solution in
$\boldsymbol{t}$ of of $\mathrm{E}\left[  \exp\left\{  \boldsymbol{t}%
^{T}\boldsymbol{g}(\boldsymbol{X},\boldsymbol{\theta})\right\}  \boldsymbol{g}%
(\boldsymbol{X},\boldsymbol{\theta})\right]  =\boldsymbol{0}_{r}$.
\end{lemma}

The following two theorems evaluate the effect of a misspecified alternative
hypothesis on the asymptotic distribution of the empirical $\phi$-divergence test-statistics.

\begin{theorem}
\label{Th}Under the assumption that the pseudo-true parameter value
$\boldsymbol{\theta}_{\ast,ETEL}$ is different from $\boldsymbol{\theta}_{0}$%
\[
\frac{n^{1/2}}{\sqrt{\boldsymbol{r}_{T_{n}^{\phi}}^{T}(\boldsymbol{\theta
}_{\ast,ETEL})\boldsymbol{\Sigma}_{\sqrt{n}\widehat{\boldsymbol{\theta}%
}_{ETEL}}\boldsymbol{r}_{T_{n}^{\phi}}(\boldsymbol{\theta}_{\ast,ETEL})}%
}\left(  \frac{\phi^{\prime\prime}(1)T_{n}^{\phi}(\widehat{\boldsymbol{\theta
}}_{ETEL},\boldsymbol{\theta}_{0})}{2n}-\mu_{T_{n}^{\phi}}(\boldsymbol{\theta
}_{0},\boldsymbol{\theta}_{\ast,ETEL})\right)  \overset{\mathcal{L}%
}{\underset{n\rightarrow\infty}{\longrightarrow}}\mathcal{N}\left(
0,1\right)  ,
\]
where $\boldsymbol{\Sigma}_{\sqrt{n}\widehat{\boldsymbol{\theta}}_{ETEL}}$ is
given by (\ref{sigmaTh}), $\boldsymbol{r}_{T_{n}^{\phi}}(\boldsymbol{\theta
}_{\ast,ETEL})$ by (\ref{s}) and
\begin{align*}
\mu_{T_{n}^{\phi}}(\boldsymbol{\theta}_{0},\boldsymbol{\theta}_{\ast,ETEL})
&  =\mathrm{E}^{-1}\left[  \exp\{\overline{\boldsymbol{t}}_{ET}^{T}%
(\boldsymbol{\theta}_{\ast,ETEL})\boldsymbol{g}(\boldsymbol{X}%
,\boldsymbol{\theta}_{\ast,ETEL})\}\right] \\
&  \times\mathrm{E}\left[  \exp\{\overline{\boldsymbol{t}}_{ET}^{T}%
(\boldsymbol{\theta}_{\ast,ETEL})\boldsymbol{g}(\boldsymbol{X}%
,\boldsymbol{\theta}_{\ast,ETEL})\}\phi\left(  \frac{E\left[  \exp
\{\overline{\boldsymbol{t}}_{ET}^{T}(\boldsymbol{\theta}_{\ast,ETEL}%
)\boldsymbol{g}(\boldsymbol{X},\boldsymbol{\theta}_{\ast,ETEL})\}\right]
}{\exp\{\overline{\boldsymbol{t}}_{ET}^{T}(\boldsymbol{\theta}_{\ast
,ETEL})\boldsymbol{g}(\boldsymbol{X},\boldsymbol{\theta}_{\ast,ETEL}%
)\}}\right)  \right] \\
&  -\mathrm{E}^{-1}\left[  \exp\{\overline{\boldsymbol{t}}_{ET}^{T}%
(\boldsymbol{\theta}_{0})\boldsymbol{g}(\boldsymbol{X},\boldsymbol{\theta}%
_{0})\}\right] \\
&  \times\mathrm{E}\left[  \exp\{\overline{\boldsymbol{t}}_{ET}^{T}%
(\boldsymbol{\theta}_{0})\boldsymbol{g}(\boldsymbol{X},\boldsymbol{\theta}%
_{0})\}\phi\left(  \frac{E\left[  \exp\{\overline{\boldsymbol{t}}_{ET}%
^{T}(\boldsymbol{\theta}_{0})\boldsymbol{g}(\boldsymbol{X},\boldsymbol{\theta
}_{0})\}\right]  }{\exp\{\overline{\boldsymbol{t}}_{ET}^{T}(\boldsymbol{\theta
}_{0})\boldsymbol{g}(\boldsymbol{X},\boldsymbol{\theta}_{0})\}}\right)
\right]  ,
\end{align*}
with $\overline{\boldsymbol{t}}_{ET}(\boldsymbol{\theta})$ being the solution
in $\boldsymbol{t}$ of $\mathrm{E}\left[  \exp\left\{  \boldsymbol{t}%
^{T}\boldsymbol{g}(\boldsymbol{X},\boldsymbol{\theta})\right\}  \boldsymbol{g}%
(\boldsymbol{X},\boldsymbol{\theta})\right]  =\boldsymbol{0}_{r}$.
\end{theorem}

\begin{proof}
The first order Taylor expansion of $D_{\phi}\left(  \boldsymbol{u}%
,\boldsymbol{p}_{ET}\left(  \boldsymbol{\theta}\right)  \right)  $ around
$\boldsymbol{\theta}_{\ast,ETEL}$ is%
\[
D_{\phi}\left(  \boldsymbol{u},\boldsymbol{p}_{ET}\left(  \boldsymbol{\theta
}\right)  \right)  =D_{\phi}\left(  \boldsymbol{u},\boldsymbol{p}_{ET}\left(
\boldsymbol{\theta}_{\ast,ETEL}\right)  \right)  +\frac{\partial}%
{\partial\boldsymbol{\theta}}\left.  D_{\phi}\left(  \boldsymbol{u}%
,\boldsymbol{p}_{ET}\left(  \boldsymbol{\theta}\right)  \right)  \right\vert
_{\boldsymbol{\theta}=\boldsymbol{\theta}_{\ast,ETEL}}\left(
\boldsymbol{\theta-\theta}_{\ast,ETEL}\right)  +o\left(  \left\Vert
\boldsymbol{\theta-\theta}_{\ast,ETEL}\right\Vert \right)  .
\]
In particular, for $\boldsymbol{\theta=}\widehat{\boldsymbol{\theta}}_{ETEL}$%
\begin{align*}
D_{\phi}\left(  \boldsymbol{u},\boldsymbol{p}_{ET}(\widehat{\boldsymbol{\theta
}}_{ETEL})\right)   &  =D_{\phi}\left(  \boldsymbol{u},\boldsymbol{p}%
_{ET}\left(  \boldsymbol{\theta}_{\ast,ETEL}\right)  \right)  +\frac{\partial
}{\partial\boldsymbol{\theta}}\left.  D_{\phi}\left(  \boldsymbol{u}%
,\boldsymbol{p}_{ET}\left(  \boldsymbol{\theta}\right)  \right)  \right\vert
_{\boldsymbol{\theta}=\boldsymbol{\theta}_{\ast,ETEL}}%
(\widehat{\boldsymbol{\theta}}_{ETEL}\boldsymbol{-\theta}_{\ast,ETEL})\\
&  +o\left(  \left\Vert \widehat{\boldsymbol{\theta}}_{ETEL}%
\boldsymbol{-\theta}_{\ast,ETEL}\right\Vert \right)  .
\end{align*}
According to Theorem \ref{ThDD1} $\frac{\partial}{\partial\boldsymbol{\theta}%
}D_{\phi}\left(  \boldsymbol{u},\boldsymbol{p}_{ET}\left(  \boldsymbol{\theta
}\right)  \right)  $ converges in probability to a fixed vector, and so%
\[
D_{\phi}(\boldsymbol{u},\boldsymbol{p}_{ET}(\widehat{\boldsymbol{\theta}%
}_{ETEL}))=D_{\phi}(\boldsymbol{u},\boldsymbol{p}_{ET}\left(
\boldsymbol{\theta}_{\ast,ETEL}\right)  )+\boldsymbol{r}_{T_{n}^{\phi}}%
^{T}\left(  \boldsymbol{\theta}_{\ast,ETEL}\right)
(\widehat{\boldsymbol{\theta}}_{ETEL}-\boldsymbol{\theta}_{\ast,ETEL}%
)+o(||\widehat{\boldsymbol{\theta}}_{ETEL}-\boldsymbol{\theta}_{\ast
,ETEL}||).
\]
From (\ref{thetaETEL}) it holds $\sqrt{n}\,o(||\widehat{\boldsymbol{\theta}%
}_{ETEL}-\boldsymbol{\theta}_{\ast,ETEL}||)=o_{p}(1)$. Thus, the random
variables
\[
\sqrt{n}\left(  D_{\phi}(\boldsymbol{u},\boldsymbol{p}_{ET}%
(\widehat{\boldsymbol{\theta}}_{ETEL}))-D_{\phi}(\boldsymbol{u},\boldsymbol{p}%
_{ET}\left(  \boldsymbol{\theta}_{\ast,ETEL}\right)  )\right)  \text{ and
}\boldsymbol{r}_{T_{n}^{\phi}}^{T}\left(  \boldsymbol{\theta}^{\ast}\right)
\sqrt{n}(\widehat{\boldsymbol{\theta}}_{ETEL}\boldsymbol{-\theta}_{\ast
,ETEL})
\]
have the same asymptotic distribution, and since%
\begin{align*}
\tfrac{\phi^{\prime\prime}(1)}{2\sqrt{n}}\left(  T_{n}^{\phi}%
(\widehat{\boldsymbol{\theta}}_{ETEL},\boldsymbol{\theta}_{0})-T_{n}^{\phi
}(\boldsymbol{\theta}_{\ast,ETEL},\boldsymbol{\theta}_{0})\right)   &
=\sqrt{n}\left(  D_{\phi}\left(  \boldsymbol{u},\boldsymbol{p}_{ET}%
(\widehat{\boldsymbol{\theta}}_{ETEL})\right)  -D_{\phi}\left(  \boldsymbol{u}%
,\boldsymbol{p}_{ET}\left(  \boldsymbol{\theta}_{\ast,ETEL}\right)  \right)
\right)  ,\\
E_{F_{\boldsymbol{\theta}_{\ast,ETEL}}}\left[  \frac{\phi^{\prime\prime
}(1)T_{n}^{\phi}(\boldsymbol{\theta}_{\ast,ETEL},\boldsymbol{\theta}_{0})}%
{2n}\right]   &  =\mu_{T_{n}^{\phi}}(\boldsymbol{\theta}_{0}%
,\boldsymbol{\theta}_{\ast,ETEL}),
\end{align*}
the desired result is obtained.
\end{proof}

\begin{theorem}
Under the assumption that the pseudo-true parameter value $\boldsymbol{\theta
}_{\ast,ETEL}$ is different from $\boldsymbol{\theta}_{0}$%
\[
\frac{n^{1/2}}{\sqrt{\boldsymbol{q}_{S_{n}^{\phi}}^{T}(\boldsymbol{\theta
}_{\ast,ETEL},\boldsymbol{\theta}_{0})\boldsymbol{\Sigma}_{\sqrt
{n}\widehat{\boldsymbol{\theta}}_{ETEL}}\boldsymbol{q}_{S_{n}^{\phi}%
}(\boldsymbol{\theta}_{\ast,ETEL},\boldsymbol{\theta}_{0})}}\left(  \frac
{\phi^{\prime\prime}(1)S_{n}^{\phi}(\widehat{\boldsymbol{\theta}}%
_{ETEL},\boldsymbol{\theta}_{0})}{2n}-\mu_{S_{n}^{\phi}}(\boldsymbol{\theta
}_{0},\boldsymbol{\theta}_{\ast,ETEL})\right)  \overset{\mathcal{L}%
}{\underset{n\rightarrow\infty}{\longrightarrow}}\mathcal{N}\left(
0,1\right)  ,
\]
where $\boldsymbol{\Sigma}_{\sqrt{n}\widehat{\boldsymbol{\theta}}_{ETEL}}$ is
given by (\ref{sigmaTh}), $\boldsymbol{q}_{S_{n}^{\phi}}(\boldsymbol{\theta
}_{\ast,ETEL},\boldsymbol{\theta}_{0})$ by (\ref{q}) and
\begin{align*}
&  \mu_{S_{n}^{\phi}}(\boldsymbol{\theta}_{0},\boldsymbol{\theta}_{\ast
,ETEL})=\mathrm{E}^{-1}\left[  \exp\{\overline{\boldsymbol{t}}_{ET}%
^{T}(\boldsymbol{\theta}_{0})\boldsymbol{g}(\boldsymbol{X},\boldsymbol{\theta
}_{0})\}\right] \\
&  \times\mathrm{E}\left[  \exp\{\overline{\boldsymbol{t}}_{ET}^{T}%
(\boldsymbol{\theta}_{0})\boldsymbol{g}(\boldsymbol{X},\boldsymbol{\theta}%
_{0})\}\phi\left(  \frac{\exp\{\overline{\boldsymbol{t}}_{ET}^{T}%
(\boldsymbol{\theta}_{\ast,ETEL})\boldsymbol{g}(\boldsymbol{X}%
,\boldsymbol{\theta}_{\ast,ETEL})\}}{\exp\{\overline{\boldsymbol{t}}_{ET}%
^{T}(\boldsymbol{\theta}_{0})\boldsymbol{g}(\boldsymbol{X},\boldsymbol{\theta
}_{0})\}}\frac{\mathrm{E}\left[  \exp\{\overline{\boldsymbol{t}}_{ET}%
^{T}(\boldsymbol{\theta}_{0})\boldsymbol{g}(\boldsymbol{X},\boldsymbol{\theta
}_{0})\}\right]  }{\mathrm{E}\left[  \exp\{\overline{\boldsymbol{t}}_{ET}%
^{T}(\boldsymbol{\theta}_{\ast,ETEL})\boldsymbol{g}(\boldsymbol{X}%
,\boldsymbol{\theta}_{\ast,ETEL})\}\right]  }\right)  \right]  .
\end{align*}
with $\overline{\boldsymbol{t}}_{ET}(\boldsymbol{\theta})$ being the solution
in $\boldsymbol{t}$ of $\mathrm{E}\left[  \exp\left\{  \boldsymbol{t}%
^{T}\boldsymbol{g}(\boldsymbol{X},\boldsymbol{\theta})\right\}  \boldsymbol{g}%
(\boldsymbol{X},\boldsymbol{\theta})\right]  =\boldsymbol{0}_{r}$.
\end{theorem}

\begin{proof}
It is omitted since similar steps of the proof for Theorem \ref{Th} are needed.
\end{proof}

\begin{corollary}
\label{cor}Under the assumption that the pseudo-true parameter value
$\boldsymbol{\theta}_{\ast,ETEL}$ is different from $\boldsymbol{\theta}_{0}$,
the asymptotic distribution of the likelihood ratio test-statistics is given
by%
\[
\frac{n^{1/2}}{\sqrt{\boldsymbol{r}_{G^{2}}^{T}(\boldsymbol{\theta}%
_{\ast,ETEL})\boldsymbol{\Sigma}_{\sqrt{n}\widehat{\boldsymbol{\theta}}%
_{ETEL}}\boldsymbol{r}_{G^{2}}(\boldsymbol{\theta}_{\ast,ETEL})}}\left(
\frac{\phi^{\prime\prime}(1)G_{n}^{2}(\widehat{\boldsymbol{\theta}}%
_{ETEL},\boldsymbol{\theta}_{0})}{2n}-\mu_{G^{2}}(\boldsymbol{\theta}%
_{0},\boldsymbol{\theta}_{\ast,ETEL})\right)  \overset{\mathcal{L}%
}{\underset{n\rightarrow\infty}{\longrightarrow}}\mathcal{N}\left(
0,1\right)  ,
\]
where%
\begin{equation}
\boldsymbol{r}_{G^{2}}(\boldsymbol{\theta}_{\ast,ETEL})=\boldsymbol{r}%
_{G,2}(\boldsymbol{\theta}_{\ast,ETEL})-\boldsymbol{r}_{G,3}%
(\boldsymbol{\theta}_{\ast,ETEL}), \label{derGT}%
\end{equation}
with%
\begin{align*}
\boldsymbol{r}_{G,2}(\boldsymbol{\theta}_{\ast,ETEL})  &  =\boldsymbol{K}%
(\boldsymbol{\theta}_{\ast,ETEL})\mathrm{E}\left[  \boldsymbol{g}%
(\boldsymbol{X},\boldsymbol{\theta}_{\ast,ETEL})\right]  ,\\
\boldsymbol{r}_{G,3}(\boldsymbol{\theta}_{\ast,ETEL})  &  =\left\{
\mathrm{E}^{-1}\left[  \exp\{\overline{\boldsymbol{t}}_{ET}^{T}%
(\boldsymbol{\theta}_{\ast,ETEL})\boldsymbol{g}(\boldsymbol{X}%
,\boldsymbol{\theta}_{\ast,ETEL})\}\right]  \right. \\
&  \left.  \times\mathrm{E}\left[  \exp\{\overline{\boldsymbol{t}}_{ET}%
^{T}(\boldsymbol{\theta}_{\ast,ETEL})\boldsymbol{g}(\boldsymbol{X}%
,\boldsymbol{\theta}_{\ast,ETEL})\}\boldsymbol{G}_{\boldsymbol{X}}%
^{T}(\boldsymbol{\theta}_{\ast,ETEL})\right]  -\mathbf{S}_{12}^{T}%
(\boldsymbol{\theta}_{\ast,ETEL})\right\}  \overline{\boldsymbol{t}}%
_{ET}(\boldsymbol{\theta}_{\ast,ETEL}),
\end{align*}
$\boldsymbol{K}(\boldsymbol{\theta}_{\ast,ETEL})$ is given by (\ref{K}) and%
\begin{align}
\mu_{G^{2}}(\boldsymbol{\theta}_{0},\boldsymbol{\theta}_{\ast,ETEL})  &
=\log\frac{E_{F_{\boldsymbol{\theta}_{\ast,ETEL}}}\left[  \exp\{\overline
{\boldsymbol{t}}_{ET}^{T}(\boldsymbol{\theta}_{\ast,ETEL})\boldsymbol{g}%
(\boldsymbol{X},\boldsymbol{\theta}_{\ast,ETEL})\}\right]  }%
{E_{F_{\boldsymbol{\theta}_{\ast,ETEL}}}\left[  \exp\{\overline{\boldsymbol{t}%
}_{ET}^{T}(\boldsymbol{\theta}_{0})\boldsymbol{g}(\boldsymbol{X}%
,\boldsymbol{\theta}_{0})\}\right]  }\nonumber\\
&  -E_{F_{\boldsymbol{\theta}_{\ast,ETEL}}}\left[  \overline{\boldsymbol{t}%
}_{ET}^{T}(\boldsymbol{\theta}_{\ast,ETEL})\boldsymbol{g}(\boldsymbol{X}%
,\boldsymbol{\theta}_{\ast,ETEL})-\overline{\boldsymbol{t}}_{ET}%
^{T}(\boldsymbol{\theta}_{0})\boldsymbol{g}(\boldsymbol{X},\boldsymbol{\theta
}_{0})\right]  . \label{muG}%
\end{align}

\end{corollary}

\begin{proof}
With $\phi(x)=x\log x-x+1$ plugged into (\ref{psi})%
\[
\psi\left(  \frac{\mathrm{E}\left[  \exp\{\overline{\boldsymbol{t}}_{ET}%
^{T}(\boldsymbol{\theta}_{\ast,ETEL})\boldsymbol{g}(\boldsymbol{X}%
,\boldsymbol{\theta}_{\ast,ETEL})\right]  }{\exp\{\overline{\boldsymbol{t}%
}_{ET}^{T}(\boldsymbol{\theta}_{\ast,ETEL})\boldsymbol{g}(\boldsymbol{X}%
,\boldsymbol{\theta}_{\ast,ETEL})\}}\right)  =\frac{\exp\{\overline
{\boldsymbol{t}}_{ET}^{T}(\boldsymbol{\theta}_{\ast,ETEL})\boldsymbol{g}%
(\boldsymbol{X},\boldsymbol{\theta}_{\ast,ETEL})\}}{\mathrm{E}\left[
\exp\{\overline{\boldsymbol{t}}_{ET}^{T}(\boldsymbol{\theta}_{\ast
,ETEL})\boldsymbol{g}(\boldsymbol{X},\boldsymbol{\theta}_{\ast,ETEL}%
)\}\right]  }-1,
\]
is obtained, and then according to Theorem \ref{Th}, plugging%
\begin{align*}
&  \mathrm{E}\left[  \psi\left(  \frac{\mathrm{E}\left[  \exp\{\overline
{\boldsymbol{t}}_{ET}^{T}(\boldsymbol{\theta}_{\ast,ETEL})\boldsymbol{g}%
(\boldsymbol{X},\boldsymbol{\theta}_{\ast,ETEL})\right]  }{\exp\{\overline
{\boldsymbol{t}}_{ET}^{T}(\boldsymbol{\theta}_{\ast,ETEL})\boldsymbol{g}%
(\boldsymbol{X},\boldsymbol{\theta}_{\ast,ETEL})\}}\right)  \boldsymbol{G}%
_{\boldsymbol{X}}^{T}(\boldsymbol{\theta}_{\ast,ETEL})\right] \\
&  =\mathrm{E}^{-1}\left[  \exp\{\overline{\boldsymbol{t}}_{ET}^{T}%
(\boldsymbol{\theta}_{\ast,ETEL})\boldsymbol{g}(\boldsymbol{X}%
,\boldsymbol{\theta}_{\ast,ETEL})\right]  \mathrm{E}\left[  \exp
\{\overline{\boldsymbol{t}}_{ET}^{T}(\boldsymbol{\theta}_{\ast,ETEL}%
)\boldsymbol{g}(\boldsymbol{X},\boldsymbol{\theta}_{\ast,ETEL}%
)\}\boldsymbol{G}_{\boldsymbol{X}}^{T}(\boldsymbol{\theta}_{\ast
,ETEL})\right]  -\boldsymbol{S}_{12}^{T}(\boldsymbol{\theta}_{\ast,ETEL}),
\end{align*}%
\[
\mathrm{E}\left[  \psi\left(  \frac{\mathrm{E}\left[  \exp\{\overline
{\boldsymbol{t}}_{ET}^{T}(\boldsymbol{\theta}_{\ast,ETEL})\boldsymbol{g}%
(\boldsymbol{X},\boldsymbol{\theta}_{\ast,ETEL})\right]  }{\exp\{\overline
{\boldsymbol{t}}_{ET}^{T}(\boldsymbol{\theta}_{\ast,ETEL})\boldsymbol{g}%
(\boldsymbol{X},\boldsymbol{\theta}_{\ast,ETEL})\}}\right)  \boldsymbol{g}%
(\boldsymbol{X},\boldsymbol{\theta}_{\ast,ETEL})\right]  =-\mathrm{E}\left[
\boldsymbol{g}(\boldsymbol{X},\boldsymbol{\theta}_{\ast,ETEL})\right]  ,
\]%
\[
\mathrm{E}\left[  \psi\left(  \frac{\mathrm{E}\left[  \exp\{\overline
{\boldsymbol{t}}_{ET}^{T}(\boldsymbol{\theta}_{\ast,ETEL})\boldsymbol{g}%
(\boldsymbol{X},\boldsymbol{\theta}_{\ast,ETEL})\right]  }{\exp\{\overline
{\boldsymbol{t}}_{ET}^{T}(\boldsymbol{\theta}_{\ast,ETEL})\boldsymbol{g}%
(\boldsymbol{X},\boldsymbol{\theta}_{\ast,ETEL})\}}\right)  \right]  =0,
\]
into $\boldsymbol{r}_{1}^{\phi}(\boldsymbol{\theta}_{\ast,ETEL})$,
$\boldsymbol{r}_{2}^{\phi}(\boldsymbol{\theta}_{\ast,ETEL})$, $\boldsymbol{r}%
_{3}^{\phi}(\boldsymbol{\theta}_{\ast,ETEL})$\ of Theorem \ref{Th}%
\ respectively, the desired result is obtained. The expression of (\ref{muG})
is a particular case of $\mu_{T_{n}^{\phi}}(\boldsymbol{\theta}_{0}%
,\boldsymbol{\theta}_{\ast,ETEL})$ with $\phi(x)=x\log x-x+1$.
\end{proof}

\begin{remark}
From the previous two theorems, we can present an approximation of the power
function under misspecification $\beta_{T_{n}^{\phi}}(\boldsymbol{\theta
}_{\ast,ETEL})$, at $\boldsymbol{\theta}_{\ast,ETEL}\neq\boldsymbol{\theta
}_{0}$, of the empirical $\phi$-divergence test $T_{n}^{\phi}%
(\widehat{\boldsymbol{\theta}}_{ETEL},\boldsymbol{\theta}_{0})$ for a
significance level $\alpha$, as%
\[
\beta_{T_{n}^{\phi}}^{\ast}(\boldsymbol{\theta}_{\ast,ETEL})=1-\Phi\left(
\nu_{T_{n}^{\phi}}(\boldsymbol{\theta}_{\ast,ETEL}%
,\boldsymbol{\boldsymbol{\theta}_{0}})\right)  \simeq\beta_{T_{n}^{\phi}%
}(\boldsymbol{\theta}_{\ast,ETEL}),
\]
where
\[
\nu_{T_{n}^{\phi}}(\boldsymbol{\theta}_{\ast,ETEL}%
,\boldsymbol{\boldsymbol{\theta}_{0}})=\frac{n^{1/2}}{\sqrt{\boldsymbol{r}%
_{T_{n}^{\phi}}^{T}(\boldsymbol{\theta}_{\ast,ETEL})\boldsymbol{\Sigma}%
_{\sqrt{n}\widehat{\boldsymbol{\theta}}_{ETEL}}\boldsymbol{r}_{T_{n}^{\phi}%
}(\boldsymbol{\theta}_{\ast,ETEL})}}\left(  \frac{\phi^{\prime\prime}%
(1)T_{n}^{\phi}(\widehat{\boldsymbol{\theta}}_{ETEL},\boldsymbol{\theta}_{0}%
)}{2n}-\mu_{T_{n}^{\phi}}(\boldsymbol{\theta}_{0},\boldsymbol{\theta}%
_{\ast,ETEL})\right)  .
\]
\newline In a similar way, an approximation to the asymptotic power function
under misspecification $\beta_{S_{n}^{\phi}}(\boldsymbol{\theta}_{\ast,ETEL}%
)$, at $\boldsymbol{\theta}_{\ast,ETEL}\neq\boldsymbol{\theta}_{0}$, for the
empirical $\phi$-divergence test $T_{n}^{\phi}(\widehat{\boldsymbol{\theta}%
}_{ETEL},\boldsymbol{\theta}_{0})$ can be obtained\ as%
\[
\beta_{S_{n}^{\phi}}^{\ast}(\boldsymbol{\theta}_{\ast,ETEL})=1-\Phi\left(
\nu_{S_{n}^{\phi}}(\boldsymbol{\theta}_{\ast,ETEL}%
,\boldsymbol{\boldsymbol{\theta}_{0}})\right)  \simeq\beta_{S_{n}^{\phi}%
}(\boldsymbol{\theta}_{\ast,ETEL}),
\]
where
\[
\nu_{S_{n}^{\phi}}(\boldsymbol{\theta}_{\ast,ETEL}%
,\boldsymbol{\boldsymbol{\theta}_{0}})=\frac{n^{1/2}}{\sqrt{\boldsymbol{q}%
_{S_{n}^{\phi}}^{T}(\boldsymbol{\theta}_{\ast,ETEL},\boldsymbol{\theta}%
_{0})\boldsymbol{\Sigma}_{\sqrt{n}\widehat{\boldsymbol{\theta}}_{ETEL}%
}\boldsymbol{q}_{S_{n}^{\phi}}(\boldsymbol{\theta}_{\ast,ETEL}%
,\boldsymbol{\theta}_{0})}}\left(  \frac{\phi^{\prime\prime}(1)S_{n}^{\phi
}(\widehat{\boldsymbol{\theta}}_{ETEL},\boldsymbol{\theta}_{0})}{2n}%
-\mu_{S_{n}^{\phi}}(\boldsymbol{\theta}_{0},\boldsymbol{\theta}_{\ast
,ETEL})\right)  .
\]
In practice, $\beta_{T_{n}^{\phi}}^{\ast}(\boldsymbol{\theta}_{\ast,ETEL}%
)$\ and $\beta_{S_{n}^{\phi}}^{\ast}(\boldsymbol{\theta}_{\ast,ETEL})$ are
unknown but their consistent estimators are obtained by replacing the
population mean by the sample mean.
\end{remark}

\begin{remark}
The class of $\phi$-divergence measures is a wide family of divergence
measures but unfortunately there are some classical divergence measures that
are not included in this family of $\phi$-divergence measures such as the
R\'{e}nyi's divergence or the Sharma and Mittal's divergence. The expression
of R\'{e}nyi's divergence is given by
\begin{equation}
D_{\text{\textrm{R\'{e}nyi}}}^{a}\left(  \boldsymbol{p}_{ET}\left(
\boldsymbol{\theta}\right)  ,\boldsymbol{p}_{ET}\left(  \boldsymbol{\theta
}_{0}\right)  \right)  =\frac{1}{a\left(  a-1\right)  }\log\sum\limits_{i=1}%
^{n}\boldsymbol{p}_{ET}\left(  \boldsymbol{\theta}\right)  ^{a}\boldsymbol{p}%
_{ET}\left(  \boldsymbol{\theta}_{0}\right)  ^{1-a}\text{, if }a\neq0,1,
\label{renyi}%
\end{equation}
with
\[
D_{\text{\textrm{R\'{e}nyi}}}^{0}\left(  \boldsymbol{p}_{ET}\left(
\boldsymbol{\theta}\right)  ,\boldsymbol{p}_{ET}\left(  \boldsymbol{\theta
}_{0}\right)  \right)  =\lim_{a\rightarrow0}D_{\text{\textrm{R\'{e}nyi}}%
}\left(  \boldsymbol{p}_{ET}\left(  \boldsymbol{\theta}\right)
,\boldsymbol{p}_{ET}\left(  \boldsymbol{\theta}_{0}\right)  \right)
=D_{\mathrm{Kull}}\left(  \boldsymbol{p}_{ET}\left(  \boldsymbol{\theta
}\right)  ,\boldsymbol{p}_{ET}\left(  \boldsymbol{\theta}_{0}\right)  \right)
\]
and
\[
D_{\text{\textrm{R\'{e}nyi}}}^{1}\left(  \boldsymbol{p}_{ET}\left(
\boldsymbol{\theta}\right)  ,\boldsymbol{p}_{ET}\left(  \boldsymbol{\theta
}_{0}\right)  \right)  =\lim_{a\rightarrow1}D_{\text{\textrm{R\'{e}nyi}}%
}\left(  \boldsymbol{p}_{ET}\left(  \boldsymbol{\theta}\right)
,\boldsymbol{p}_{ET}\left(  \boldsymbol{\theta}_{0}\right)  \right)
=D_{\mathrm{Kull}}\left(  \boldsymbol{p}_{ET}\left(  \boldsymbol{\theta}%
_{0}\right)  ,\boldsymbol{p}_{ET}\left(  \boldsymbol{\theta}\right)  \right)
.
\]
This measure of divergence was introduced in R\'{e}nyi (1961) for $a>0$ and
$a\neq1$ and Liese and Vajda (1987) extended it for all $a\neq1,0$. An
interesting divergence measure related to R\'{e}nyi divergence measure is the
Bhattacharya divergence defined as the R\'{e}nyi divergence for $a=1/2$
divided by $4.$ Other interesting example of divergence measure that is not
included in the family of $\phi$-divergence measures is the divergence
measures introduced by Sharma and Mittal (1997).
\end{remark}

In order to unify the previous divergence measures, as well as another
divergence measures, Men\'{e}ndez et al. (1995, 1997) introduced the family of
divergences called \textquotedblleft$(h,\phi)$-divergence
measures\textquotedblright\ in the following way%

\[
D_{\phi}^{h}\left(  \boldsymbol{p}_{ET}\left(  \boldsymbol{\theta}\right)
,\boldsymbol{p}_{ET}\left(  \boldsymbol{\theta}_{0}\right)  \right)  =h\left(
D_{\phi}\left(  \boldsymbol{p}_{ET}\left(  \boldsymbol{\theta}\right)
,\boldsymbol{p}_{ET}\left(  \boldsymbol{\theta}_{0}\right)  \right)  \right)
,
\]
where $h$ is a differentiable increasing function mapping from $\left[
0,\phi\left(  0\right)  +\lim_{t\rightarrow\infty}\frac{\phi\left(  t\right)
}{t}\right]  $ onto $\left[  0,\infty\right)  $, with $h(0)=0$, $h^{\prime
}(0)>0$, and $\phi\in\Phi$. In Table \ref{t1}, these divergence measures are
presented, along with the corresponding expressions of $h$ and $\phi$.\newline%
\begin{table}[htbp] \tabcolsep0.8pt  \centering
$%
\begin{tabular}
[c]{ccccc}\hline
Divergence & \hspace*{0.5cm} & $h\left(  x\right)  $ & \hspace*{0.5cm} &
$\phi\left(  x\right)  $\\\hline
\multicolumn{1}{l}{R\'{e}nyi} &  & \multicolumn{1}{l}{$\frac{1}{a\left(
a-1\right)  }\log\left(  a\left(  a-1\right)  x+1\right)  ,\quad a\neq0,1$} &
& \multicolumn{1}{l}{$\frac{x^{a}-a\left(  x-1\right)  -1}{a\left(
a-1\right)  },\quad a\neq0,1$}\\
\multicolumn{1}{l}{Sharma-Mittal} &  & \multicolumn{1}{l}{$\frac{1}%
{b-1}\left\{  [1+a\left(  a-1\right)  x]^{\frac{b-1}{a-1}}-1\right\}  ,\quad
b,a\neq1$} &  & \multicolumn{1}{l}{$\frac{x^{a}-a\left(  x-1\right)
-1}{a\left(  a-1\right)  },\quad a\neq0,1$}\\\hline
\end{tabular}
$\caption{Some specific
$(h,\phi)$-divergence measures.\label{t1}}%
\end{table}%
\newline

Based on the $(h,\phi)$-divergence measures we can define two new families of
empirical $(h,\phi)$-divergence test statistics,
\begin{equation}
S_{n}^{\phi,h}\left(  \widehat{\boldsymbol{\theta}}_{ETEL},\boldsymbol{\theta
}_{0}\right)  =\frac{2n}{\phi^{\prime\prime}(1)h^{\prime}(0)}h\left(  D_{\phi
}\left(  \left(  \boldsymbol{p}_{ET}\left(  \widehat{\boldsymbol{\theta}%
}_{ETEL}\right)  ,\boldsymbol{p}_{ET}\left(  \boldsymbol{\theta}_{0}\right)
\right)  \right)  \right)  \label{Ah}%
\end{equation}
and
\begin{equation}
T_{n}^{\phi,h}\left(  \widehat{\boldsymbol{\theta}}_{ETEL},\boldsymbol{\theta
}_{0}\right)  =\frac{2n}{\phi^{\prime\prime}(1)h^{\prime}(0)}\left(  h\left(
D_{\phi}\left(  \boldsymbol{u},\boldsymbol{p}_{ET}\left(  \boldsymbol{\theta
}_{0}\right)  \right)  \right)  -h\left(  D_{\phi}\left(  \boldsymbol{u}%
,\boldsymbol{p}_{ET}\left(  \widehat{\boldsymbol{\theta}}_{ETEL}\right)
\right)  \right)  \right)  . \label{Bh}%
\end{equation}
The results obtained in this paper for the empirical $\phi$-divergence test
statistics $T_{n}^{\phi}(\widehat{\boldsymbol{\theta}}_{ETEL}%
,\boldsymbol{\theta}_{0})$ and $S_{n}^{\phi}(\widehat{\boldsymbol{\theta}%
}_{ETEL},\boldsymbol{\theta}_{0})$ can be obtained for the empirical
$(h,\phi)$-divergence test statistics defined in (\ref{Ah}) and (\ref{Bh}).

\section{Simulation study\label{Simulation}}

The aim of this simulation study is to analyze the performance of the
empirical $\phi$-divergence test-statistics when the ETEL estimator of an
unknown parameter is considered. In this regard, robustness under
misspecification and efficiency are studied, based on the design of the
simulation study given in Schennach (2007). Let $X$ be an unknown univariate
random variable, with mean $\theta\in%
\mathbb{R}
$ and variance $\sigma^{2}\in%
\mathbb{R}
^{+}$ both unknown, but it is supposed to be known that $\sigma^{2}=\theta
^{2}+1$. The corresponding moment based vectorial estimating function is
$\boldsymbol{g}(X,\theta)=\boldsymbol{0}_{2}$, with $\boldsymbol{g}%
(X,\theta)=(g_{1}(X,\theta),g_{2}(X,\theta))^{T}$,
\begin{align}
g_{1}(X,\theta)  &  =X-\theta,\label{g1}\\
g_{2}(X,\theta)  &  =X^{2}-2\theta^{2}-1. \label{g2}%
\end{align}
By modifying (\ref{g2}) to%
\begin{equation}
g_{2}(X,\theta)=X^{2}-2\theta^{2}-\delta,\quad\delta\in(-2\theta^{2}%
,\infty)-\{1\}, \label{g2m}%
\end{equation}
we are considering a misspecified model, with $\delta$ being a tuning
parameter for the model misspecification degree. Since the correctly specified
model has a variance equal to $\theta^{2}+\delta$ with $\delta=1$, less
variance than the correct one is specified when $\delta\in(-2\theta^{2},1)$,
while a bigger variance than the correct one is specified when $\delta
\in(1,\infty)$. The EL estimator of $\theta$ is given by%
\[
\widehat{\theta}_{EL}=\arg\min_{\theta\in%
\mathbb{R}
}\left(  -%
{\textstyle\sum\limits_{i=1}^{n}}
\log p_{i,EL}(\theta)\right)  ,
\]
with%
\begin{align}
&  p_{i,EL}(\theta)=\frac{1}{n}\frac{1}{1+\sum_{h=1}^{2}t_{h,EL}(\theta
)g_{h}(x_{i},\theta)},\quad i=1,...,n,\label{ppEL}\\
&  t_{1,EL}(\theta),\;t_{2,EL}(\theta)\quad\text{s.t.}\sum_{i=1}^{n}\frac
{1}{1+\sum_{h=1}^{2}t_{h,EL}(\theta)g_{h}(x_{i},\theta)}g_{r}(x_{i}%
,\theta)=0,\quad r=1,2,\nonumber
\end{align}
the ET estimator of $\theta$ by%
\[
\widehat{\theta}_{ET}=\arg\min_{\theta\in%
\mathbb{R}
}%
{\textstyle\sum\limits_{i=1}^{n}}
p_{ET,i}(\boldsymbol{\theta})\log\left(  p_{ET,i}(\boldsymbol{\theta})\right)
,
\]
with%
\begin{align}
&  p_{i,ET}(\theta)=\frac{\exp\left\{
{\textstyle\sum\nolimits_{h=1}^{2}}
t_{h,ET}(\theta)g_{h}(x_{i},\theta)\right\}  }{%
{\textstyle\sum\nolimits_{i=1}^{n}}
\exp\left\{
{\textstyle\sum\nolimits_{h=1}^{2}}
t_{h,ET}(\theta)g_{h}(x_{i},\theta)\right\}  },\quad i=1,...,n,\label{ppET}\\
&  t_{1,ET}(\theta),\;t_{2,ET}(\theta)\quad\text{s.t.}\sum_{i=1}^{n}%
\exp\left\{
{\textstyle\sum\nolimits_{h=1}^{2}}
t_{h,ET}(\theta)g_{h}(x_{i},\theta)\right\}  g_{r}(x_{i},\theta)=0,\quad
r=1,2,\nonumber
\end{align}
and the ETEL of $\theta$ estimator by%
\[
\widehat{\theta}_{ETEL}=\arg\min_{\theta\in%
\mathbb{R}
}\left(  -%
{\textstyle\sum\limits_{i=1}^{n}}
\log p_{i,ETEL}(\theta)\right)  ,
\]
with $p_{i,ETEL}(\theta)=p_{i,ET}(\theta)$, $i=1,...,n$. The test-statistics
$T_{n}^{\phi_{\lambda}}(\widehat{\theta}_{\ell},\theta_{0})$ and $S_{n}%
^{\phi_{\lambda}}(\widehat{\theta}_{\ell},\theta_{0})$, with $\ell
\in\{EL,ET,ETEL\}$, and%
\[
\phi_{\lambda}(x)=\left\{
\begin{array}
[c]{ll}%
\frac{1}{\lambda(\lambda+1)}\left(  x^{\lambda+1}-x-\lambda(x-1)\right)  , &
\lambda\in%
\mathbb{R}
-\{0,-1\}\\
\lim_{s\rightarrow0}\phi_{s}(x)=x\log x-x+1, & \lambda=0\\
\lim_{s\rightarrow-1}\phi_{s}(x)=-\log x+x-1, & \lambda=-1
\end{array}
\right.  ,
\]
are the so-called empirical power divergence based test-statistics of Cressie
and Read (1984), valid in this new setting for testing%
\begin{equation}
H_{0}:\theta=\theta_{0}\text{\quad vs.\quad}H_{1}:\theta\neq\theta_{0}\text{,
with }\theta_{0}=0. \label{test}%
\end{equation}
The expressions of the empirical power divergence based test-statistics are%
\[
T_{n}^{\phi_{\lambda}}(\widehat{\theta}_{\ell},\theta_{0})=\left\{
\begin{array}
[c]{ll}%
\frac{2}{\lambda(1+\lambda)}\left(  \sum\limits_{i=1}^{n}\left(  np_{i,\ell
}(\theta_{0})\right)  ^{-\lambda}-\sum\limits_{i=1}^{n}\left(  np_{i,\ell
}(\widehat{\theta}_{\ell})\right)  ^{-\lambda}\right)  , & \lambda\in%
\mathbb{R}
-\{0,-1\}\\
2\sum\limits_{i=1}^{n}\log\left(  \frac{p_{i,\ell}(\widehat{\theta}_{\ell}%
)}{p_{i,\ell}(\theta_{0})}\right)  , & \lambda=0\\
2n\left(  \sum\limits_{i=1}^{n}p_{i,\ell}(\theta_{0})\log\left(  np_{i,\ell
}(\theta_{0})\right)  -\sum\limits_{i=1}^{n}p_{i,\ell}(\widehat{\theta}_{\ell
})\log(np_{i,\ell}(\widehat{\theta}_{\ell}))\right)  , & \lambda=-1
\end{array}
\right.  ,
\]

\[
S_{n}^{\phi_{\lambda}}(\widehat{\theta}_{\ell},\theta_{0})=\left\{
\begin{array}
[c]{ll}%
\frac{2n}{\lambda(1+\lambda)}\left(  \sum\limits_{i=1}^{n}\frac{p_{i,\ell
}^{\lambda+1}(\widehat{\theta}_{\ell})}{p_{i,\ell}^{\lambda}(\theta_{0}%
)}-1\right)  , & \lambda\in%
\mathbb{R}
-\{0,-1\}\\
2n\sum\limits_{i=1}^{n}p_{i,\ell}(\widehat{\theta}_{\ell})\log\left(
\frac{p_{i,\ell}(\widehat{\theta}_{\ell})}{p_{i,\ell}(\theta_{0})}\right)  , &
\lambda=0\\
2n\sum\limits_{i=1}^{n}p_{i,\ell}(\theta_{0})\log\left(  \frac{p_{i,\ell
}(\theta_{0})}{p_{i,\ell}(\widehat{\theta}_{\ell})}\right)  , & \lambda=-1
\end{array}
\right.  ,
\]
with $\ell\in\{EL,ET,ETEL\}$, $p_{i,EL}(\theta)$ given by (\ref{ppEL}) and
$p_{i,ETEL}(\theta)=p_{i,ET}(\theta)$ by (\ref{ppET}). It is worth of
mentioning that the empirical likelihood ratio test-statistic of Qin and
Lawless (1994) matches the case of $\lambda=0$ when the EL estimator of
$\theta$ is applied, i.e. $T_{n}^{\phi_{0}}(\widehat{\theta}_{EL},\theta
_{0})=G_{n}^{2}(\widehat{\theta}_{EL},\theta_{0})$.

For the study of the performance of $T_{n}^{\phi_{\lambda}}(\widehat{\theta
}_{\ell},\theta_{0})$\ and $S_{n}^{\phi_{\lambda}}(\widehat{\theta}_{\ell
},\theta_{0})$, for illustrative purposes, a subset of tuning parameters of
the empirical power divergence based test-statistics are considered,
$\lambda\in\{-1,-0.5,0,\frac{2}{3}\}$. When the model is correctly specified,
the population's distribution is simulated with a standard normal
distribution, i.e. $X\sim\mathcal{N}(\theta,\theta^{2}+\delta)$, with
$\theta=0$ and $\delta=1$ ($\sigma^{2}=1$). When the model is misspecified,
two cases are considered, by simulating the population distribution either
through $X\sim\mathcal{N}(\theta,\theta^{2}+\delta)$, with $\theta=0$ and
$\delta=0.7$ ($\sigma^{2}=0.7<1$) or $\theta=0$ and $\delta=1.3$ ($\sigma
^{2}=1.3>1$). The pseudo true value of the ETEL estimator is $\theta
_{\ast,ETEL}=\theta_{0}=0$ for $\delta>\frac{1}{2}$, and $t_{\ast,1,ETEL}=0$,
$t_{\ast,2,ETEL}=\frac{1-\delta}{2\delta}$, so even being a misspecified model
$\widehat{\theta}_{ETEL}$ is a consistent estimator of the true value of
$\theta$. Using $R=10,000$ replications, the following results are obtained.

In Figure \ref{fig4} the simulated cumulative distribution functions (CDF) of
$\widehat{\theta}_{EL}$, $\widehat{\theta}_{ET}$ and $\widehat{\theta}_{ETEL}$
are shown with a sample size of $n=1000$, for the correctly specified model
($\delta=1$) as well as the two misspecified models ($\delta\in\{0.7,1.3\}$).
Since the sample size is very big, the three types of estimators exhibit
almost the same CDF. The gray color line of the figures indicates the
theoretical distribution with correct specification, i.e. the reference line
to be compared. Under misspecification, as expected according to Schennach
(2007), the most robust estimator under misspecification is $\widehat{\theta
}_{ET}$ (it is closer to the gray line), the least robust $\widehat{\theta
}_{EL}$ (it is further from the gray line), and $\widehat{\theta}_{ETEL}$
tends to be between the two. In addition, $\widehat{\theta}_{ETEL}$\ tends to
be in between the two in efficiency with respect to the exact size of the
asymptotic test for small sample sizes, no as efficient as $\widehat{\theta
}_{EL}$ but more efficient than $\widehat{\theta}_{ET}$. In the same way, we
would like to identify a test-statistic $T_{n}^{\phi_{\lambda}}%
(\widehat{\theta}_{ETEL},\theta_{0})$ or $S_{n}^{\phi_{\lambda}}%
(\widehat{\theta}_{ETEL},\theta_{0})$ with good performance at the same in
robustness under misspecification and efficiency.%

\begin{figure}[htbp]  \tabcolsep2.8pt  \centering
\begin{tabular}
[c]{c}%
${%
{\includegraphics[
height=2.7821in,
width=5.4405in
]%
{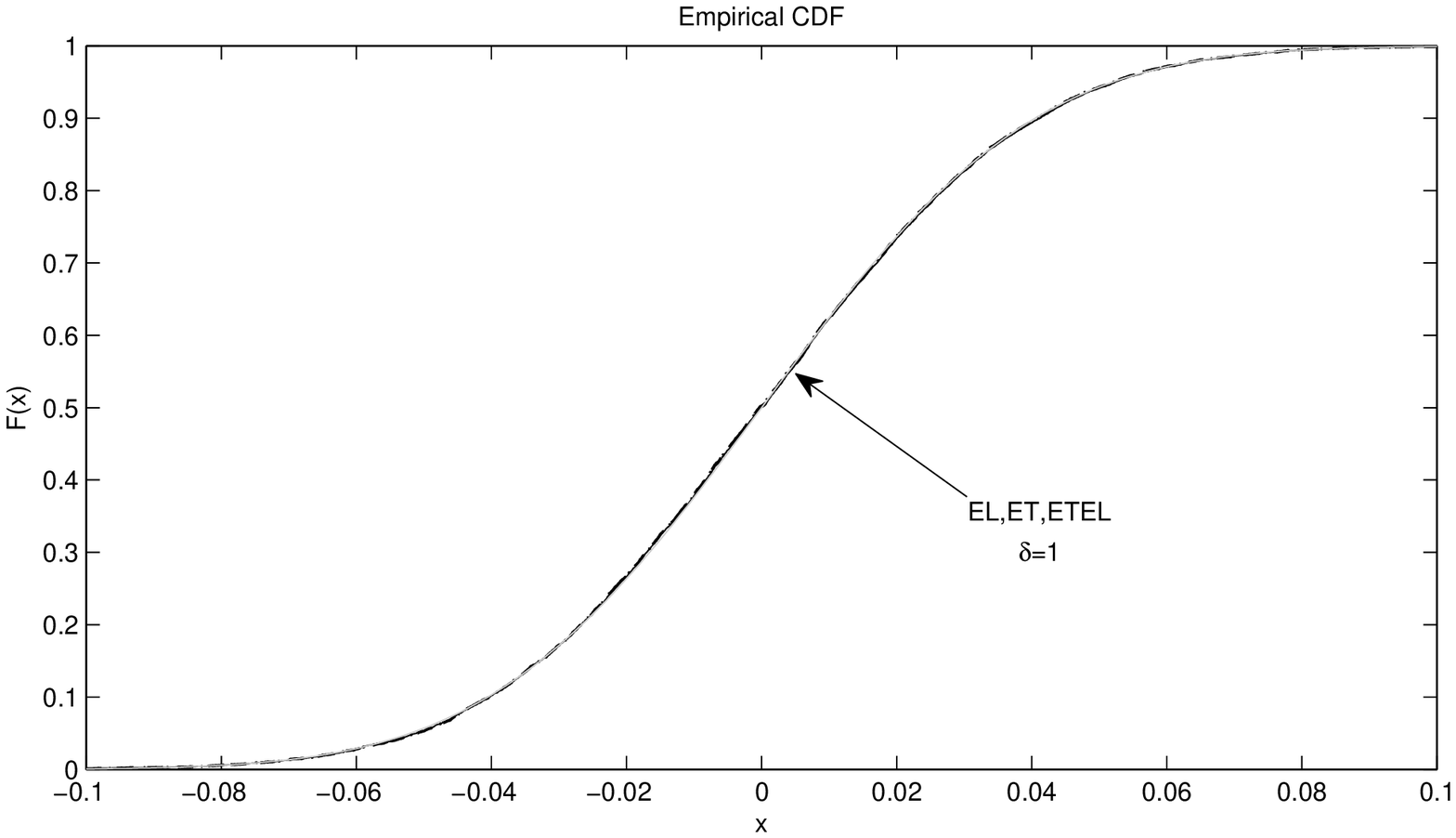}%
}%
}$\\%
{\includegraphics[
height=3.0199in,
width=5.9084in
]%
{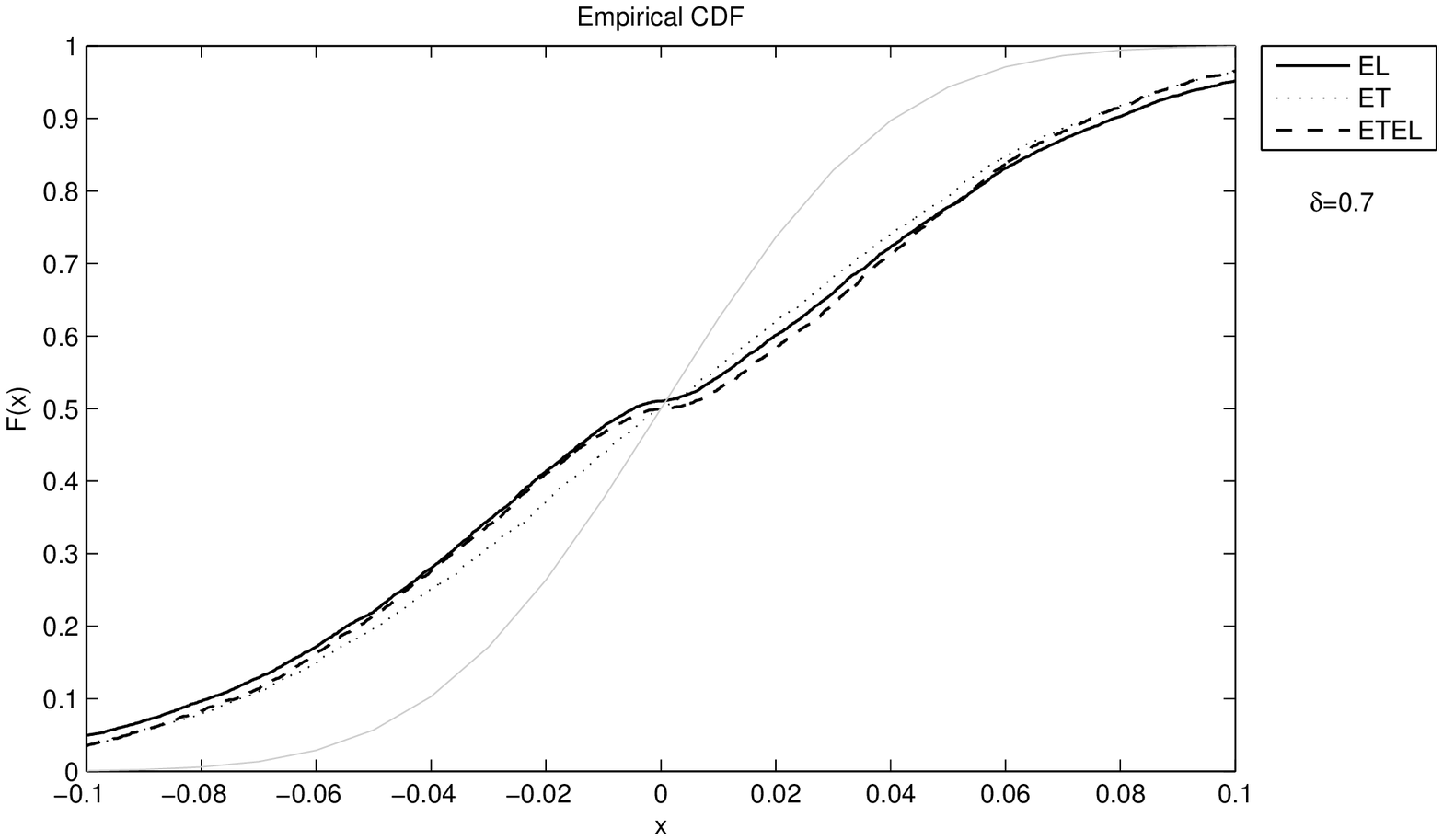}%
}%
\\%
{\includegraphics[
height=3.0199in,
width=5.9084in
]%
{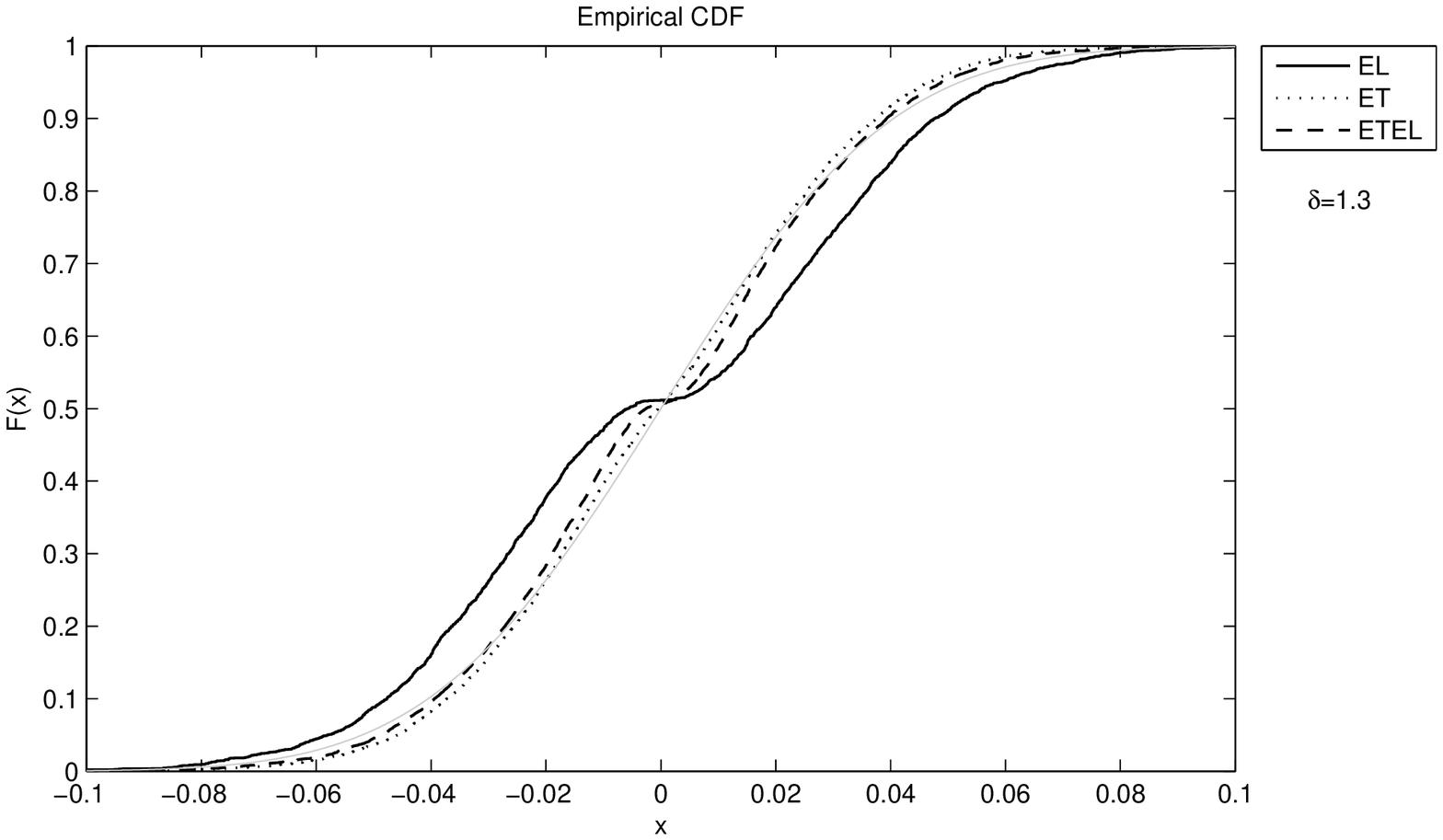}%
}%
\end{tabular}
\caption{Cumulative distribution function of the three types of estimators, for $n=1000$, when the model is correctly specilied (top), and is misspecified ($\delta=0.7$, $\delta=1.3$). \label{fig4}}%
\end{figure}%

The simulations showed that in robustness under misspecification $S_{n}%
^{\phi_{\lambda}}(\widehat{\theta}_{\ell},\theta_{0})$ is much worse than
$T_{n}^{\phi_{\lambda}}(\widehat{\theta}_{\ell},\theta_{0})$, with $\ell
\in\{EL,ET,ETEL\}$, for this reason the following figures are focussed only on
$T_{n}^{\phi_{\lambda}}(\widehat{\theta}_{\ell},\theta_{0})$. In Figure
\ref{fig6} the simulated CDFs of $T_{n}^{\phi_{\lambda}}(\widehat{\theta
}_{\ell},\theta_{0})$ are plotted with the three types of estimators and a
degree of misspecification equal to $\delta=1.3$, while in Figure \ref{fig5}
are plotted with a degree of misspecification equal to $\delta=0.7$. From
them, the test-statistic $T_{n}^{\phi_{\lambda}}(\widehat{\theta}_{\ell
},\theta_{0})$ with $\lambda=-1$ seems to be the most robust test-statistic
under misspecification. Figure \ref{fig7} has been plotted to compare the
performance of $T_{n}^{\phi_{\lambda}}(\widehat{\theta}_{\ell},\theta_{0})$
with $\lambda=-1$ when different types of estimators are plugged, $\ell
\in\{EL,ET,ETEL\}$. As expected, the most robust test-statistic is
$T_{n}^{\phi_{-1}}(\widehat{\theta}_{ET},\theta_{0})$, the worst one
$T_{n}^{\phi_{-1}}(\widehat{\theta}_{EL},\theta_{0})$, and $T_{n}^{\phi_{-1}%
}(\widehat{\theta}_{ETEL},\theta_{0})$ is in between. From Figures \ref{fig6}
and \ref{fig5}, for the misspecified model (either with $\delta=1.3$ or
$\delta=0.7$), the exact significance levels can be visually compared with
respect to the $0.05$ nominal level, comparing the values of the black color
curves just at $\chi_{0.05}^{2}=3.84$\ in the abscissa axis, with respect to
the gray color curve. In this regard, the exact sizes for $\delta=1.3$ are
better than for $\delta=0.7$: for ETEL estimators the exact significance
levels are $0.048$ ($\lambda=-1$), $0.036$ ($\lambda=-0.5$), $0.031$
($\lambda=0$), $0.025$ ($\lambda=\frac{2}{3}$) when $\delta=1.3$ and $0.176$
($\lambda=-1$), $0.208$ ($\lambda=-0.5$), $0.258$ ($\lambda=0$), $0.391$
($\lambda=\frac{2}{3}$)\ when $\delta=0.7$. The figures of the simulations for
$S_{n}^{\phi_{\lambda}}(\widehat{\theta}_{\ell},\theta_{0})$, with $n=1000$,
were omitted, but the exact sizes are as follows: the exact significance
levels are $0.017$ ($\lambda=-1$), $0.017$ ($\lambda=-0.5$), $0.017$
($\lambda=0$), $0.017$ ( $\lambda=\frac{2}{3}$) when $\delta=1.3$ and $0.417$
($\lambda=-1$), $0.417$ ($\lambda=-0.5$), $0.418$ ($\lambda=0$), $0.419$
($\lambda=\frac{2}{3}$)\ when $\delta=0.7$.%

\begin{figure}[htbp]  \tabcolsep2.8pt  \centering
\begin{tabular}
[c]{c}%
{\includegraphics[
height=2.8764in,
width=5.6273in
]%
{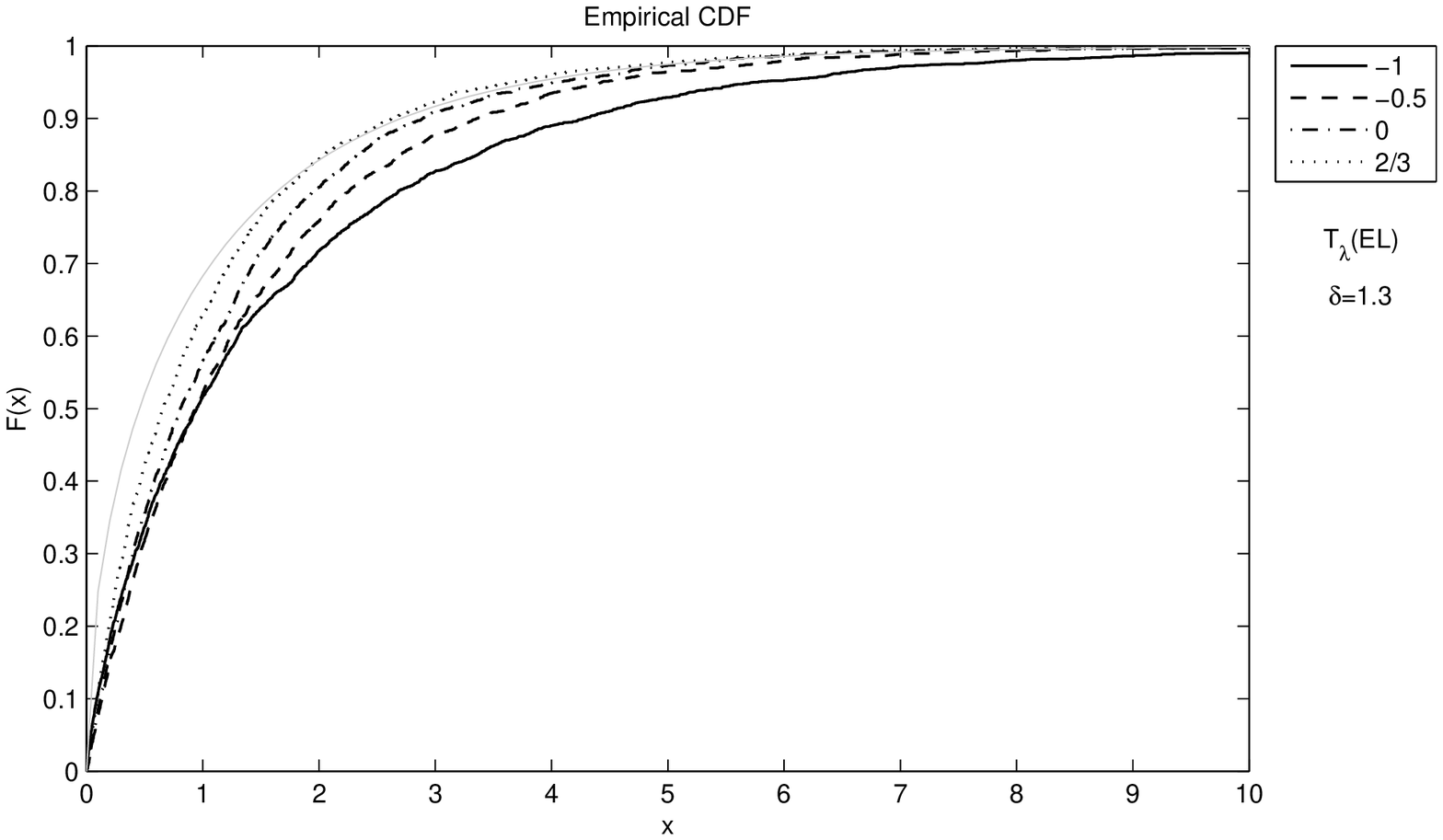}%
}%
\\%
{\includegraphics[
height=2.8764in,
width=5.6273in
]%
{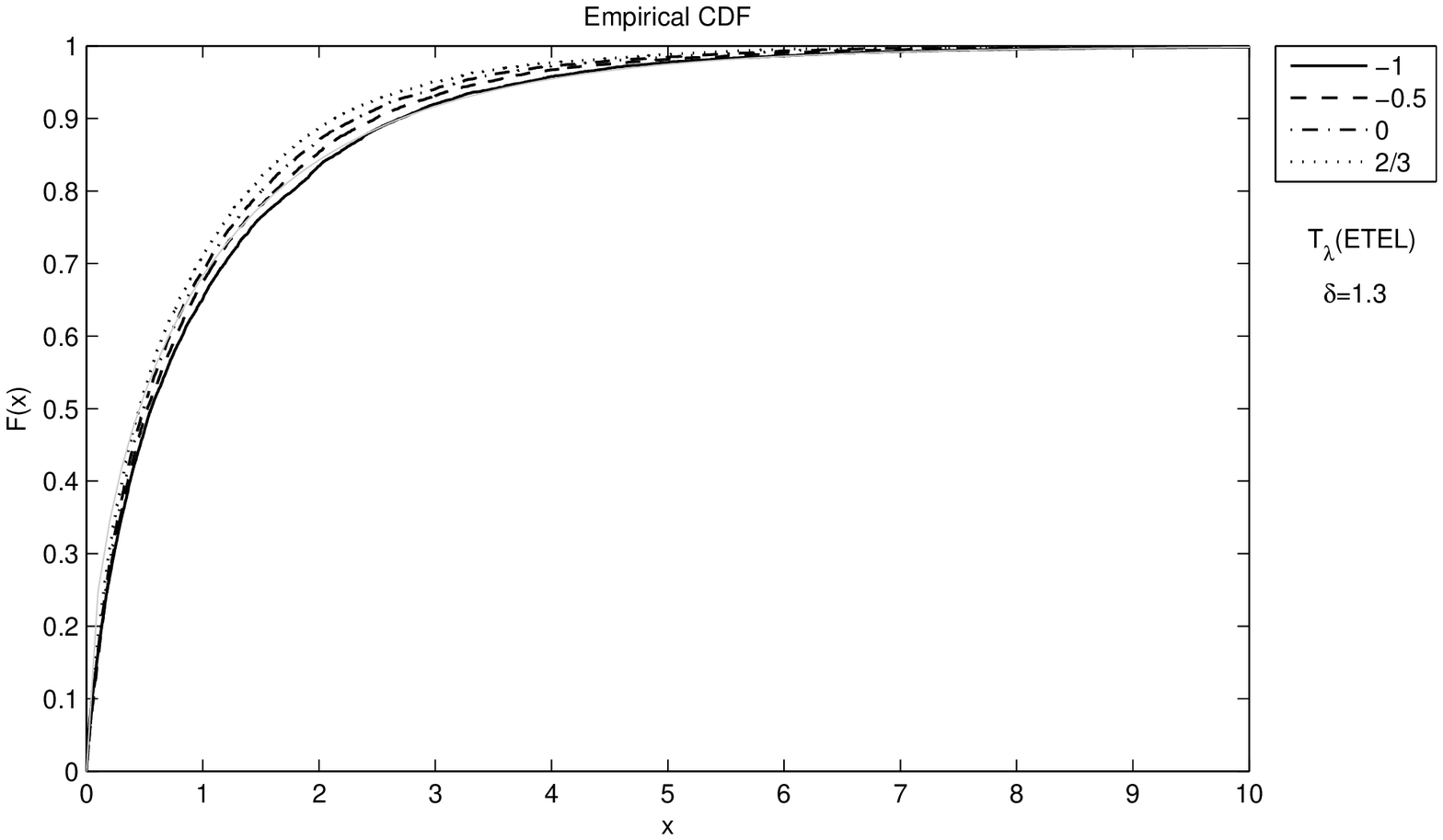}%
}%
\\%
{\includegraphics[
height=2.8764in,
width=5.6273in
]%
{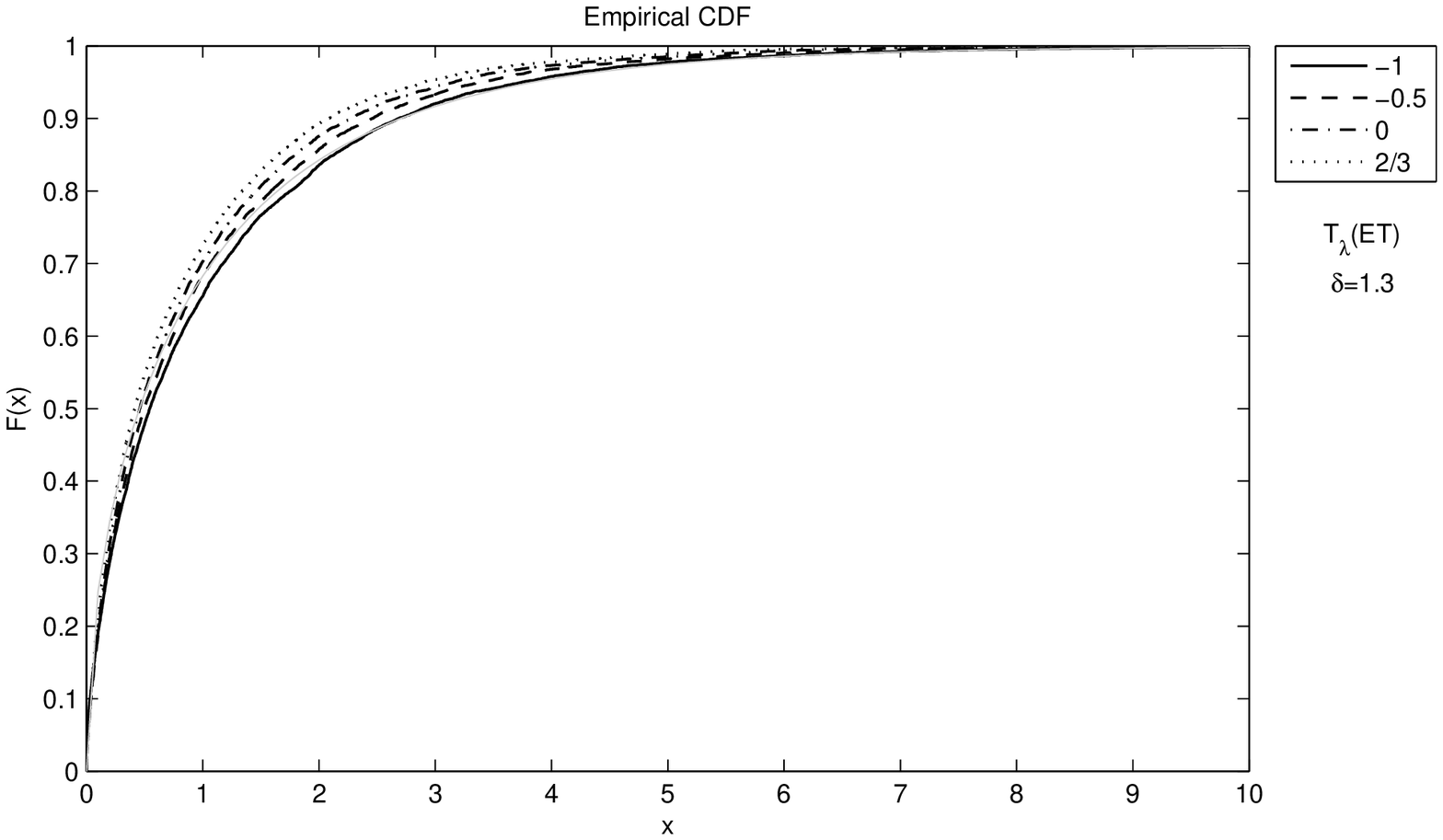}%
}%
\end{tabular}
\caption{Cumulative distribution function of the empirical power divergence based test-statistics with the three types of estimators, for $n=1000$, when the model is misspecified with $\delta=1.3$. \label{fig6}}%
\end{figure}%
%

\begin{figure}[htbp]  \tabcolsep2.8pt  \centering
\begin{tabular}
[c]{c}%
{\includegraphics[
height=2.8764in,
width=5.6273in
]%
{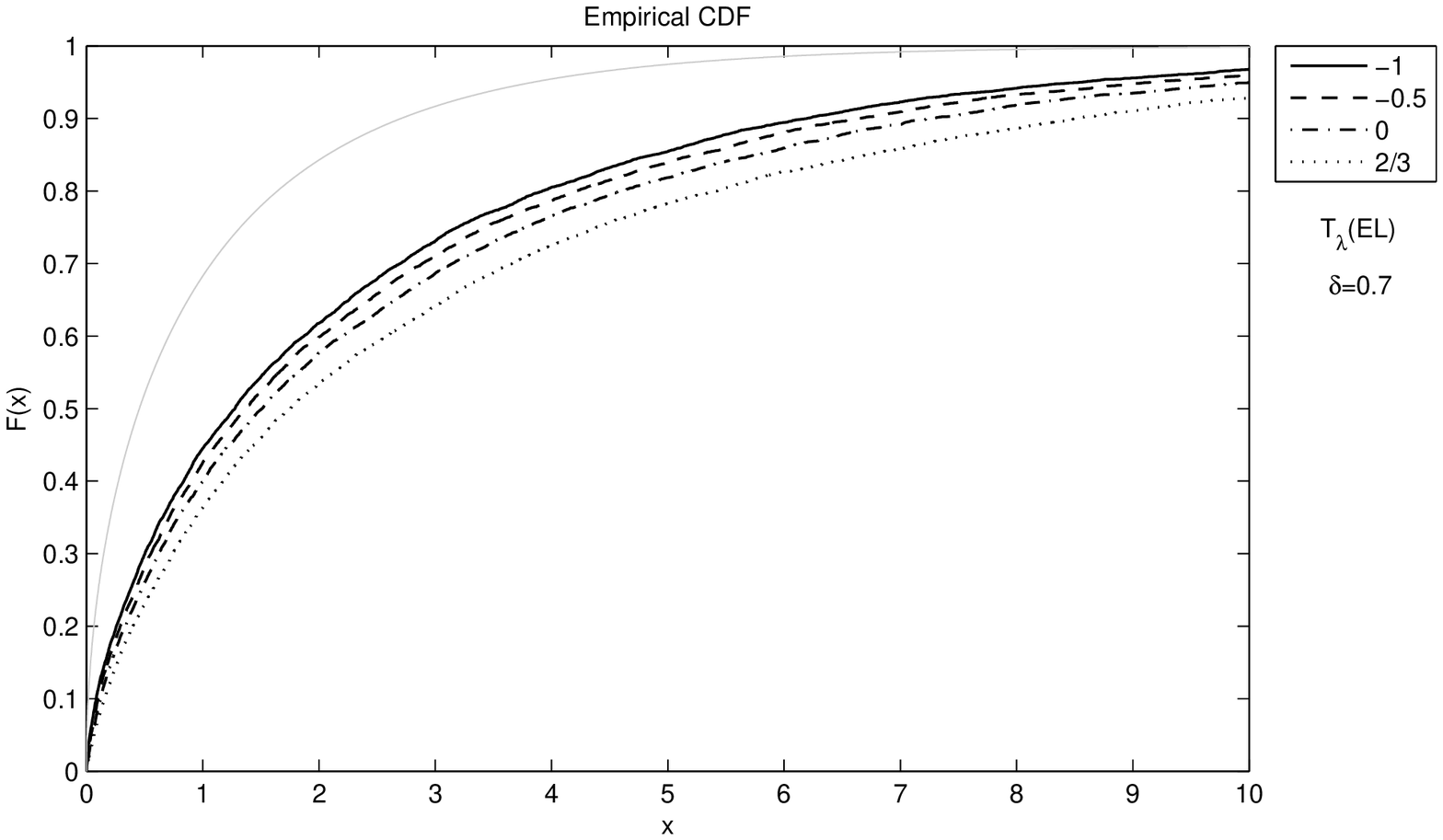}%
}%
\\%
{\includegraphics[
height=2.8764in,
width=5.6273in
]%
{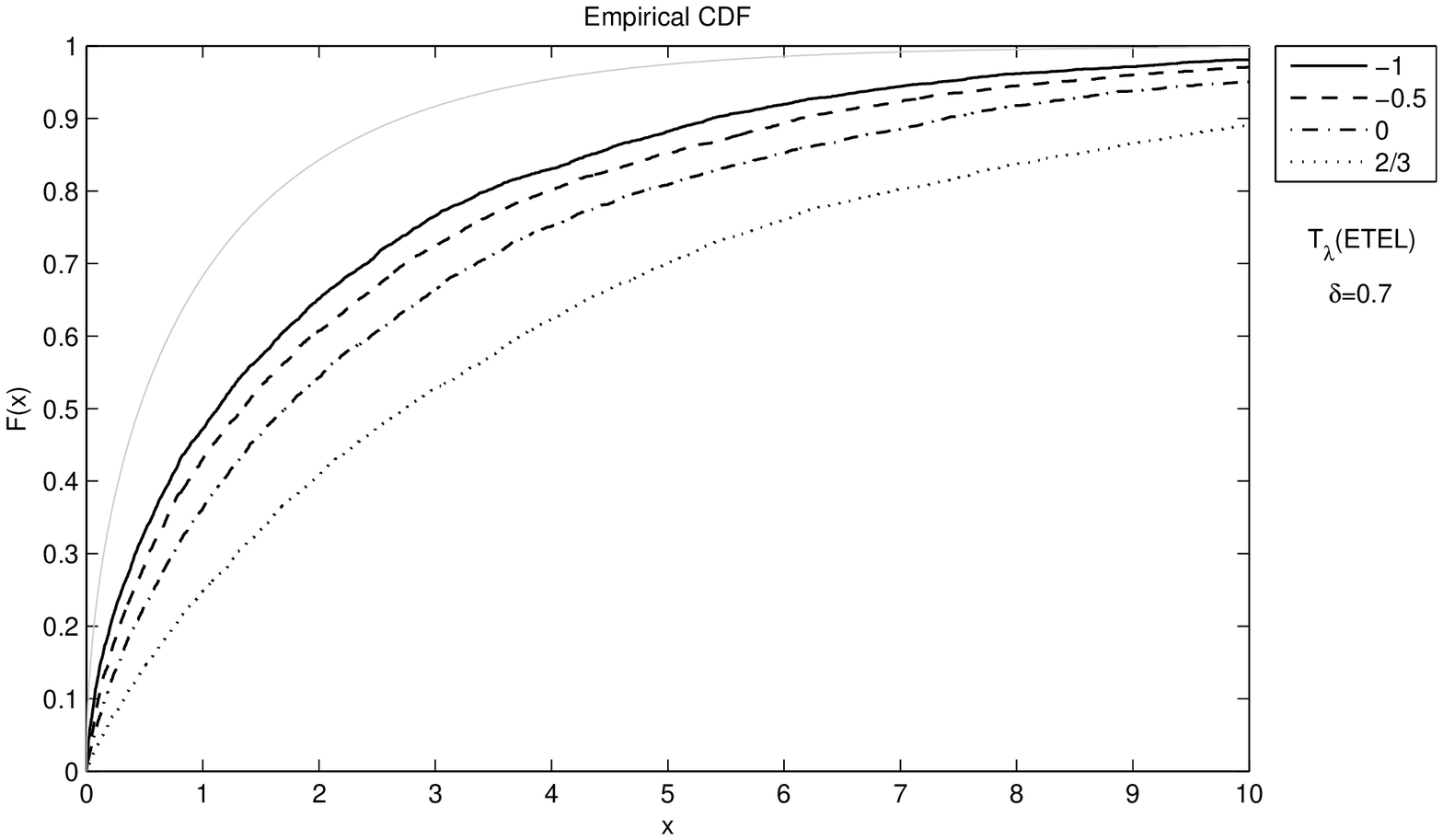}%
}%
\\%
{\includegraphics[
height=2.8764in,
width=5.6273in
]%
{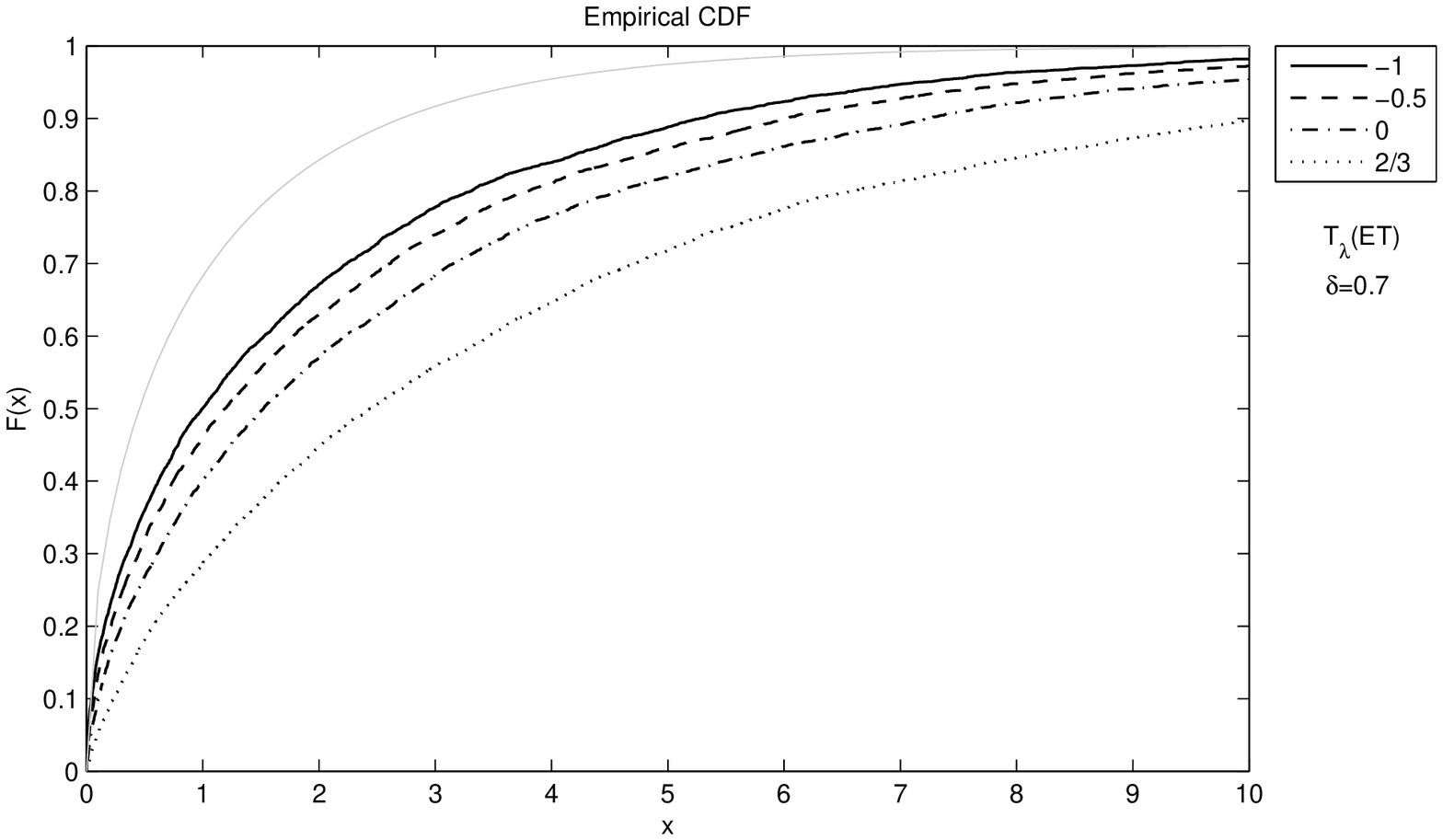}%
}%
\end{tabular}
\caption{Cumulative distribution function of the empirical power divergence based test-statistics with the three types of estimators, for $n=1000$, when the model is misspecified with $\delta=0.7$. \label{fig5}}%
\end{figure}%
%

\begin{figure}[htbp]  \tabcolsep2.8pt  \centering
\begin{tabular}
[c]{c}%
{\includegraphics[
height=2.8764in,
width=5.6273in
]%
{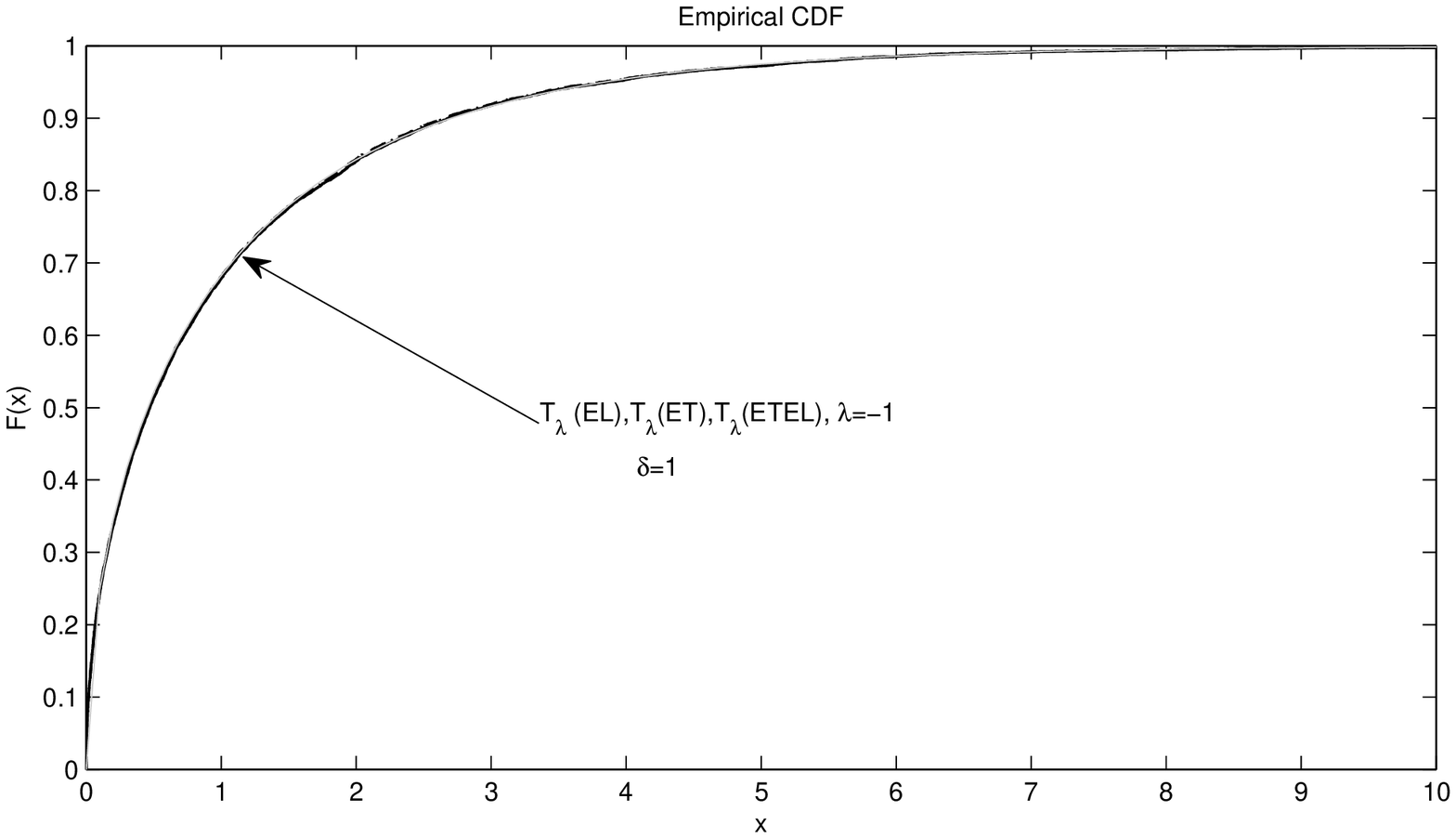}%
}%
\\%
{\includegraphics[
height=2.8764in,
width=5.6273in
]%
{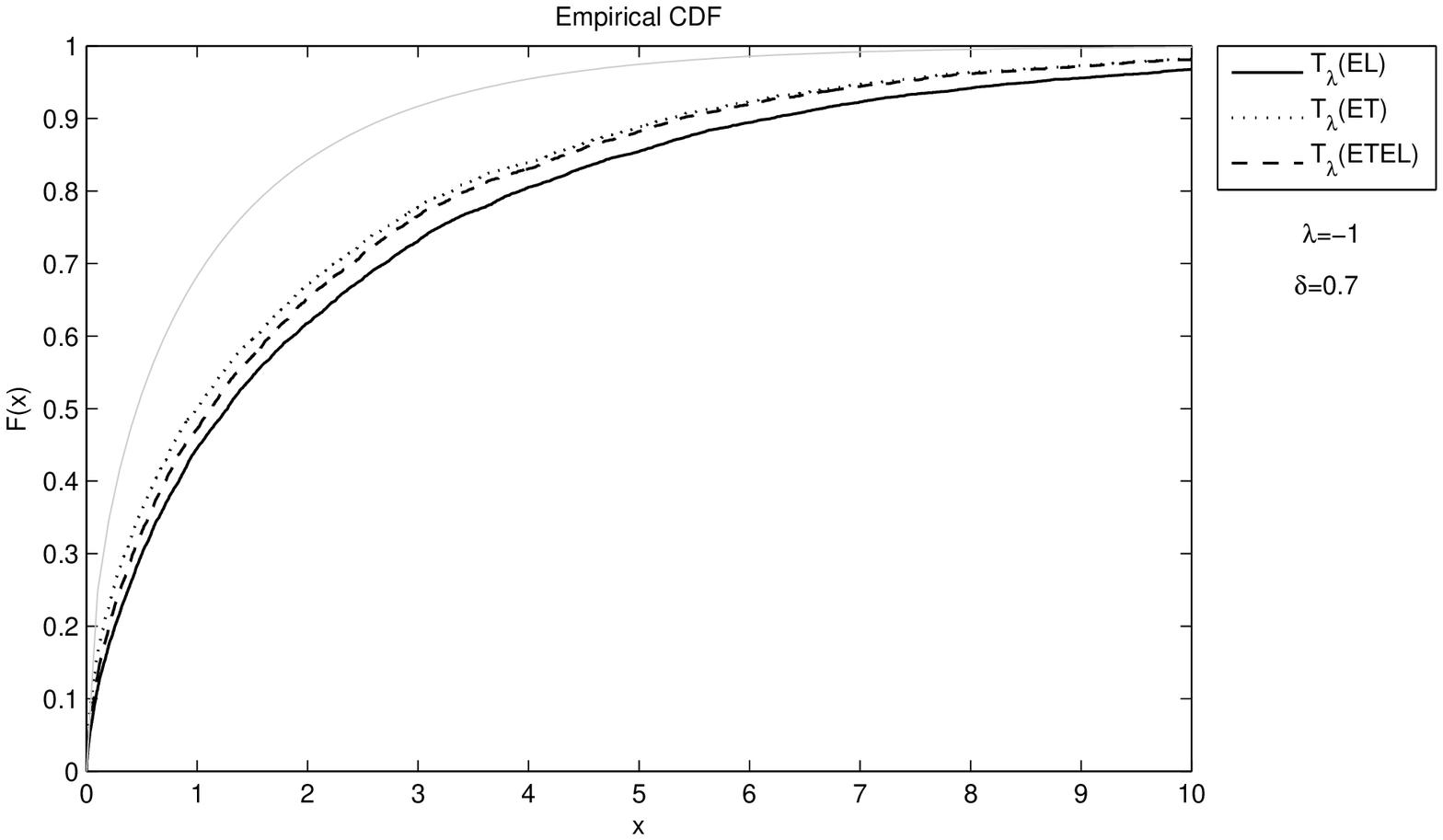}%
}%
\\%
{\includegraphics[
height=2.8764in,
width=5.6273in
]%
{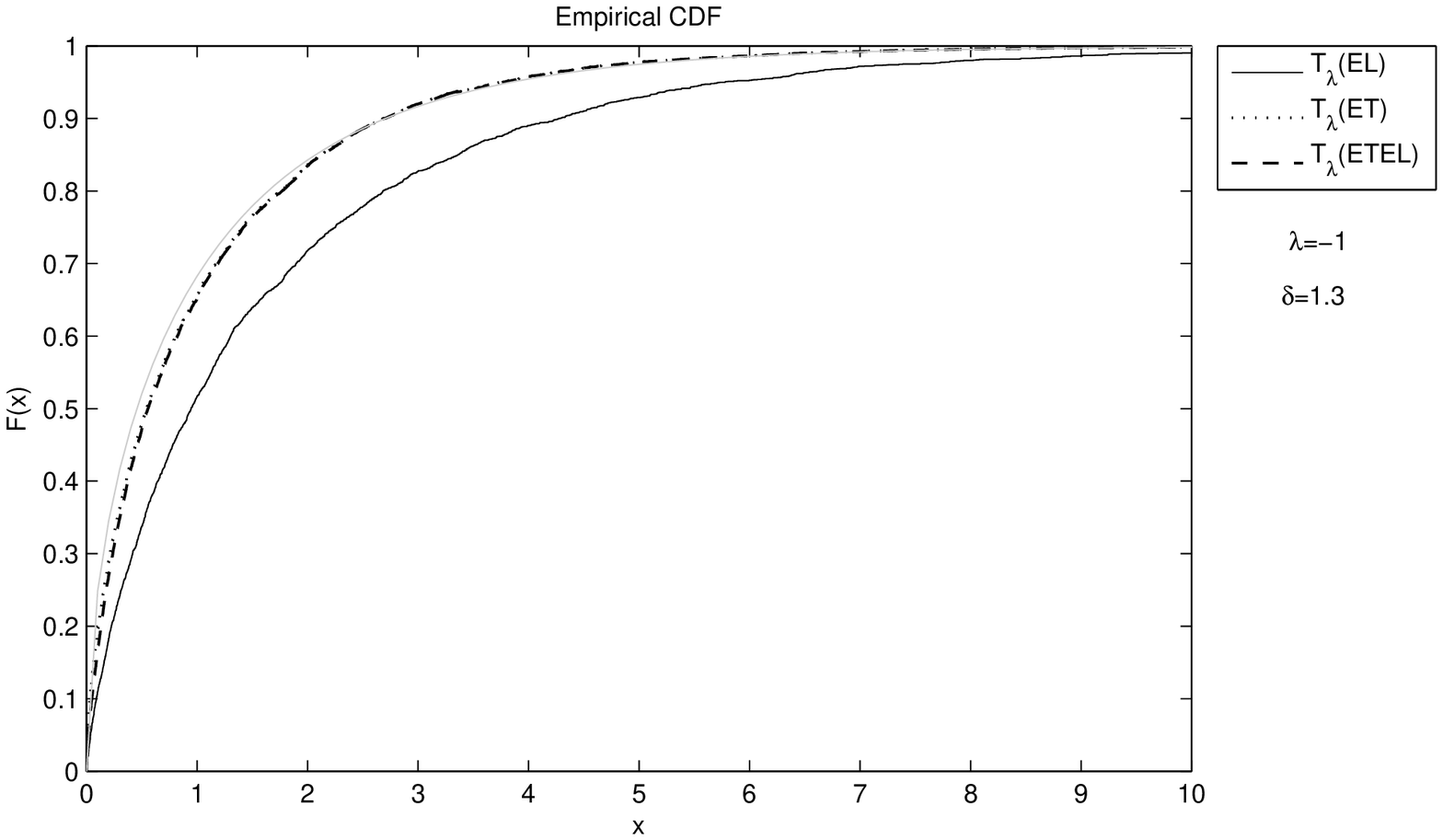}%
}%
\end{tabular}
\caption{Cumulative distribution function of $T_{n}^{\phi_{-1}}$ with the three types of estimators, for $n=1000$, when the model is correctly specilied (top), and is misspecified ($\delta=0.7$, $\delta=1.3$). \label{fig7}}%
\end{figure}%

Figure \ref{fig1} and \ref{fig2} represent, only for $n=100$ for illustrative
purposes, the asymptotic power based on the power-divergence test statistics
$T_{n}^{\phi_{\lambda}}(\widehat{\theta}_{ETEL},0)$ and $S_{n}^{\phi_{\lambda
}}(\widehat{\theta}_{ETEL},0)$, $\beta_{T_{n}^{\phi}}(\theta^{\ast})$ and
$\beta_{S_{n}^{\phi}}(\theta^{\ast})$ when the nominal significance level is
$\alpha=0.05$. There are no substantial differences for a generic small or
moderate samples size. The test-statistics $T_{n}^{\phi_{\lambda}%
}(\widehat{\theta}_{ETEL},0)$ and $S_{n}^{\phi_{\lambda}}(\widehat{\theta
}_{ETEL},0)$, with $\lambda=-1$, exhibit the exact significance levels closest
to the nominal significance level, $0.058$ for $T_{n}^{\phi_{\lambda}%
}(\widehat{\theta}_{ETEL},0)$\ and $0.052$ for $S_{n}^{\phi_{\lambda}%
}(\widehat{\theta}_{ETEL},0)$. In the results obtained in Balakrishnan et al.
(2015) $S_{n}^{\phi_{\lambda}}(\widehat{\theta}_{EL},\theta_{0})$ was found
out to be much more efficient than $T_{n}^{\phi_{\lambda}}(\widehat{\theta
}_{EL},\theta_{0})$ with small sample sizes, for being $S_{n}^{\phi_{\lambda}%
}(\widehat{\theta}_{EL},\theta_{0})$ closer to the nominal level than
$T_{n}^{\phi_{\lambda}}(\widehat{\theta}_{EL},\theta_{0})$. Such a difference
is less pronounced for $T_{n}^{\phi_{\lambda}}(\widehat{\theta}_{ETEL}%
,\theta_{0})$\ and $S_{n}^{\phi_{\lambda}}(\widehat{\theta}_{ETEL},\theta
_{0})$. The performance of $T_{n}^{\phi_{\lambda}}(\widehat{\theta}%
_{EL},\theta_{0})$ with $\lambda=-1$\ is then relatively good in efficiency
with small sample sizes as well as in robustness under misspecification (with
small and big sample sizes).

The approximation to the asymptotic power $\beta_{T_{n}^{\phi}}(\theta^{\ast
})$, at $\theta^{\ast}\neq0$, of the power-divergence test $T_{n}%
^{\phi_{\lambda}}(\widehat{\theta}_{ETEL},0)$ for the correctly specified
model, with a significance level $\alpha$, is according to Remark \ref{beta1}
and doing some algebraic manipulations, $\beta_{T_{n}^{\phi_{\lambda}}}^{\ast
}(\theta^{\ast})=1-\Phi\left(  \nu_{T_{n}^{\phi_{\lambda}}}(\theta^{\ast
},0)\right)  $, where%
\[
\nu_{T_{n}^{\phi_{\lambda}}}(\theta^{\ast},0)=\left(  \frac{1}{n}%
\boldsymbol{s}_{T_{n}^{\phi_{\lambda}}}^{T}(\theta^{\ast},0)\boldsymbol{M}%
_{T_{n}}(\theta^{\ast},0)\boldsymbol{s}_{T_{n}^{\phi_{\lambda}}}(\theta^{\ast
},0)\right)  ^{-\frac{1}{2}}\left(  \frac{\chi_{p,\alpha}^{2}}{2n}-\mu
_{\phi_{\lambda}}(\theta^{\ast},0)\right)  ,
\]%
\[
\mu_{T_{n}^{\phi_{\lambda}}}(\theta^{\ast},0)=\left\{
\begin{array}
[c]{ll}%
\frac{1}{\lambda(\lambda+1)}\left(  \tfrac{\exp\left\{  \frac{\lambda
(\lambda+1)\theta^{\ast2}}{2\left(  1-\lambda\theta^{\ast2}\right)  }\right\}
}{\sqrt{\left(  1-\lambda\theta^{\ast2}\right)  \left(  \theta^{\ast
2}+1\right)  ^{\lambda}}}-1\right)  , & \lambda\in%
\mathbb{R}
-\{0,-1\}\\
\theta^{\ast2}-\frac{1}{2}\log\left(  1+\theta^{\ast2}\right)  & \lambda=0\\
\frac{1}{2}\log\left(  1+\theta^{\ast2}\right)  & \lambda=-1
\end{array}
\right.  ,
\]%
\begin{align*}
&  \boldsymbol{s}_{T_{n}^{\phi_{\lambda}}}^{T}(\theta^{\ast},0)\boldsymbol{M}%
_{T_{n}}(\theta^{\ast},0)\boldsymbol{s}_{T_{n}^{\phi_{\lambda}}}(\theta^{\ast
},0)=\frac{\theta^{\ast2}\exp\left\{  \theta^{\ast2}\left(  \frac
{\lambda(\lambda+1)}{1-\lambda\theta^{\ast2}}+\frac{1}{2\theta^{\ast2}%
+1}\right)  \right\}  }{\sqrt{\left(  2\theta^{\ast2}+1\right)  ^{5}}\left(
1-\lambda\theta^{\ast2}\right)  ^{3}\left(  \theta^{\ast2}+1\right)
^{\lambda-1}}\\
&  \times%
\begin{pmatrix}
1 & (\lambda+2)\theta^{\ast}%
\end{pmatrix}%
\begin{pmatrix}
2\theta^{\ast4}+4\theta^{\ast2}+1 & -\frac{\theta}{2\theta^{2}+1}\left(
\theta^{\ast4}+3\theta^{\ast2}+1\right) \\
-\frac{\theta^{\ast}}{2\theta^{\ast2}+1}\left(  \theta^{\ast4}+3\theta^{\ast
2}+1\right)  & \frac{1}{2(2\theta^{\ast2}+1)^{2}}\left(  6\theta^{\ast
8}+16\theta^{\ast6}+19\theta^{\ast4}+8\theta^{\ast2}+1\right)
\end{pmatrix}%
\begin{pmatrix}
1\\
(\lambda+2)\theta^{\ast}%
\end{pmatrix}
.
\end{align*}
In particular, Figure \ref{fig3} shows the approximated and exact asymptotic
powers of $T_{n}^{\phi_{\lambda}}(\widehat{\theta}_{ETEL},0)$ when
$\lambda=-1$, $\beta_{T_{n}^{\phi_{-1}}}^{\ast}(\theta^{\ast})$ and
$\beta_{T_{n}^{\phi_{-1}}}(\theta^{\ast})$, for two sample sizes $n=100$ and
$n=200$. The approximation is quite good for the values not very close to
$\theta_{0}=0$, $\theta^{\ast}\notin(-0.11,0.11)$ when $n=100$, and
$\theta^{\ast}\notin(-0.075,0.075)$ when $n=200$.%

\begin{figure}[htbp]  \tabcolsep2.8pt  \centering
\begin{tabular}
[c]{c}%
${%
{\includegraphics[
height=3.2638in,
width=6.7196in
]%
{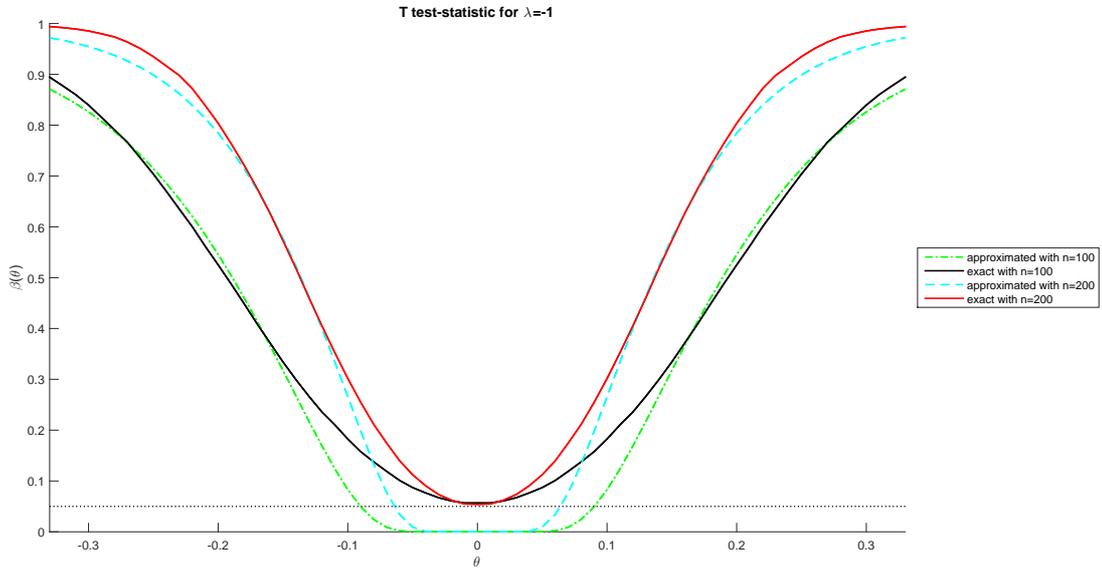}%
}%
}$%
\end{tabular}
\caption{$\beta_{T_{n}^{\phi_{-1}}}^{\ast}(\theta^{\ast})$ and $\beta
_{T_{n}^{\phi_{-1}}}(\theta^{\ast},0)$ when the ETEL estimator is plugged. \label{fig3}}%
\end{figure}%
%

\begin{figure}[htbp]  \tabcolsep2.8pt  \centering
\begin{tabular}
[c]{c}%
${%
{\includegraphics[
height=3.4938in,
width=7.3647in
]%
{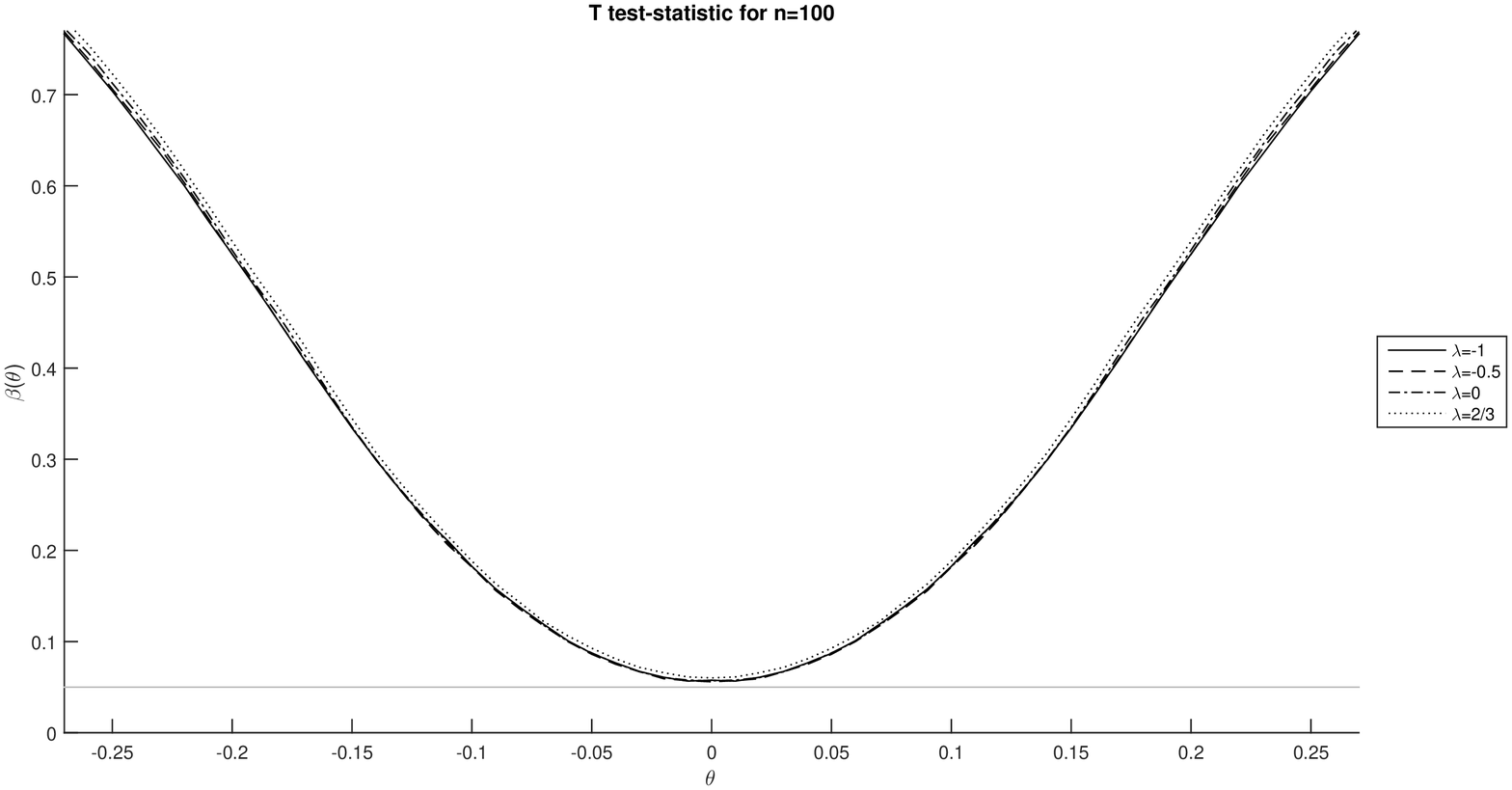}%
}%
}$%
\end{tabular}
\caption{Power function $\beta_{T_{n}^{\phi_{\lambda}}}(\theta^{\ast})$, for different values of $\lambda$ when the ETEL estimator is plugged. \label{fig1}}%
\end{figure}%
%

\begin{figure}[htbp]  \tabcolsep2.8pt  \centering
\begin{tabular}
[c]{c}%
${%
{\includegraphics[
height=3.4938in,
width=7.3647in
]%
{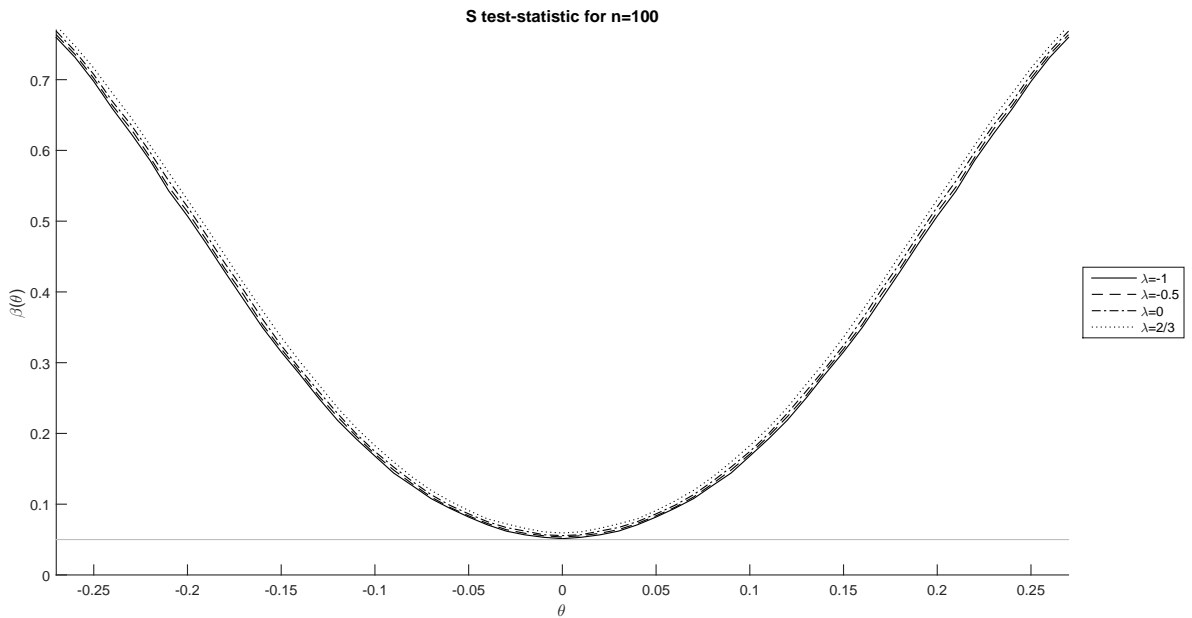}%
}%
}$%
\end{tabular}
\caption{Power function $\beta_{S_{n}^{\phi_{\lambda}}}(\theta^{\ast})$, for different values of $\lambda$  when the ETEL estimator is plugged. \label{fig2}}%
\end{figure}%
\newpage

\section{Conclusion\label{secComp}}

This paper introduces empirical $\phi$-divergence test-statistics using
exponentially tilted empirical likelihood estimators, as alternative to the
empirical likelihood ratio test-statistic. It is shown that these
test-statistics follow the same efficiency and robustness patterns of the
corresponding estimators, empirical likelihood estimators, exponential tilted
estimators and exponentially tilted empirical likelihood estimators. This
justifies the practical choice of the exponentially tilted empirical
likelihood estimator to be plugged into the empirical $\phi$-divergence
test-statistics, for being a good compromise between the efficiency of the
exact size of the test for small or moderate sample sizes and the robustness
under model misspecification. According to the results of the simulation
study, the modified empirical likelihood ratio test
\[
T_{n}^{\phi_{-1}}(\widehat{\theta}_{ETEL},\theta_{0})=2n\left(  \sum
\limits_{i=1}^{n}p_{i,ET}(\theta_{0})\log\left(  np_{i,ET}(\theta_{0})\right)
-\sum\limits_{i=1}^{n}p_{i,ET}(\widehat{\theta}_{ETEL})\log(np_{i,ET}%
(\widehat{\theta}_{ETEL}))\right)  ,
\]
exhibits, by far, the best performance.

A possible future research could include a correction of the critical value
for $T_{n}^{\phi_{-1}}(\widehat{\theta}_{ETEL},\theta_{0})$ test-statistic.
For instance, in the line of Lee (2014), bootstrap critical values of
$T_{n}^{\phi_{-1}}(\widehat{\theta}_{ETEL},\theta_{0})$ could be studied to be
compared with the Wald type test-statistic's bootstrap critical values
proposed in the aforementioned paper.

\appendix\newpage

\noindent{\Large Appendix}\medskip

\begin{proof}
[Proof of Lemma \ref{ThDp}]Taking into account (\ref{pET}),%
\begin{align*}
\frac{\partial}{\partial\boldsymbol{\theta}}p_{ET,i}\left(  \boldsymbol{\theta
}\right)   &  =\frac{\partial}{\partial\boldsymbol{\theta}}\frac{1}{n}%
\frac{\exp\{\boldsymbol{t}_{ET}^{T}(\boldsymbol{\theta})\boldsymbol{g}%
(\boldsymbol{X}_{i},\boldsymbol{\theta})\}}{\overline{\exp}_{ET}%
(\boldsymbol{\theta})}\\
&  =\frac{1}{n}\frac{\exp\{\boldsymbol{t}_{ET}^{T}(\boldsymbol{\theta
})\boldsymbol{g}(\boldsymbol{X}_{i},\boldsymbol{\theta})\}\left[
\frac{\partial}{\partial\boldsymbol{\theta}}\left(  \boldsymbol{t}_{ET}%
^{T}(\boldsymbol{\theta})\boldsymbol{g}(\boldsymbol{X}_{i},\boldsymbol{\theta
})\right)  \overline{\exp}_{ET}(\boldsymbol{\theta})-\frac{\partial}%
{\partial\boldsymbol{\theta}}\overline{\exp}_{ET}(\boldsymbol{\theta})\right]
}{\overline{\exp}_{ET}^{2}(\boldsymbol{\theta})}\\
&  =p_{ET,i}\left(  \boldsymbol{\theta}\right)  \left[  \frac{\partial
}{\partial\boldsymbol{\theta}}\left(  \boldsymbol{t}_{ET}^{T}%
(\boldsymbol{\theta})\boldsymbol{g}(\boldsymbol{X}_{i},\boldsymbol{\theta
})\right)  -\overline{\exp}_{ET}^{-1}(\boldsymbol{\theta})\frac{\partial
}{\partial\boldsymbol{\theta}}\overline{\exp}_{ET}(\boldsymbol{\theta
})\right]  ,
\end{align*}
where%
\begin{align*}
\frac{\partial}{\partial\boldsymbol{\theta}}\left(  \boldsymbol{t}_{ET}%
^{T}(\boldsymbol{\theta})\boldsymbol{g}(\boldsymbol{X}_{i},\boldsymbol{\theta
})\right)   &  =\boldsymbol{G}_{\boldsymbol{X}_{i}}^{T}(\boldsymbol{\theta
})\boldsymbol{t}_{ET}(\boldsymbol{\theta})+\frac{\partial}{\partial
\boldsymbol{\theta}}\boldsymbol{t}_{ET}^{T}(\boldsymbol{\theta})\boldsymbol{g}%
(\boldsymbol{X}_{i},\boldsymbol{\theta})\\
&  =\boldsymbol{G}_{\boldsymbol{X}_{i}}^{T}(\boldsymbol{\theta})\boldsymbol{t}%
_{ET}(\boldsymbol{\theta})-\left(  \frac{1}{n}%
{\textstyle\sum_{i=1}^{n}}
\exp\{\boldsymbol{t}_{ET}(\boldsymbol{\theta})\boldsymbol{g}(\boldsymbol{X}%
_{i},\boldsymbol{\theta})\}\boldsymbol{G}_{\boldsymbol{X}_{i}}^{T}%
(\boldsymbol{\theta})\left(  \boldsymbol{t}_{ET}(\boldsymbol{\theta
})\boldsymbol{g}^{T}(\boldsymbol{X}_{i},\boldsymbol{\theta})+\boldsymbol{I}%
_{r}\right)  \right) \\
&  \times\left(  \frac{1}{n}%
{\textstyle\sum_{i=1}^{n}}
\exp\{\boldsymbol{t}_{ET}(\boldsymbol{\theta})\boldsymbol{g}(\boldsymbol{X}%
_{i},\boldsymbol{\theta})\}\boldsymbol{g}(\boldsymbol{X}_{i}%
,\boldsymbol{\theta})\boldsymbol{g}^{T}(\boldsymbol{X}_{i},\boldsymbol{\theta
})\right)  ^{-1}\boldsymbol{g}(\boldsymbol{X}_{i},\boldsymbol{\theta}),
\end{align*}%
\begin{align*}
\frac{\partial}{\partial\boldsymbol{\theta}}\boldsymbol{t}_{ET}^{T}%
(\boldsymbol{\theta})  &  =-\left(
{\textstyle\sum_{i=1}^{n}}
p_{ET,i}\left(  \boldsymbol{\theta}\right)  \boldsymbol{G}_{\boldsymbol{X}%
_{i}}^{T}(\boldsymbol{\theta})\left(  \boldsymbol{t}_{ET}(\boldsymbol{\theta
})\boldsymbol{g}^{T}(\boldsymbol{X}_{i},\boldsymbol{\theta})+\boldsymbol{I}%
_{r}\right)  \right) \\
&  \times\left(
{\textstyle\sum_{i=1}^{n}}
p_{ET,i}\left(  \boldsymbol{\theta}\right)  \boldsymbol{g}(\boldsymbol{X}%
_{i},\boldsymbol{\theta})\boldsymbol{g}^{T}(\boldsymbol{X}_{i}%
,\boldsymbol{\theta})\right)  ^{-1}\\
&  =-\left(  \frac{1}{n}%
{\textstyle\sum_{i=1}^{n}}
\exp\{\boldsymbol{t}_{ET}^{T}(\boldsymbol{\theta})\boldsymbol{g}%
(\boldsymbol{X}_{i},\boldsymbol{\theta})\}\boldsymbol{G}_{\boldsymbol{X}_{i}%
}^{T}(\boldsymbol{\theta})\left(  \boldsymbol{t}_{ET}(\boldsymbol{\theta
})\boldsymbol{g}^{T}(\boldsymbol{X}_{i},\boldsymbol{\theta})+\boldsymbol{I}%
_{r}\right)  \right) \\
&  \times\left(  \frac{1}{n}%
{\textstyle\sum_{i=1}^{n}}
\exp\{\boldsymbol{t}_{ET}^{T}(\boldsymbol{\theta})\boldsymbol{g}%
(\boldsymbol{X}_{i},\boldsymbol{\theta})\}\boldsymbol{g}(\boldsymbol{X}%
_{i},\boldsymbol{\theta})\boldsymbol{g}^{T}(\boldsymbol{X}_{i}%
,\boldsymbol{\theta})\right)  ^{-1}%
\end{align*}
according to (41) of Schennach (2007), and%
\begin{align}
&  \frac{\partial}{\partial\boldsymbol{\theta}}\overline{\exp}_{ET}%
(\boldsymbol{\theta})=\frac{1}{n}\sum_{i=1}^{n}\exp\{\boldsymbol{t}_{ET}%
^{T}(\boldsymbol{\theta})\boldsymbol{g}(\boldsymbol{X}_{i},\boldsymbol{\theta
})\}\frac{\partial}{\partial\boldsymbol{\theta}}\left(  \boldsymbol{t}%
_{ET}^{T}(\boldsymbol{\theta})\boldsymbol{g}(\boldsymbol{X}_{i}%
,\boldsymbol{\theta})\right) \nonumber\\
&  =\left(  \frac{1}{n}\sum_{i=1}^{n}\exp\{\boldsymbol{t}_{ET}^{T}%
(\boldsymbol{\theta})\boldsymbol{g}(\boldsymbol{X}_{i},\boldsymbol{\theta
})\}\boldsymbol{G}_{\boldsymbol{X}_{i}}^{T}(\boldsymbol{\theta})\right)
\boldsymbol{t}_{ET}(\boldsymbol{\theta})\nonumber\\
&  -\left(  \frac{1}{n}%
{\textstyle\sum_{i=1}^{n}}
\exp\{\boldsymbol{t}_{ET}^{T}(\boldsymbol{\theta})\boldsymbol{g}%
(\boldsymbol{X}_{i},\boldsymbol{\theta})\}\boldsymbol{G}_{\boldsymbol{X}_{i}%
}^{T}(\boldsymbol{\theta})\left(  \boldsymbol{t}_{ET}(\boldsymbol{\theta
})\boldsymbol{g}^{T}(\boldsymbol{X}_{i},\boldsymbol{\theta})+\boldsymbol{I}%
_{r}\right)  \right) \nonumber\\
&  \times\left(  \frac{1}{n}%
{\textstyle\sum_{i=1}^{n}}
\exp\{\boldsymbol{t}_{ET}^{T}(\boldsymbol{\theta})\boldsymbol{g}%
(\boldsymbol{X}_{i},\boldsymbol{\theta})\}\boldsymbol{g}(\boldsymbol{X}%
_{i},\boldsymbol{\theta})\boldsymbol{g}^{T}(\boldsymbol{X}_{i}%
,\boldsymbol{\theta})\right)  ^{-1}\left(  \frac{1}{n}\sum_{i=1}^{n}%
\exp\{\boldsymbol{t}_{ET}^{T}(\boldsymbol{\theta})\boldsymbol{g}%
(\boldsymbol{X}_{i},\boldsymbol{\theta})\}\boldsymbol{g}(\boldsymbol{X}%
_{i},\boldsymbol{\theta})\right) \nonumber\\
&  =\overline{\exp_{ET}\boldsymbol{G}^{T}}(\boldsymbol{\theta})\boldsymbol{t}%
_{ET}(\boldsymbol{\theta})-\widehat{\boldsymbol{K}}(\boldsymbol{\theta
})\overline{\exp_{ET}\boldsymbol{g}}(\boldsymbol{\theta})\nonumber\\
&  =\overline{\exp_{ET}\boldsymbol{G}^{T}}(\boldsymbol{\theta})\boldsymbol{t}%
_{ET}(\boldsymbol{\theta}), \label{two}%
\end{align}
with%
\[
\overline{\exp_{ET}\boldsymbol{g}}(\boldsymbol{\theta})=\frac{1}{n}\sum
_{i=1}^{n}\exp\{\boldsymbol{t}_{ET}^{T}(\boldsymbol{\theta})\boldsymbol{g}%
(\boldsymbol{X}_{i},\boldsymbol{\theta})\}\boldsymbol{g}(\boldsymbol{X}%
_{i},\boldsymbol{\theta})=\boldsymbol{0}_{r}%
\]
from (\ref{eET}). Using the previous notation%
\begin{equation}
\frac{\partial}{\partial\boldsymbol{\theta}}\left(  \boldsymbol{t}_{ET}%
^{T}(\boldsymbol{\theta})\boldsymbol{g}(\boldsymbol{X}_{i},\boldsymbol{\theta
})\right)  =\boldsymbol{G}_{\boldsymbol{X}_{i}}^{T}(\boldsymbol{\theta
})\boldsymbol{t}_{ET}(\boldsymbol{\theta})-\widehat{\boldsymbol{K}%
}(\boldsymbol{\theta})\boldsymbol{g}(\boldsymbol{X}_{i},\boldsymbol{\theta}),
\label{one}%
\end{equation}
and replacing (\ref{one}) and (\ref{two}) in the expression of $\frac
{\partial}{\partial\boldsymbol{\theta}}p_{ET,i}\left(  \boldsymbol{\theta
}\right)  $, the desired result is obtained.
\end{proof}

\begin{proof}
[Proof of Lemma \ref{ThDD1}]Taking into account the expression of (\ref{F1}),%
\[
\frac{\partial}{\partial\boldsymbol{\theta}}D_{\phi}\left(  \boldsymbol{u}%
,\boldsymbol{p}_{ET}\left(  \boldsymbol{\theta}\right)  \right)  =\frac{1}%
{n}\sum_{i=1}^{n}\frac{\partial}{\partial\boldsymbol{\theta}}\left(
np_{ET,i}\left(  \boldsymbol{\theta}\right)  \right)  \psi\left(  \frac
{1}{np_{ET,i}\left(  \boldsymbol{\theta}\right)  }\right)  .
\]
By plugging $\frac{\partial}{\partial\boldsymbol{\theta}}p_{ET,i}\left(
\boldsymbol{\theta}\right)  $ from Theorem \ref{ThDp} into the previous
expression, (\ref{DD1}) is obtained. Since according to the weak law of large
numbers $\frac{1}{n}%
{\textstyle\sum\nolimits_{i=1}^{n}}
h(\boldsymbol{X}_{i})\underset{n\rightarrow\infty}{\overset{P}{\longrightarrow
}}$\textrm{$E$}$[h(\boldsymbol{X})]$ for any integrable function $h:%
\mathbb{R}
^{p}\longrightarrow%
\mathbb{R}
$, taking the approppriate functions in the role of $h$, the limiting value of
$\frac{\partial}{\partial\boldsymbol{\theta}}D_{\phi}\left(  \boldsymbol{u}%
,\boldsymbol{p}_{ET}\left(  \boldsymbol{\theta}\right)  \right)  $ is obtained.
\end{proof}

\end{document}